\documentclass[12pt]{iopart}

\usepackage[utf8]{inputenc}
\usepackage{braket}
\usepackage{cite}
\usepackage{bm}
\usepackage{hyperref}
\usepackage{fancyhdr}
\usepackage{verbatim}

\usepackage{color}
\usepackage[dvipsnames]{xcolor}

\usepackage{amsthm}
\usepackage{amsfonts}
\newtheorem{theorem}{Theorem}
\newtheorem{assumption}{Assumption}
\newtheorem{lemma}{Lemma}

\newtheorem{definition}{Definition}
\newtheorem{proposition}{Proposition}
\newtheorem*{nono-theorem}{Theorem}
\newcommand{\E}{\mathcal{E}_{N}}
\newcommand{\linn}{\mathcal{L}_{N}}
\newcommand{\lin}{\mathcal{L}}

\usepackage{graphicx}

\begin{document}

\begin{center}{\large
\textbf{Quantum fluctuation dynamics of open quantum systems with collective operator-valued rates, and applications to Hopfield-like networks}  }
\end{center}

\begin{center}
Eliana Fiorelli
\end{center}

\begin{center}
Instituto de F\'isica Interdisciplinar y Sistemas Complejos (IFISC), UIB–CSIC, UIB Campus, Palma de Mallorca, 07122, Spain
\\

* eliana@ifisc.uib-csic.es 
\end{center}

\begin{center}
\today
\end{center}


\section*{Abstract}

We consider a class of open quantum many-body systems that evolves in a Markovian fashion, the dynamical generator being in GKS-Lindblad form. Here, the Hamiltonian contribution is characterized by an all-to-all coupling, and the dissipation features local transitions that depend on collective, operator-valued rates, encoding average properties of the system. These types of generators can be formally obtained by generalizing, to the quantum realm, classical (mean-field) stochastic Markov dynamics, with state-dependent transitions. Focusing on the dynamics emerging in the limit of infinitely large systems, we build on the exactness of the mean-field equations for the dynamics of average operators. In this framework, we derive the dynamics of quantum fluctuation operators, that can be used in turn to understand the fate of quantum correlations in the system. We apply our results to quantum generalized Hopfield associative memories, showing that, asymptotically and at the mesoscopic scale only a very weak amount of quantum correlations, in the form of quantum discord, emerges beyond classical correlations.

\newpage

\section{Introduction}

Open quantum many-body systems currently represent an extremely active research field. From a fundamental perspective, their importance stems from noticing that any quantum system is never truly isolated, hence the relevance of modelling the effect of a quantum environment \cite{BreuerP:2002}. The many literature contributions to this field focus, on the one hand, on unveiling the features that open quantum many-body dynamics can exhibit. On the other hand, on the possibility of tackling the openness as a resource, that is, on engineering dissipative processes so as to realize interesting stationary quantum many-body states \cite{Diehl08, Verstraete:NatPhys:2009, Weimer:Nat:10}. For instance, the competition between the coherent term and the dissipative processes in markovian open quantum systems can give rise to interesting nonequilibrium stationary or dynamical phases \cite{diehl2010,dallatorre2010,Schindler2013,tauber2014,marcuzzi2016,minganti2018,iemini2018,carollo2019,chertkov2022}, and to nonequilibrium critical dynamics \cite{sieberer2013,helmrich2020,jo2021,jo2022, KellyRM21}.

Notoriously, these type of systems are generically hard to treat, both from a rigorous analytical viewpoint, as well as from a numerical perspective \cite{weimer2021}. To get around the impractical description at the microscopic scale, one typically addresses the problem by looking at the \textit{collective} behavior of the system. To characterize the latter, standard observables are the so-called system-average ones (e.g., the total magnetization in a many-body spin system) \cite{LandfordR69, BenattiEtAl18, SewellS88, Strocchi85}. Here, a first analytical comprehension can be obtained when employing a mean-field approximation -- roughly speaking, the omission of correlations amongst these observables. In relation to this, a number of research efforts deal with proving the validity of such a treatment, in the thermodynamic limit, and for given class of (open) systems \cite{hepp1973, alicki1983, BenedikterPS15, MerkliR18, Porta16, Pickl11, hioe1973, CarolloL:PRL:21}. It is important to highlight that a common characteristic of system-average operators is that they scale as $1/N$ and give rise, in the limit of infinite systems, to a classical algebra of commuting observables. For this reason, they are linked to a \textit{macroscopic} scale description \cite{BenattiCF_ADP_15}. Different types of collective observables are taken in consideration when one aims at unveiling possible quantum footprints left beyond the macroscopic scale. These are referred to as quantum fluctuations, to highlight their conceptual connection with fluctuations associated to random variables in statistical mechanics. Quantum fluctuations, at variance with mean-field operators, scale as $1/\sqrt{N}$, and can be associated, in the limit of infinite systems, to a bosonic algebra of operators \cite{Verbeure10, Goderis89, Goderis90s}. In this sense, they retain a quantum character, and describe the so-called \textit{mesoscopic} scale \cite{Benatti2016, BenattiEtAl_JPA_17Qflu,  BenattiEtAl18, BenattiCF_ADP_15, NarnhoferT_PRA_02}.

Making use of the collective description given by system-average operators, as well as by quantum fluctuations, in this work we focus on many-body, open quantum systems in the Markovian regime. Their evolution can then be described in terms of a quantum master equation, given in the Gorini-Kossakowski-Sudarshan and Lindblad, GKS-Lindblad form \cite{Lindblad76, GoriniKS76, BreuerP:2002}. Furthermore, we use the fact that, in this formalism, one can embed a classical stochastic dynamics as a dissipative contribution \cite{Garrahan18, CarolloGK:JSP:21}, and then introduce additional quantum terms, so as to subsequently analyze the impact of the quantum effects on paradigmatic classical models \cite{marcuzzi2016,jo2021,chertkov2022}. More specifically, the open quantum dynamics is assumed to be given in terms of an all-to-all coupling Hamiltonian, and by dissipative (stochastic) single-body transitions, dependent on system-average operators thorugh \textit{collective} operator-valued functions. In relation to this class of systems, the validity of the mean-field theory for system-average operators has been demonstrated in Ref. \cite{FiorelliEtALC_NJP_23}. Building on this result, we derive the effective map that implements the time-evolution of quantum fluctuation operators, in the thermodynamic limit. To give an example of the use of our result, we apply it to quantum generalizations of Hopfield-like associative-memory dynamics \cite{Hopfield:1982}. These  are recently receiving attention \cite{LabayMoraZG23, BoedekerFM23, Rotondo:JPA:2018} also due to the investigation of their possible realization in multi-modal cavity setups \cite{MarshEtAl:PhysRevX:21, GuoEtAl:PRL:2019, VaidyaEtAl:PRX:2018, Rotondo:PRL:2015}. More in general, our results are of interest for the physical scenarios where a collective representation of the master equation can be employed, ranging from quantum-optical settings to superradiant atomic ensembles \cite{LuoEtAl:arxiv:23, SeetharamEtAl:PRR:2022, KellyRM21, MarshEtAl:PhysRevX:21, GuoEtAl:PRL:2019, VaidyaEtAl:PRX:2018, NorciaEtAl:Science:2018, TorgglerKR17}.

Our paper is organized as follows: Section \ref{sec3} introduces the system-average operators, the quantum fluctuation ones, and the mathematical tools we use to deal with them. In Section \ref{sec_generators} the form of the dynamical generator is specified, and an example of the type of employed dynamics is discussed in \ref{motivation}. The dynamics of the system-averaged operators is contained in Section \ref{sec4}. As previously anticipated, these results have been already presented by Ref.\cite{FiorelliEtALC_NJP_23}, and we report them here for the sake of a global understanding. The original contribution of this paper consists in deriving the effective map that evolves quantum fluctuation operators (Theorem \ref{theorem-mesoscopicdynamics}), and it is contained in Section \ref{dyn-mesoscopic}, while part of the related Lemmas are given in \ref{app:B}. Finally, as an example of application of our results, Section \ref{sec5} considers open quantum generalized Hopfield networks, focusing on the behavior of system-average operators as well as on quantum correlations.

\section{Model for many-body quantum systems}
\label{sec3}
In this section we deal with the mathematical tools employed for describing many body quantum systems: the algebra of operators and a functional representation of the quantum states \cite{BratteliR82}. We focus on two types of collective, many-body observables that are of interest for this manuscript. First, we introduce the so-called average operators, which can be regarded as sample-mean averages of a given single-particle operator \cite{BenattiEtAl18,Verbeure10} over the total system. Subsequently, we present the fluctuation operators, that can be considered as the analogous of fluctuations for stochastic, classical variables, here however defined with respect to average operators. We further discuss the properties of the two types of operators when considering states with sufficiently short-ranged correlations (in the sense defined in \ref{clustering}), called clustering states \cite{LandfordR69,Strocchi05,thirring2013quantum}.   

\subsection{Quasi-local algebra operators and states}
We consider a many-body quantum system $S$, composed of a (countably) infinite number of identical (distinguishable) particles. Here, each particle is assumed to have a finite-dimensional Hilbert space, with dimension $d<\infty$, and its physical properties can be thus described by the algebra of $ d \times d $ complex matrices, $M_d(\mathbb{C})$. To move the description from the single particle, to the whole many-body system, notice first that, being distinguishable, each particle can be identified by an integer index $k\in\mathbb{N}$. Given any single-particle operator $x^{(k)}$, with $x\in M_d(\mathbb{C})$, that acts non-trivially only on the $k$th particle, then an operator acting on the many-body system is obtained as
$$
x^{(k)}={\bf 1}_d\otimes {\bf 1}_d\otimes \dots \otimes x\otimes {\bf 1}_d\otimes {\bf 1}_d \otimes \dots ,
$$
where ${\bf 1}_d$ is the identity in $M_{d}(\mathbb{C})$ and $x$ takes $k$th entry of the tensor product. More in general, all the observables of the many-body system are contained in the {\it quasi-local} $C^*$-algebra $\mathcal{A}$, obtained as the norm closure \footnote{We consider the operator norm, denoted as $\|\cdot\|$, given by the largest eigenvalue, in modulus, of the operator.} of the union of all possible local sub-algebras (e.g., $M_{d}(\mathbb{C})$ or $\otimes_{k} M_{d}^{(k)}(\mathbb{C})$, with $k$ extending over a finite number of particles) of the system \cite{BratteliR82}. Practically speaking, the quasi-local algebra $\mathcal{A}$ contains all the operators that are supported on a finite number of particles, which are called strictly local operators, as well as the operators that are extended over the whole system, and yet can be obtained as the limit of a converging sequence of local operators. These are referred to as quasi-localised operators. 

Given the algebra $\mathcal{A}$, states of the quantum system are linear, positive and normalized functionals, $\omega$, on the quasi-local algebra \cite{BratteliR82}, which associate to each operator $A\in\mathcal{A}$ a complex number $\langle A\rangle$ embodying the expectation of the operator itself, 
$$\mathcal{A}\ni A\mapsto \omega(A)=\langle A\rangle \, ,$$
and $\omega({\bf 1})=1$ where ${\bf 1}$ is the identity of $\mathcal{A}$. That is, information about the state of a physical system is equivalent to the knowledge of all possible expectation values for its operators. {The state is called translation invariant when the expectation values of single-particle operators do not depend on the considered particle, $\omega(x^{(k)})=\langle x\rangle$, $\forall k\in \mathbb{N}$, for $x\in M_d(\mathbb{C})$, [see also Definition \ref{clustering}]. }

\subsection{System-average operators}
The operators belonging to the quasi-local algebra $\mathcal{A}$ contribute to the microscopic description of the system. However, when considering many-body systems, one is often interested in unveiling the behavior of collective operators, so as to exemplify average properties of the system. This is typically the situation when dealing with, e.g., equilibrium as well as nonequilibrium phase transitions, where collective quantities serving as order-parameters play such a central role. 

We thus focus on operators of the form
\begin{equation}
X_{N} \equiv \frac{1}{N} \sum_{k=1}^{N} x^{(k)}, \qquad \mbox{with } \qquad x\in M_d(\mathbb{C})\, ,
\label{eq:average-operators}
\end{equation}
that represent sample-mean averages of a given single-particle operator, and can be linked to the random variables in the law of large numbers \cite{grimmett2020probability}. It is clear that $\forall N <\infty$, $X_N \in \mathcal{A}$, as the number of particles involved in the summation is finite. Differently, the limiting point $X_\infty$ of the sequence $X_N$ is not an element of the quasi-local algebra $\mathcal{A}$, as the sequence $X_N$ does not converge in the norm topology \cite{BratteliR82}. 

In the thermodynamic limit, average operators give rise to an emergent classical algebra: the norm of the commutator between any two of them, $[X_N,Y_N]$, is bounded by $2\|x\| \|y\|/N$, so their commutator goes to zero in the large $N$ limit \cite{LandfordR69,BratteliR82,BenattiEtAl18}. However, for investigating the structure of the operators themselves, a weaker forms of convergence needs to be introduced. This is referred to as {\it weak operator topology} \cite{Strocchi05}, and it is defined as follows: 
\begin{definition}\label{w-lim}
A sequence of operators $C_n$ converges weakly to the operator $C$,
\begin{equation*}
C=(\mathrm{w\mbox{--}})\!\lim_{n\to\infty} C_n, \quad if \quad
\lim_{n\to\infty} \omega( A^\dagger C_n B)=\omega(A^\dagger  C B)\, \qquad \forall A,B\in \mathcal{A}\, . 
\end{equation*}   
\end{definition}
According to the above definition \footnote{This form of convergence coincides with the  weak operator convergence within the so-called GNS representation of the algebra $\mathcal{A}$ induced by the state $\omega$ \cite{BratteliR82,Strocchi05}} the information on the limiting operator $C$ are characterized by controlling the expectation values [associated to the state $\omega(\cdot)$] of all of its possible correlation functions with any quasi-local operator. 

For \textit{clustering} quantum states, that are states with sufficiently short-ranged correlations, the limiting operators $X_\infty$ of the sequences $X_N$ are multiples of the identity \cite{LandfordR69,BratteliR82,Strocchi05,Verbeure10} in the sense that
$$X_\infty={\rm (w\mbox{--}})\!\lim_{N\to\infty}X_N=\omega(x) \, ,$$
having further assumed translation invariance of the state, and omitted the identity multiplying the right-hand side of the above limit, for the sake of a lighter notation. More specifically, we define clustering states through the following:

\begin{definition}\label{clustering}
We refer to quantum states $\omega$ of the quasi-local algebra $\mathcal{A}$ as translation-invariant clustering states if the following properties are satisfied: 
\begin{eqnarray}
    & \mbox{i)} \quad \omega(x^{(k)})=\omega(x^{(h)})=\langle x\rangle, \qquad &\forall x\in M_d(\mathbb{C}), \forall k,h\in \mathbb{N}\, ; \\
    & \mbox{ii)}\lim_{N\to\infty}\omega([X_N-\langle x\rangle]^2)=0, \qquad &\forall x=x^\dagger \in M_d(\mathbb{C})\, . \label{def_clust} 
\end{eqnarray}
\end{definition}
From $ii)$ follows that states for which average operators $X_N$ display vanishing variance in the large $N$ limit, feature the limiting point $X_\infty$ converging to multiples of the identity. It is indeed possible to show that Eq.~\eref{def_clust} implies the weak convergence of $X_N$ to $X_\infty=\langle x\rangle$, as defined by Def.~\eref{w-lim}. 

\subsection{Quantum fluctuation operators}

We have seen that average operators generate a commuting algebra, in the thermodynamic limit. As such, they provide a collective and classical description of the many-body quantum system. For the case in which one is interested to unveil (possible) quantum behavior of collective observables, the right quantities to take into account are combinations of microscopic operators of the form 
\begin{equation}\label{e:qFlu_general}
    F^{N}(x) = \frac{1}{\sqrt{N}}\sum_{k=1}^{N}\left(x^{(k)}-\omega(x)\right),
\end{equation}
for which $\omega(F^{N}(x)) = 0$. The operators above defined are referred to as local quantum fluctuations, as they provide a description of the `deviation' from the average operators, and represent an analog of the fluctuations defined for classical, stochastic variables \cite{Goderis89, Goderis90s}. Moreover, the presence of quantum features in these collective operators can be understood by considering commutators of quantum fluctuations, reading
$$[F^N(x), F^N(y)] = Z_N \, ,$$
with $Z_N$ an average observable. As a result, commutators of quantum fluctuations tend, in the weak operator topology [see Def.~\eref{w-lim}] to multiples of the identity, ${\rm (w\mbox{--}})\lim_{N\rightarrow \infty} Z_N = \omega(Z)$. This fact suggests that in the large $N$ limit the set of quantum fluctuation operators converge to a bosonic algebra. In the following, we summarize the main steps that show how to construct such an algebra (for a more exhaustive treatment see, e.g., Ref. \cite{BratteliR82,Verbeure10, Goderis89, Goderis90s, BenattiEtAl_JPA_17Qflu, BenattiEtAl18}. Before going ahead we also remark that, due to the scaling $\frac{1}{\sqrt{N}}$, the operators \eref{e:qFlu_general} provide a description at an intermediate level between the microscopic scale, given by strictly local operators, and the macroscopic degrees of freedom such as the average operators. Indeed, quantum fluctuations are collective operators, with a non-commutative character being the footprint of the microscopic level they emerge from. For this reason, they are said to identify the \textit{mesoscopic} level \cite{NarnhoferT_PRA_02, Benatti2016, BenattiEtAl18}.  

In order to construct the algebra of the quantum fluctuations, let us consider the set of single particle operators $\chi = \lbrace x_{i} \rbrace_{i=1}^{p} \in M_d(\mathbb{C})$, and the corresponding fluctuation operators $F^{N}(x_i)$:
\begin{definition}\label{normal_fluctuations}
The set $(\omega, \chi)$, with $\chi$ translational invariant, and $\omega$ clustering state [in the sense of Def.~\eref{clustering}], has normal fluctuation if, $\forall x_i, x_j \in \chi$,
$$\sum_{l\in\mathbb{Z}} |\omega(x_i^{(0)}x_j^{(l)})-\omega(x_i^{(0)})\omega(x_j^{(l)})| < +\infty, $$
i.e. $\omega$ is $L_1$ clustering, and
$$\lim_{N \rightarrow \infty} \omega\left([F^{N}(x_i)]^{2}\right) \equiv \Sigma_{ii}^{\omega} \qquad \lim_{N \rightarrow \infty} \omega( e^{i \alpha F^N(x_i)}) = e^{-\frac{\alpha^{2}}{2}\Sigma_{ii}^{\omega}} \, $$
\end{definition}
In other words, the set is specified by a characteristic function $ \omega( e^{i \alpha F^N(x_i)})$  converging to a gaussian function, the latter with covariance
$$\Sigma_{ij}^{\omega} = \frac{1}{2} \lim_{N \rightarrow \infty}\omega\left(\lbrace F^{N}(x_i), F^{N}(x_j)\rbrace \right).$$
Within the normal quantum fluctuations, and considering the linear, real span $\bm{\chi} = \lbrace x_{\vec{r}} = \sum_{i=1}^{p} r_i x_i, x_i \in \chi , (r_{1},...,r_{p}) \in \mathbb{R}^p \rbrace$, the following bilinear, positive and symmetric map is well-defined,
$$
(x_{\vec{r}_1}, x_{\vec{r}_2}) \rightarrow (x_{\vec{r}_1}, \Sigma^{\omega }x_{\vec{r}_2}) = \sum_{ij} r_{1i} r_{2j} \Sigma^{\omega }_{ij} \, ,
$$
as well as the symplectic bilinear form $\sigma^{\omega}$, that reads
$$\sigma^{\omega}_{kl}  = - i \lim_{N \rightarrow \infty} \omega( [F^N(x_k), F^N(x_l)]) \, , \qquad \sigma^{\omega}_{kl}= -\sigma^{\omega}_{lk} \, .$$
Notice that, with the above definition, the correlation matrix,
$$C_{ij}^{\omega} \equiv \lim_{N\rightarrow \infty} \omega (F^{N}(x_i)F^{N}(x_j)) \, ,$$
is well defined and it is related to the covariance matrix and the symplectic matrix by the following relation
$$C^{\omega} = \Sigma^{\omega} + \frac{i}{2} \sigma^{\omega}. $$ 

By means of the symplectic form  $\sigma^{\omega}$, one can construct the abstract Weyl algebra $\mathcal{W}(\chi, \sigma^{\omega})$, linearly generated by generic Weyl operators $W(\vec{r}) = W(\sum_{j=1}^{p} r_j x_j)$, $\vec{r} = (r_{1},...,r_{p}) \in \mathbb{R}^p$, obeying the following conditions
\begin{equation}\label{Weyl algebra}
\eqalign{
 i) \, & W(\vec{r})^{\dagger} =  W(-\vec{r}) \, , \\
ii) \, &  W(\vec{r}_1)W(\vec{r}_2) = W(\vec{r}_1+\vec{r}_2)e^{-\frac{i}{2} \vec{r}_1 \cdot(\sigma^{\omega} \vec{r}_2)} \, . 
}
\end{equation}
Upon introducing the notation $\vec{r} \cdot \vec{F}^{N} \equiv \sum_{i=1}^{p} r_i F^N(x_i) $, it can be now understood in which manner Weyl-like operators, $W^N(\vec{r}) \equiv e^{i \vec{r} \cdot \vec{F}^{N}}$ yield, in the large $N$ limit, Weyl operator $W(\vec{r})$. Indeed, it can be proven a theorem (see e.g. \cite{Verbeure10}) stating that any system $(\omega, \bm{\chi})$ with normal fluctuations admits a regular gaussian state $\Omega$ on $\mathcal{W}(\chi, \sigma^{\omega})$ such that $$\lim_{N \rightarrow \infty } \omega \left( W^N(\vec{r}_1)W^N(\vec{r}_2 \right)\cdots W^N(\vec{r}_n)) = \Omega(W(\vec{r}_1)W(\vec{r}_2)\cdots W(\vec{r}_n)) \, ,$$
    where $W^N(\vec{r}_j)$ obeys Eq.~\eref{Weyl algebra}, and 
    \begin{equation}\label{gaussianstateWeyl}
     \Omega(W(\vec{r})) = e^{-\frac{\vec{r}\cdot (\Sigma^{\omega} \vec{r})}{2}} \, , \qquad \forall \vec{r} \in \mathbb{R}^{p} \, . 
    \end{equation}
Notice that \eref{normal_fluctuations} guarantees the regular character of the Gaussian state $\Omega$. This in turn allows one to write a regular representation of Weyl operators acting on Hilbert spaces, $W(\vec{r}) = e^{i \vec{r} \cdot \vec{F}}$. Here, $\vec{F}$ is a $p$-dimensional vector, with components $F(x_i)$ that are obtained from the local quantum fluctuations through the so called \textit{mesoscopic limit}: 
\begin{definition}\label{def:mesoscopicLimit}
A sequence of operator $X_N \in \bf{\mathcal{A}}$ converges to the operator $X \in \mathcal{W}(\chi, \sigma^{\omega})$, 
$X \equiv (m)-\lim_{N \rightarrow \infty }X_{N}$ if and only if $$\Omega_{\vec{r}_1 \vec{r}_2}(X) = \lim_{N \rightarrow \infty} \omega_{\vec{r}_1 \vec{r}_2}(X_N) \, ,$$ where we have employed the notation $\omega_{\vec{r}_2 \vec{r}_2}(X_N) = \omega\left(W^N(\vec{r}_1) X_N W^N(\vec{r}_2)\right)$, and $\Omega_{\vec{r}_2 \vec{r}_2}(X) = \Omega\left(W(\vec{r}_1) X W(\vec{r}_2)\right)$.
\end{definition}
As a consequence of the above definition, by varying $\vec{r}_1, \vec{r}_2 \in \mathbb{R}^{p}$, the expectation values $\Omega_{\vec{r}_1, \vec{r}_2}(x)$ allows one to reconstruct complete information on the operator $X$ on the algebra $ \mathcal{W}(\chi, \sigma^{\omega})$ \cite{BratteliR82}.

\section{Model dynamical generators }
\label{sec_generators}

In this section we present the class of dynamical generators under investigation. To start with, we focus on many-body systems subject to a Markovian open quantum dynamics \cite{Lindblad76,BreuerP:2002}. This is described by means of a quantum master equation, with a time-independent dynamical generator in Gorini-Kossakowski-Sudarsan and Lindblad (GKS-Lindblad) form \cite{Lindblad76, GoriniKS76}. 

Any operator $O\in \mathcal{A}$ evolves in time according to the equation
\begin{equation}
\dot{O}(t)=\mathcal{L}_N[O(t)]\, ,
    \label{QME}
\end{equation}
where $\mathcal{L}_N$ identifies the GKS-Lindblad operator evolving observables, (i.e., the dual with respect
to the trace operation of the generator evolving states in the Hilbert space). The above equation is formally solved by $O(t)=e^{t\mathcal{L}_N}[O]$. Here, the label $N$ stresses that, according to the typical procedure for analyzing the emergent dynamics at infinite system size, we first define the dynamical generator $\mathcal{L}_N$ for finite, sized-$N$ systems, deriving only eventually the asymptotic dynamics (in the limit $N\to\infty$). 

Before further describing the form of the dynamical generator, let us introduce a basis, $\{v_\alpha\}_{\alpha=1}^{d^2}$, for the single-particle algebra $M_{d}(\mathbb{C})$. More specifically,  we consider the collection of operators $\{v_\alpha\}_{\alpha=1}^{d^2}$, to form an hermitian, $v_\alpha=v_\alpha^\dagger$, and orthogonal $\tr{(v_{\alpha} v_{\beta})}=\delta_{\alpha \beta}$ (implying  $\|v_\alpha\|\le 1$) basis. Any operator $x\in M_d(\mathbb{C})$ can thus be written as
\begin{equation}
x=\sum_{\alpha=1}^{d^2} \tr{(x \, v_{\alpha})}v_{\alpha}\, .
\end{equation}
We also introduce the structure coefficients $a_{\alpha\beta}^{\gamma}$ for the above defined basis,
\begin{equation}
[v_{\alpha},v_{\beta}] = \sum_{\gamma=1}^{d^2}  a_{\alpha\beta}^{\gamma}  v_{\gamma}, \quad a_{\alpha\beta}^{\gamma} \equiv \tr{([v_{\alpha},v_{\beta}] v_{\gamma})} \, ,
\end{equation}
which will be useful in the next Sections.

The GKS-Lindblad generator can be decomposed as
\begin{equation}
 \mathcal{L}_N[O]=i[H,O]+\sum_{\ell=1}^{q} \mathcal{D}_\ell[O]\, ,
    \label{Lindblad}
\end{equation}
where, $H$ is the Hamiltonian of the system and it reads
\begin{equation}\label{e0_totalHamiltonian}
H= \sum_{k=1}^{N} \sum_{\alpha=1}^{d^2} \epsilon_{\alpha} v_{\alpha}^{(k)} + \frac{1}{N} \sum_{k,j=1}^{N} \sum_{\alpha, \beta=1}^{d^2}  h_{\alpha \beta} v_{\alpha}^{(k)} v_{\beta}^{(j)}\, ,
\end{equation}
with $\epsilon_\alpha \in \mathbb{R}$, $h_{\alpha\beta}=h_{\beta\alpha}^*$. The first term on the right-hand side of the above equation represents a single-particle contribution. The second one represents an all-to-all, two-body interaction with a strength proportional to $1/N$. The terms contributing to the collection of maps   $\{ \mathcal{D}_\ell \}$, that we can dub dissipators, describe instead dissipative contributions to the time evolution. We take them to be of the form 
\begin{equation}
    \mathcal{D}_\ell[O]=\frac{1}{2}\sum_{k=1}^N \left([J_{\ell}^{k\,  \dagger}, O] J_{\ell}^{k} + J_{\ell}^{k\, \dagger }[ O, J_{\ell}^{k} ]  \right) \, ,
    \label{dissipator}
\end{equation}
with
\begin{equation}
J_\ell^{k}=j_\ell^{(k)}\Gamma_\ell(\Delta_N^\ell)
    \label{jumps}
\end{equation} 
being the jump operators. These operators are thus structured as follows: $j_\ell^{(k)}$ acts solely on site $k$; instead, $\Gamma_\ell(\Delta_N^\ell)=[\Gamma_\ell(\Delta_N^\ell)]^\dagger$ is an (hermitian) operator-valued function computed for the operator $\Delta_N^\ell=[\Delta_N^\ell]^\dagger$. The latter operators are taken to be real, linear combinations of average operators \eref{eq:average-operators}, i.e., 
\begin{equation}\label{eq:Delta_lc_average}
\Delta_N^\ell=\sum_{\alpha=1}^{d^2}r_{\ell\alpha} \left[\frac{1}{N}\sum_{k=1}^{N}v_\alpha^{(k)}\right]\,, \qquad \mbox{ with }r_{\ell \alpha}\in \mathbb{R}.
\label{delta}
\end{equation}
From their definition, the above defined linear combination are norm-bounded operators, i.e., $\|\Delta_N^\ell\|\le \delta_\ell$, where 
\begin{equation}
 \delta_\ell =\sum_{\alpha=1}^{d^2}|r_{\ell\alpha}|<\infty\, .
 \label{delta_alpha}
\end{equation}
In the rest of this work, we consider functions $\Gamma_\ell(\Delta_N^\ell)$ that satisfy the following:

\begin{assumption}
\label{Gamma}
The operator-valued functions $\Gamma_\ell(\Delta_N^\ell)$ can be written as power series 
$$
\Gamma_\ell(\Delta_N^\ell)=\sum_{n=0}^\infty c_\ell^n (\Delta_N^\ell)^{n}\, , 
$$
with coefficient $c_\ell^n$ such that for any $z\in\mathbb{R}$
\begin{equation}
    \gamma(z)=\sum_{n=0}^\infty |c_\ell^n||z|^n <\infty\, .
\end{equation}
The assumption on the series $\gamma(z)$ also implies that
$$
\gamma'(z):=\sum_{n=0}^\infty n|c_\ell^n||z|^{n-1} <\infty\, .
$$
\end{assumption}
Therefore, we are considering functions $\Gamma_\ell$ that admit a Taylor expansion, around zero, with infinite radius of convergence. Although this is a strong assumption, it allows us to find results for a broad class of dynamical generators. Considering Assumption~\eref{Gamma}, the following result holds true (for a proof, see Appendix of \cite{FiorelliEtALC_NJP_23})
\begin{lemma}
\label{lemma_commutators}
For any given operator-valued function $\Gamma_\ell(\Delta_N^\ell)$ satisfying Assumption~\eref{Gamma}, the following relations hold 
\begin{eqnarray*}
	&i) \left\| \left[\Gamma_\ell(\Delta_N^\ell), O\right] \right\| \le \frac{2N_O}{N}\|O\|\delta_\ell \gamma'(\delta_\ell) \, ,\\
	&ii) \left\| \left[\Gamma_\ell(\Delta_N^\ell), X_{N}\right] \right\| \le \frac{2}{N}\|x\| \delta_\ell  \gamma'(\delta_\ell)\, ,\\
	&iii) \left\| \left[\Gamma_\ell(\Delta_N^\ell),\left[\Gamma_\ell(\Delta_N^\ell), O\right]\right] \right\| \le \frac{4N_O^2}{N^2}\|O\|\delta_\ell^2 [\gamma'(\delta_\ell)]^2  \, , \\
	& iv)  \left\| \left[\Gamma_\ell(\Delta_N^\ell),\left[\Gamma_\ell(\Delta_N^\ell), X_N\right]\right] \right\| \le \frac{4}{N^2}\|x\|\delta_\ell^2 [\gamma'(\delta_\ell)]^2    \, ,
\end{eqnarray*}
with $O$ being any operator with strictly local support, $N_O$ the length of such support, and $X_N$ any average operator as defined in Eq.~\eref{eq:average-operators}.
\end{lemma}

\begin{figure}[t]
\centering
\includegraphics[width=0.8\textwidth]{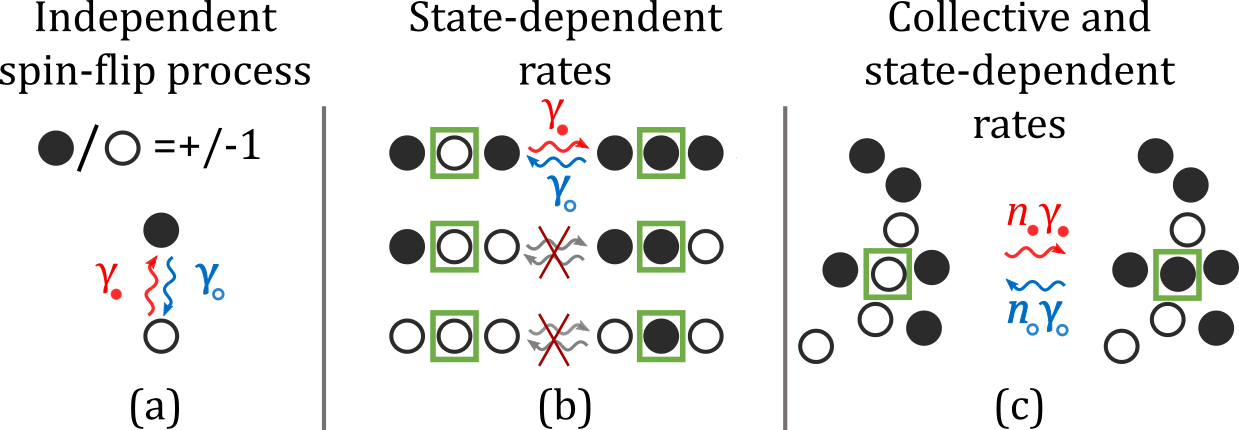}
\caption{{\bf Classical, collective state-dependent rates.} Representation of transitions in stochastic Markovian processes for systems made of classical Ising spins. Each one of the particle can be found in the configuration $\bullet / \circ$, corresponding to, e.g., excited state and ground state, respectively. In panel (a) a classical stochastic non-interacting dynamics is sketched. It consists of independent spin-flips, $\bullet {{\, }}^{\to}  /_{ \leftarrow} \circ$, occurring at rates $\gamma_{\circ /\bullet}$, respectively. Here, transitions rates are independent of the state of neighboring particles. Panel (b) represents an instance of kinetically-constrained model, in which the central particle can change its state only if the neighboring ones are both in the excited state $\bullet \,$. Finally, in panel (c), a collective all-to-all model is displayed, as the dynamics sketched in panel (b) features here transition rates that depend on the square of the density of excited particles, $n_\bullet$. }
\label{Fig1}
\end{figure} 

\subsection{Origin of dissipation with operator-valued rates}
\label{motivation}
Before moving further, we illustrate the motivation behind the choice of the dissipators with jumps operators defined by Eq.~\eref{jumps}. Let us consider an ensemble of $N$ classical (Ising) spin-$1/2$ particles, where each spin-$1/2$ particle can either be found in an {\it excited state} $\ket{\bullet}$ or in a {\it ground state} $\ket{\circ}$ [cf.~Fig.~\ref{Fig1}(a)]. For this system, the simplest stochastic Markovian dynamics is given by a non-interacting ``thermal" time evolution. Here, each particle undergoes an independent spin-flip process at rate  $\gamma_\circ$ [$\gamma_\bullet$] for the transition from the excited state to the ground state, $\ket{\bullet}\to\ket{\circ}$ [and viceversa], as shown in Fig.~\ref{Fig1}(a). This example does not show particularly interesting dynamical nor stationary features. 

More complex dynamics can emerge when the rate for the single-particle transitions does depend on the configuration of the rest of the system. We exemplify this case in [cf.~Fig.~\ref{Fig1}(b)], where the rate of flipping the $k$-th spin into the excited state depends on whether particles $k-1$ and $k+1$ are in their excited or in their ground state. Situations with configuration-dependent rates include, e.g., the so-called kinetically-constrained models \cite{fredrickson1984,cancrini2008,garrahan2011}, where only transitions satisfying given constraints are permitted, and collective all-to-all classical Hamiltonian functions, where transition rates depend on collective properties of the system \cite{glauber1963,walter2015}. To this regard, we can give an example of how to generalize the example illustrated in Fig.~\ref{Fig1}(b) to the case of collective rates. This can be achieved, e.g., by choosing rates that depend on the square of the operator $n_{\bullet}$, describing the density of excited states, $n_\bullet=\frac{1}{N}\sum_{k=1}^{N}n^{(k)}$, $n^{(k)}$ representing the operator $n=\ket{\bullet}\!\bra{\bullet}$ for the $k$th particle [see an illustration in Fig.~\ref{Fig1}(c)]. 
This dynamics, like any classical stochastic dynamics, can be embedded in the density-matrix formalism of open quantum systems \cite{Garrahan18}. One can indeed introduce a dynamical generator --- which preserves diagonal density matrices (see, e.g., Ref.~\cite{CarolloGK:JSP:21}) --- as follows \footnote{\label{gen-dualgen} The dynamical generators acting on density-matrices is identified with a $*$. Its dual operator, with respect to the trace operation, evolves instead operator.}
$$
\mathcal{D}^{*}[\rho]=\sum_{k=1}^N \left(J_\bullet^{k}\rho J_\bullet^{k\, \dagger }-\frac{1}{2}\left\{J_\bullet^{k\, \dagger } J_\bullet^{k},\rho \right\}\right)+\sum_{k=1}^N \left(J_\circ^{k}\rho J_\circ^{k\, \dagger }-\frac{1}{2}\left\{J_\circ^{k\, \dagger } J_\circ^{k},\rho \right\}\right)\, ,
$$
with  
$$
J_\bullet^k=\sqrt{\gamma_\bullet} \sigma_+^{(k)} n_\bullet\, , \qquad \qquad J_\circ^k=\sqrt{\gamma_\circ}\sigma_-^{(k)} n_\bullet\, ,
$$
and $\sigma_{+}=\ket{\bullet}\!\bra{\circ}$,  $\sigma_-=\sigma_+^\dagger$. Note that the rates are now operator-valued functions of a collective observable, namely the density of excited particles $n_\bullet$. 

Even though formulated in a quantum language, the above dynamics is fully classical. Nonetheless, it is straightforward to add quantum coherent Hamiltonian contributions to such a dissipative generator, and to investigate their impact on the behavior of the system. This can be done by considering the more general quantum master equation
\begin{equation}
\dot{\rho}_t=\mathcal{L}^*[\rho_t]:=-i[H,\rho_t]+\mathcal{D}^*[\rho_t]\, .
\label{Lindblad-generator-Schr}
\end{equation}
As introduced in the previous section, these are the type of GKS-Lindblad generators we consider: with single-particle and all-to-all interacting Hamiltonian $H$, and dissipation characterized by collective operator-valued functions. In relation to the latter, the argument of the operator-valued function is an average operator, $\Delta^{\ell}_N$, so that the squared function $\Gamma^2_\ell(\Delta_N^\ell)$ should have the structure of an average operator on clustering states. This consideration, together with the last example [cf. Fig.~\ref{Fig1}(c)], is meant to introduce the idea that, given the jump operator $J^k_\ell$ of Eq.~\eref{jumps}, the squared function $\Gamma^2_\ell(\Delta_N^\ell)$ behaves as an operator-valued \textit{rate} for the transition implemented by the on-site operator $j^{(k)}_\ell$. Rigorously, this has been shown in Ref.\cite{FiorelliEtALC_NJP_23} and reported in the next section. 

\section{Dynamics at macroscopic scale}
\label{sec4}

In this section, we consider the behavior of the average operators introduced by Eq.~\eref{eq:average-operators}. As they give rise, in the thermodynamic limit, to a commuting algebra, they account for a collective and \textit{macroscopic} description of the many-body system. The results shown in this Section summarize the main contribution of Ref.~\cite{FiorelliEtALC_NJP_23}, i.e., a proof of the validity of the mean-field treatment with respect to the average operators introduced in Sec. \ref{sec3}, in the thermodynamic limit. 

Before focusing on the relation between the evolution of the averages operator and their corresponding equations under mean-field theory, we state a preliminary result concerning the action of the map on strictly local operators, and on average operators:

\begin{lemma}
\label{Cor_diss_dyn} 
The maps $\mathcal{D}_\ell$ defined by Eqs.~\eref{dissipator}-\eref{jumps} with functions $\Gamma_\ell(\Delta_N^\ell)$ obeying Assumption~\eref{Gamma} are such that 
\begin{eqnarray*}
	& \left\|\mathcal{D}_\ell[O]-\Gamma_{\ell}^2(\Delta_N^\ell) \mathcal{D}_\ell^{\rm Loc}[O]\right\|\le \frac{C_O}{N}\, ,\\
	& \left\|\mathcal{D}_\ell[X_{N}]-\Gamma_{\ell}^2(\Delta_N^\ell) \mathcal{D}_\ell^{\rm Loc}[X_N]\right\|\le \frac{C_{x}}{N}\, , \qquad 
\end{eqnarray*}
with 
\begin{equation}
\mathcal{D}^{\rm Loc}_\ell[A]=\frac{1}{2}\sum_{k=1}^N \left([{j}_{\ell}^{ \dagger\, (k)}, A] {j}_{\ell}^{(k)} + {j}_{\ell}^{ \dagger\, (k)}[A, {j}_{\ell}^{(k)}]  \right)\, .
\label{D_loc}
\end{equation} 
In the above expression, $O$ is any local operator with support on a finite number of sites, $N_O$ is the extension of its support, and $X_N$ the average operator constructed from the  single-particle operator $x$, as shown in Eq.~\eref{eq:average-operators}, and $C_O$, $C_{x}$ are $N$ independent constants, reading
$$
C_O= 2 N_O \|O\| \| j_{\ell} \|^2 \left\lbrace \delta_{\ell}\gamma'(\delta_{\ell})[ \delta_{\ell} \gamma'(\delta_{\ell})(1+N_O)+ 3 N_O \gamma(\delta_{\ell})]  + 2 \gamma(\delta_{\ell}) \delta^2_{\ell} \gamma''(\delta_{\ell}) \right\rbrace \, ,
$$
$$C_{x} = 2 \| x \| \| j_{\ell} \|^2 \left\lbrace [2 \delta_{\ell}\gamma'(\delta_{\ell}) +3\gamma(\delta_{\ell})]\delta_{\ell}\gamma'(\delta_{\ell}) + 2 \gamma(\delta_{\ell}) \delta^2_{\ell} \gamma''(\delta_{\ell})\right\rbrace \, ,
$$ 
respectively.
\end{lemma}
That is, in the thermodynamic limit, the maps $\mathcal{D}_{\ell}$ act on strictly local operators and average ones as if they were local maps weighted by a pre-factor equal to $\Gamma_{\ell}^2(\Delta_N^\ell)$. This together with Assumption~\eref{Gamma}, clarifies that the considered jump operators implement \textit{local} transitions associated with operator-valued \textit{rates}. 

\subsection{Heisenberg equations and mean-field dynamics}
\label{Heis-MF}

Let us consider average operators of the type defined by Eq.~\eref{eq:average-operators}. By means of the single-particle algebra $\{v_{\alpha} \}_{\alpha=1}^{d^2}$, we can construct a basis for all possible average operators, and deal with the set
\begin{equation}
    m_\alpha^N=\frac{1}{N}\sum_{k=1}^N v_\alpha^{(k)}\, , \qquad \alpha=1,2,\dots d^2\, ,
\end{equation}
In order to derive their time evolution, $e^{t\mathcal{L}_N}[m_\alpha^N]$, in the thermodynamic limit, and starting from a translation invariant clustering state as in Definition \ref{clustering}, one needs first to compute the Heisenberg equations of motion. These read
\begin{equation}
    \frac{d}{dt} e^{t\mathcal{L}_N}[m_\alpha^N]=\mathcal{L}_N\left[e^{t\mathcal{L}_N}[m_\alpha^N]\right]=e^{t\mathcal{L}_N}\left[\mathcal{L}_N [m_\alpha^N]\right]\, ,
    \label{eq:Heisenberg}
\end{equation}
and require, in turn, to control the action of the GKS-Lindblad operator on average operators, $\mathcal{L}_N [m_\alpha^N]$. 
\begin{lemma}
\label{lemma_gen_action}
Given the generator $\lin_N$ specified by Eqs.~\eref{Lindblad}-\eref{jumps}, with functions $\Gamma_\ell(\Delta_N^\ell)$ obeying Assumption~\eref{Gamma}, we have that 
$$
\| \mathcal{L}_N[m^{N}_{\alpha}] -  f_\alpha(\vec{m}^N)\| \leq \frac{C_{L}}{N}
$$
where 
\begin{eqnarray*}
& f_\alpha(\vec{m}^N) =  i \sum_{\beta = 1}^{d^2} A_{\alpha \beta} m^{N}_{\beta} + i \sum_{\beta,\gamma=1}^{d^2} B_{\alpha \beta \gamma} m^{N}_{\beta} m^{N}_{\gamma} + \sum_{\ell, \beta} M_{\ell \alpha }^{\beta} \Gamma^{2}_{\ell}(\Delta^{\ell}_N) m^{N}_{\beta} \\
& A_{\alpha \beta} = \sum_{\beta'=1}^{d^2} \epsilon_{\beta'} a_{\beta' \alpha}^{\beta} \quad B_{\alpha \beta \gamma} = \sum_{\beta'=1}^{d^2} a_{\beta' \alpha}^{\gamma}(h_{\beta \beta'} + h_{\beta' \beta})\, .
\end{eqnarray*} 
Here, $M$ is a real matrix, such that the action of $\mathcal{D}^{\mathrm{Loc}}_{\ell}[\cdot]$ on an element of the single-site operator basis $v_\alpha^{(k)}$ reads
$$
\mathcal{D}_\ell^{\rm Loc}[v_\alpha^{(k)}]=\sum_{\beta=1}^{d^2}M_{\ell \alpha}^{\beta} v_\beta^{(k)}\, . 
$$
$C_{L}$ is an $N$-independent bounded quantity, that reads
$$C_{L} = d^8 h_{{\mathrm {max}}} a^{2}_{{\mathrm {max}}} +qC_v \, $$
where $h_{{\mathrm {max}}} = {\mathrm {max}}_{\beta,\beta'} h_{\beta, \beta'}$ ,  $a_{{\mathrm {max}}} = {\mathrm {max}}_{\alpha, \beta, \gamma} a_{\alpha  \beta}^{\gamma}$, and \\ $C_v = {\mathrm {max}}_{\forall \ell} \left\lbrace 2 \| j_{\ell} \|^2  [2\delta_{\ell} \gamma'(\delta_{\ell}) +3\gamma(\delta_{\ell})] \delta_{\ell} \gamma'(\delta_{\ell}) +2 \gamma(\delta_{\ell}) \delta^2_{\ell} \gamma''(\delta_{\ell}) \right\rbrace$.
\end{lemma}
By Lemma \ref{lemma_gen_action}, we see that, in the thermodynamic limit, the GKS-Lindblad generator acts on average operators $m_{\alpha}^N$ as a nonlinear function $f_\alpha$ of the average operators themselves, for any $\alpha$. However, being open quantum dynamics not represented by an automorphism (i.e. there appear $e^{t \linn}[m_{\alpha}^N m_{\beta}^N] \neq e^{t \linn}[m_{\alpha}^N]e^{t \linn}[ m_{\beta}^N]$), the nonlinear character of the function $f_\alpha$ prevents the Heisenberg equations \eref{eq:Heisenberg} to close on the set operators $e^{t\mathcal{L}_N}[m_\alpha^N]$. Moreover, the action of the generator on the function $f_\alpha$ gives rise, in the thermodynamic limit, to an infinite hierarchy of equations, which is hardly solvable. 

Nonetheless, the evolution equations for the expectation values of the average operators can be derived from Eq.~\eref{eq:Heisenberg}. By defining $\omega_t\left(A\right):=\omega\left(e^{t\mathcal{L}_N}[A]\right)$ and by Lemma \ref{lemma_gen_action}, we have
\begin{equation}
    \frac{d}{dt} \omega_t\left(m_\alpha^N\right)\approx \omega_t\left(f_\alpha(\vec{m}^N)\right)\, .
    \label{eq:expect}
\end{equation}
At this point, to make some progresses, one would assume that expectations of products of average operators factorize into the product of the expectations, $  \omega_t(m^N_\alpha m^N_\beta)\approx \omega_t(m^N_\alpha)\omega_t(m^N_\beta)$, yielding the so-called mean-field equations of motion for the considered GKS-Lindblad generator,
\begin{equation}
    \frac{d}{dt} m_\alpha (t)= f_\alpha(\vec{m})\, .
    \label{eq:mean-field}
\end{equation}
In the above expression, $m_\alpha(t)$ represents a time-dependent function that could reproduce the dynamics of $\omega_t\left(m_\alpha^N\right)$ in the thermodynamic limit. For this to be case, the initial conditions must be appropriately chosen, i.e.
\begin{equation}
    m_\alpha(0)=\lim_{N\to\infty}\omega (m_\alpha^N)\, .
    \label{init-cond}
\end{equation} 
However, albeit the system of Eqs.~\eref{eq:mean-field}, \eref{init-cond} is in principle solvable, nothing yet guarantees that it provides the exact dynamics of average operators. In fact, Ref.~\eref{eq:mean-field} rigorously shows the exactness of the mean-field equations. It does so by introducing the cost function \begin{equation}
\E (t) = \sum_{\alpha = 1}^{d^2} \omega_t ([m_{\alpha}^{N} -m_{\alpha}(t)]^{2}) \,,
\end{equation} and proving the following:
\begin{theorem}\label{theorem}
Given a generator as the one in Eqs.~\eref{Lindblad}-\eref{delta}, with functions $\Gamma_{\ell}(\Delta^{\ell}_N)$ satisfying Assumption~\eref{Gamma}, 
\begin{equation}\label{e_theorem}
if \quad \lim_{N \rightarrow \infty} \E(0) = 0, \quad then \quad  \lim_{N \rightarrow \infty} \E(t) = 0,  \forall  t< \infty .
\end{equation}
\end{theorem}
Indeed, if $\lim_{N\to\infty}\E (t)=0$, then the state $\omega_t$ is clustering in the thermodynamic limit. Moreover, via the Cauchy-Schwarz inequality, it is $|\omega_t(m_{\alpha}^{N}-m_{\alpha})| \le \sqrt{ \omega_t([m_{\alpha}^{N}-m_{\alpha}]^2)} \le \sqrt{\mathcal{E}_N(t)}$, which can be used to control the limit 
\begin{equation}\label{mf_limit}
\lim_{N \rightarrow \infty }\omega_t(m_{\alpha}^{N}) - m_{\alpha}(t) = 0, \qquad \forall t \, ,
\end{equation} 
showing that the state $\omega_t$ is clustering in the sense of Definition \ref{clustering}.

Before going further, we state here some Lemmas that will be used in the following (see the Appendices of Ref.~\cite{FiorelliEtALC_NJP_23} for their proof): 
\begin{lemma}\label{lemma_bound_mf}
The system of equations \eref{eq:mean-field} with initial conditions $m_\alpha(0)$, defined by a quantum state $\omega$ as in Eq.~\eref{init-cond}, has a unique solution for $t\in[0,\infty)$. Moreover, one has 
$$
|m_\alpha(t)|\le \|v_\alpha\|\le 1\, , \qquad \forall t\, .
$$
\end{lemma}

\begin{lemma}\label{proof_Lemma-aux}
The convergence of the squared operator-valued rates to the same rates computed in their mean-field scalar function is dominated by the convergence of the mean-field operator to the mean-field scalar functions, namely we have that
\begin{eqnarray*}
& |\omega\left(A^\dagger e^{t\mathcal{L}_N}\left[(\Gamma_\ell^2(\Delta_N^\ell)-\Gamma_\ell^2(\Delta_\ell(t)))X\right] B\right)| \\
& \le  C \| X \| \sum_{\alpha=1}^{d^2} |r_{\ell \alpha}| \sqrt{\omega(A^{\dagger} e^{t \linn}[(m_{\alpha}^N-m_{\alpha}(t))^2]A)}\sqrt{\omega(B^{\dagger}B)} \, ,
\end{eqnarray*}
where $C=2\gamma(\delta_{\ell})\gamma'(\delta_\ell)$.

\end{lemma}

We will also make use of the following result, which proof can be found, e.g. in Ref.\cite{BenattiEtAl18}.

\begin{lemma}
\label{Lemma-dilation} Given any completely positive and unital map $\Lambda[\cdot]$ on the quasi-local algebra $\mathcal{A}$ and a state $\omega$, we have that 
$$
 | \omega\left(A^\dagger \Lambda[C^\dagger D]B\right) | \le \sqrt{\omega(A^\dagger \Lambda[C^\dagger C]A)}\sqrt{\omega\left(B^\dagger \Lambda[D^\dagger D]B\right)} \, .
$$
\end{lemma}

\section{Dynamics at mesoscopic scale}
\label{dyn-mesoscopic}

In this section we discuss the main focus of this paper, that is the mesoscale dynamics arising from the dissipative generator previously introduced. In particular, we investigate the time evolution of the quantum fluctuation operators, which scale as square root of $N$, and the time evolution of the covariance matrix.
To do so, we introduce some definitions that are going to be useful in the following. In particular, as in the previous section, we consider as set of relevant single-particle observables the orthonormal, hermitian basis $\lbrace v_{\alpha} \rbrace_{\alpha=1}^{d^2}$, focusing here on the corresponding set of fluctuation operators.

As we deal with a time-dependent state, $\omega_{t} (\cdot) = \omega(e^{t \mathcal{L}_N}[\cdot])$, also the local quantum fluctuations depend on time as
\begin{eqnarray}
    F^N_{\alpha}(t) = \frac{1}{\sqrt{N}} \sum_{k=1}^{N} \left( v_{\alpha}^{(k)} - \omega_t(v_{\alpha}) \right).
\end{eqnarray}
This implies that their commutator is a time-independent average operator,
\begin{equation}
\eqalign{
    [F_{\alpha}^N(t), F_{\beta}^N(t)] & = \frac{1}{N}\sum_{k=1}^{N} [v_{\alpha}^{(k)}, v_{\beta}^{(k)}] = \sum_{\gamma} a_{\alpha \beta} ^{\gamma} \frac{1}{N}\sum_{k=1}^{N} v_{\gamma}^{(k)} \\  & = \sum_{\gamma} a_{\alpha \beta} ^{\gamma} m^{N}_{\gamma}  \equiv T_{\alpha \beta}^N \, ,}
\end{equation}
whereas the elements of the symplectic matrix are time-dependent, 
\begin{equation}
\eqalign{
    \sigma_{\alpha \beta }^{\omega}(t) & = -i \lim_{N \rightarrow \infty} \omega_t \left( [F_{\alpha}^N(t), F_{\beta}^N(t)]  \right) = -i \lim_{N \rightarrow \infty} \omega_t \left( T_{\alpha \beta}^N \right) \\ &= -i \lim_{N \rightarrow \infty} \sum_{\gamma} a_{\alpha \beta} ^{\gamma} m_{\gamma}(t) \, .
}    
\end{equation}
For future convenience, we introduce here also the expression of the time-dependent covariant matrix, that reads
\begin{equation}
\Sigma_{\mu \nu}^{\omega}(t) = \frac{1}{2} \lim_{N \rightarrow \infty}\omega_t\left(\lbrace F^{N}_{\mu}(t), F^{N}_{\nu}(t)\rbrace \right) \, .
\end{equation}
Finally, for the sake of simplicity, in the remaining part of this section we make use of the more compact notation
\begin{equation}
\eqalign{    
& \vec{\omega}^N_t=\left(\omega_{1}^N(t),...,\omega_{d^2}^N(t) \right)^{T} , \quad \mathrm{where} \quad \omega_{ \alpha}^N(t) \equiv \omega_t(m_{\alpha}^N) \, , \\
& \vec{\omega}_t= \left(\omega_{1}(t),...,\omega_{d^2}(t) \right)^{T} , \quad \mathrm{where} \quad \omega_{ \alpha}(t) \equiv  m_{\alpha}(t)\, .}
\end{equation}

\subsection{Dynamics of quantum fluctuations and time evolution of the covariance matrix}

We want now to give an account on the action of the family of time-dependent map $e^{t \linn}(\cdot)$ on Weyl-like operators, i.e. on objects of the form
\begin{equation}
    W_t^N(\vec{r}) \equiv e^{t \mathcal{L}_N} [W^{N}(\vec{r})] \, .
\end{equation}
To do so, we  analyze the mesoscopic limit (in the sense of Def.~\ref{def:mesoscopicLimit}) of the dynamical map evolving local exponentials, showing that $$\lim_{N \rightarrow + \infty } \omega_{\vec{r}_1 \vec{r}_2} \left( e^{t\linn}[W^{N}(\vec{r})]\right) = \Omega_{\vec{r}_1\vec{r}_2}\left( \Phi^{\vec{\omega}}_t[W(\vec{r})]
 \right)\, .$$ First of all, we derive a preliminary result, leaving its proof in the appendix.
\begin{lemma}\label{lemma_squarefluctu_convergencemain}
    Given the generator $\linn$ specified by Eqs.~\eref{Lindblad}-\eref{delta} , with operator-valued functions $\Gamma_\ell (\Delta^\ell_N)$ obeying Assumption~\eref{Gamma}, 
    $\forall \mu = 1,...,d^2$ \, and $\forall t < \infty$,
    $$ \lim_{N \rightarrow +\infty} \omega \left(  W^N(\vec{r}) e^{t \linn} \left[ \left( F_{\mu}^{N}(t)\right)^2 \right]  W^{N \dagger}(\vec{r})  \right) < +\infty \, ,$$ 
    if it is satisfied at $t=0$.
\end{lemma}
By employing the Cauchy-Schwartz inequality on the latter condition, $\lim_{N \rightarrow + \infty} \omega( W^N(\vec{r}) [ \left( F_{\mu}^{N}\right)^2 ]  W^{N \dagger}(\vec{r})) < +\infty $, one obtains
\begin{equation}
   \lim_{N \rightarrow +\infty }\omega([F_{\mu}^N]^2) = \lim_{N \rightarrow +\infty } \omega([m_{\mu}^N-m_{\mu}(0)]^2)N \leq +\infty \, ,
\end{equation}
which, by Theorem \ref{theorem}, is satisfied when considering initial clustering states. For large but finite system size, one can enforce the condition $\E(0) \sim 1/N$, which is met by product states, or states with short-range correlations. 

We can now establish how the generator $\linn$ acts, in the thermodynamic limit, on local exponentials.
\begin{proposition}\label{proposition_mesodynamics}
    Given the local exponential $W_t^N(\vec{r}) = e^{i \vec{r}\cdot \vec{F}^N_t}$, a dynamical generator as defined by Eqs. \eref{Lindblad}-\eref{delta}, with $\Gamma_\ell(\Delta_N^\ell)$ obeying Assumption~\eref{Gamma}, it is
    \begin{equation}
        \eqalign{
         \lim_{N\rightarrow\infty} \omega_{\vec{r}_1 \vec{r}_2} & \left(e^{ t\linn} \left[ \linn \left[ W_t^N(\vec{r}) \right] \right] \right) = \\
        & \lim_{N\rightarrow\infty} \omega_{\vec{r}_1 \vec{r}_2} \left( e^{ t \linn} \left[ \left( i \sqrt{N} \vec{r} (\tilde{B}_t^N + i A +\Gamma^{2}M)\vec{\omega}_t^N \right. \right. \right. \\
        & -i \vec{r} \cdot \left[ T^N(2ih^{\mathrm{R}})-(\tilde{B}_{t}^N+iA+\Gamma^2M) \right] \vec{F}^N_t \\  
        &  -\frac{1}{2}\vec{r} \left[ T^N (2ih)T^N -(\tilde{B}_t^N + i A+\Gamma^{2}M)T^N \right. \\
        & \left. \left.\left. \left.+ T^N S T^N + \Gamma^2 G^N \right] \vec{r}  \right) W_t^N(\vec{r}) \right] \right) \, ,}
    \end{equation}
    where $ A_{\alpha \beta} = \sum_{\mu} \epsilon_{\mu} a_{\alpha \beta}^{\mu}$, $\tilde{B}_{\beta \gamma}(t) = \sum_{\mu\nu}(2ih_{\mu\nu}^{\mathrm{R}})a_{\beta\gamma}^{\mu} \omega^{N}_{\nu}(t)$, with $h^{\mathrm{R}}_{\alpha \beta} = \mathrm{Re}( h_{\alpha \beta}) $,  and 
    $$S = \sum_\ell \vec{r}_\ell \left(\Gamma_\ell'(\Delta_N^{\ell}) \right)^2  \frac{1}{N} \sum_{k}  j_\ell^{\dagger (k)} j_\ell^{ (k)}  \vec{r}_\ell  + \sum_\ell \Gamma_\ell(\Delta_N^{\ell}) \Gamma_\ell'(\Delta_N^{\ell})  \vec{r}_\ell \vec{c}_{\ell} \, ,$$
    $$\Gamma_N^2 G^N = \sum_\ell \Gamma_\ell^2(\Delta_N^{\ell}) G^N_\ell \, , \qquad G^{N}_\ell =  \sum_{\alpha \beta} \sum_{\gamma \gamma'} a_{\mu \alpha}^{\gamma} c^{j,*}_{\alpha}d_{\gamma \gamma'}^{\eta} m_{\eta}^N c^{j}_{\beta}  a_{\beta \nu}^{\gamma'}$$
\end{proposition}
The proof of the proposition is given by Lemma \ref{lemma7}-\ref{lemma10} in \ref{app:B}. \setcounter{lemma}{11} Finally, the following Lemma is reported for the sake of completeness (for its proof see, e.g., Lemma B3 of Ref. \cite{BenattiEtAl18}), being an auxiliary step for the following results.
\begin{lemma}\label{B_12:lemma_aux}
    Given a time-dependent, hermitean matrix $M_t$, and $N_t=e^{iM_t}$, then it is
    $$\dot{N}_t = \frac{d}{dt} N_t = O_t N_t, \qquad O_t \equiv \sum_{k=1}^{\infty} \frac{i^k}{k!} \mathbb{K}_{M_t}^{k-1}[\dot{M}_t] \, ,$$
    where $\mathbb{K}_{M_t}^{k}[\dot{M}_t] = \left[M_t , \mathbb{K}_{M_t}^{k-1}[\dot{M}_t]  \right]$, and $\mathbb{K}_{M_t}^{0}[\dot{M}_t]=\dot{M}_t$.
\end{lemma}

With the above results, we can state and prove the following:
\begin{theorem}\label{theorem-mesoscopicdynamics}
    Given the a dynamical generator $\linn$ as defined by Eqs.~\eref{Lindblad}-\eref{delta}, satisfying Assumption~\eref{Gamma}, the dynamics of the quantum fluctuation is given by the mesoscopic limit $\Phi_t^{\vec{\omega}} \equiv (m)- \lim_{N\rightarrow\infty } e^{{t \mathcal{L}_N}}$, with
    $$\Phi_t^{\vec{\omega}}[W(\vec{r})]= \exp{\left({-\frac{1}{2} \vec{r} \cdot \left( Y_t(\vec{\omega})\vec{r} \right) }\right)} W\left( X_t^T(\vec{\omega}) \vec{r}\right) \, ,$$
    where
    \begin{equation}
    \eqalign{
    & X_t(\vec{\omega}) \equiv \mathbb{T} e^{\int_{0}^{t} ds {Q}(\vec{\omega}_s)} \, ,\\
    & {Q}(\vec{\omega}_t) \equiv  \vec{f}(\vec{\omega}_t) - i\sigma(\vec{\omega}_t) 2ih^{\mathrm{R}} \\ \,
    & Y_t(\vec{\omega}) \equiv \int_{0}^{t} ds X_{t,s}(\vec{\omega}) \left(\sigma(\vec{\omega}_s)S^{\mathrm{sym}} \sigma^{T}(\vec{\omega}_s) + \Gamma^2(\omega_s) G^{\mathrm{sym}}(\omega_s) \right) X^T_{t,s}(\vec{\omega})  }
    \end{equation}
\end{theorem}
In the last expression, $X_{t,s}(\vec{\omega}) = X_t(\vec{\omega})X_s^{-1}(\vec{\omega})$, $S^{\mathrm{sym}}$ and $G^{\mathrm{sym}}$ are the symmetric components of the matrix $S$ and $G$, respectively, introduced by Proposition \ref{proposition_mesodynamics}, and $\vec{f}(\cdot)$ identifies the components of the mean-field equations \eref{eq:mean-field}. We show now the demonstration of the above Theorem, the structure being based on a similar demonstration in Ref. \cite{BenattiEtAl18}. 

\begin{proof} We need to show that the following quantity
\begin{equation}
    I_W(t) \equiv \omega_{ \vec{r}_1 \vec{r}_2} \left(e^{t\linn}\left[ W_t^N(\vec{r}) \right]  - \Phi^{\vec{\omega}}_t \left[ W^N(\vec{r}) \right]  \right) \, ,
\end{equation}
is vanishing, in the thermodynamic limit. This can be written as
\begin{equation}
\eqalign{
    I_W(t) & = \omega_{\vec{r}_1 \vec{r}_2}\left( \int_{0}^t ds \frac{d}{ds} \left[ e^{s \linn} \left[ W_s^N(X^T_{t,s}(\vec{\omega})\vec{r}) \right] e^{- \frac{1}{2} \vec{r} \cdot (Y_{t,s}(\vec{\omega})\vec{r})}\right] \right)\\
   } \, ,
\end{equation}
where we defined the quantity
\begin{equation}
    Y_{t,s} = \int_{s}^{t}dt'X_{t,t'}(\vec{\omega}) \left(\sigma(\vec{\omega}_{t'})S^{\mathrm{sym}} \sigma^{T}(\vec{\omega}_{t'}) +\Gamma^2(\vec{\omega}_{t'}) G^{\mathrm{sym}}(\vec{\omega}_{t'}) \right) X^T_{t,t'}(\vec{\omega}) \, .
\end{equation}
By performing the derivative it reads
\begin{equation}
     I_W(t) =  \omega_{\vec{r}_1 \vec{r}_2}\left( \int_{0}^{t} ds \, e^{ s\linn} \left[ \zeta_{t,s}^N \right] e^{- \frac{1}{2} \vec{r} \cdot (Y_{t,s}(\vec{\omega})\vec{r})} \right) \, ,
\end{equation}
where
\begin{equation}\label{difference_th4}
   \eqalign{
   \zeta_{t,s}^N  & = \linn \left[ W_s^N(X^T_{t,s}(\vec{\omega})\vec{r}) \right] + \frac{d}{ds} W_s^N(X^T_{t,s}(\vec{\omega})\vec{r}) \\ 
   & -\frac{1}{2}\frac{d}{ds} \left( \vec{r}\cdot (Y_{t,s}(\vec{\omega})\vec{r} \right) W_s^N(X^T_{t,s}(\vec{\omega})\vec{r}) \, .
   }
\end{equation}
Let us focus on the second and third terms on the right-hand side of the above expression. Reminding that $W_s^N(X^T_{t,s}\vec{r})$ reads $e^{i X^T_{t,s}(\vec{\omega})\vec{r} \cdot \vec{F}_{s}^N}$, one can employ Lemma \eref{B_12:lemma_aux}, and write
\begin{equation}
\eqalign{
    \frac{d}{ds} W_s^N(X^T_{t,s}(\vec{\omega})\vec{r}) &  = \sum_{k=1}^{\infty} \frac{i^k}{k!} \mathbb{K}_{X^T_{t,s}(\vec{\omega})\vec{r} \cdot \vec{F}_{s}^N}^{k-1} \left[\frac{d}{ds}X^T_{t,s}(\vec{\omega})\vec{r} \cdot \vec{F}_{s}^N \right] W_s^N(X^T_{t,s}(\vec{\omega})\vec{r}) \, .
    }
\end{equation}
By substituting the following time-derivative,
$$\frac{d}{ds}X^T_{t,s}(\vec{\omega})\vec{r} \cdot \vec{F}_{s}^N = - \left( X^T_{t,s}(\vec{\omega})\vec{r}  \right)\cdot \left( Q(\omega_s)\vec{F}_s^N\right) - \sqrt{N} \left( X^T_{t,s}(\vec{\omega})\vec{r}  \right)\cdot \dot{\vec{\omega}}_s^{N} \, ,$$
and noticing that all the contributions for $k \geq 3$, that we call $Z_t^N$, vanish in norm, in the thermodynamic limit, it is
$$\eqalign{
    \frac{d}{ds} W_s^N(X^T_{t,s}(\vec{\omega})\vec{r}) & \approx \left( - i X^T_{t,s}(\vec{\omega})\vec{r} \cdot \left( Q(\vec{\omega_s}) \vec{F}_{s}^N \right)- i \sqrt{N}  ( X^T_{t,s}(\vec{\omega})\vec{r} ) \dot{\vec{\omega}}_s^{N} \right. \\
    & \left. +\frac{1}{2 }\left[ X^T_{t,s}(\vec{\omega})\vec{r} \cdot \vec{F}_{s}^N, \left(  X^T_{t,s}(\vec{\omega})\vec{r} \right) \cdot \left( Q(\vec{\omega_s}) \vec{F}_{s}^N   \right)  \right] \right)W_s^N(X^T_{t,s}(\vec{\omega})\vec{r}) \, ,
    }
    $$
where the notation $\approx$ highlights the fact that, in the thermodynamic limit, the contributions collected by $Z_t^N$ vanish in norm. One can further evaluate the commutator in the above expression as 
$$\left[ X^T_{t,s}(\vec{\omega})\vec{r} \cdot \vec{F}_{s}^N, \left(  X^T_{t,s}(\vec{\omega})\vec{r} \right) \cdot \left( Q(\vec{\omega_s}) \vec{F}_{s}^N   \right)  \right] = -\left( X^T_{t,s}(\vec{\omega})\vec{r} \right) \cdot \left( Q(\vec{\omega_s})T^{N} X^T_{t,s}(\vec{\omega})\vec{r} \right)  \, .$$
As for the last term on the right-hand side of Eq. \eref{difference_th4}, it is
\begin{equation}
\eqalign{
    \frac{d}{ds}\left( \vec{r}\cdot (Y_{t,s})\vec{r} \right) = 
    X_{t,s}(\vec{\omega}) \left(\sigma(\vec{\omega}_{s})S^{\mathrm{sym}} \sigma(\vec{\omega}_{s}) - \Gamma^2(\omega_s) G(\omega_s) \right) X^T_{t,s}(\vec{\omega}) \, ,  
    }
\end{equation}
so that we can write the difference $\zeta_{t,s}^N$ of Eq.~\eref{difference_th4} as
\begin{equation}
    \eqalign{
    \zeta^N_{t,s} & \approx \linn[W_s^N(\vec{\rho})]-\left( i\sqrt{N}\vec{\rho} \cdot \dot{\vec{\omega}}_S^N + i  \vec{\rho} \left( Q(\vec{\omega}_s) \vec{F}_s^{N} \right)\right) W_{s}^N(\vec{\rho})   \\
    & -\frac{1}{2} \left( \vec{\rho} \cdot \left( Q(\vec{\omega}_s)T^{N}+ \sigma(\vec{\omega}_{s})S^{\mathrm{sym}} \sigma(\vec{\omega}_{s}) - \Gamma^2(\vec{\omega}_s) G(\vec{\omega}_s)  \right)\vec{\rho}\right)W_s^{N}(\vec{\rho}) \, ,
    }
\end{equation}
having defined, for the sake of a lighter notation, $\vec{\rho} = X^{T}_{t,s}(\vec{\omega}) \vec{r}.$ In order to go ahead, Proposition \eref{proposition_mesodynamics} can be employed, as it provides the behavior, in the mesoscopic limit, of the action of the generator $\linn$ on local exponential. As such, in the thermodynamic limit the quantity $e^{s \linn }[\zeta^N_{t,s}]$, inside the integral $I_W^N(t)$ can be written as
\begin{eqnarray}\label{pr4_clE}
    & i e^{s \linn}\left[  \sqrt{N}\vec{\rho} \cdot \left( \left( \tilde{B}^N_s + i A + \Gamma^2 M\right)\vec{\omega}^N_s - \dot{\vec{\omega}}_s^N \right) W_s^N(\vec{\rho}) \right]\\ \label{pr4_F}
    & -i  e^{s \linn} \left[ \left( \vec{\rho} \cdot \left( T^N(2ih^R)-(\tilde{B}^N_s+iA+\Gamma^2 M) \right.\right. \right. \\ \nonumber
    & \left.\left.\left.+Q(\vec{\omega}_s)\right) \vec{F}^N_s \right) W_s^N(\vec{\rho}) \right] \\  \label{pr4_tst}
    &-\frac{1}{2} e^{s \linn} \left[ \left( \vec{\rho} \cdot \left( T^N S T^N +  \sigma(\vec{\omega}_{s})S^{\mathrm{sym}} \sigma(\vec{\omega}_{s})  \right) \vec{\rho} \right) W_s^N(\vec{\rho})\right] \\ \label{pr4_gamma2G}
    & -\frac{1}{2} e^{s \linn} \left[ \left( \vec{\rho} \cdot \left( \Gamma^2 G^N  - \Gamma^2(\vec{\omega}_s) G^{\mathrm{sym}}(\vec{\omega}_s)  \right) \vec{\rho} \right) W_s^N(\vec{\rho})\right] \\ \label{pr4_last}
    & -\frac{1}{2}e^{s \linn} \left[ \left( \vec{\rho} \cdot \left( T^N (2ih) T^N - (\tilde{B}^N_t + iA + \Gamma^2 M)T^N \right. \right. \right. \\ \nonumber
    & \left. \left. \left. + Q(\vec{\omega}_s)T^{N} \right)\vec{\rho} \right) W_s^N(\vec{\rho})  \right] \, .
\end{eqnarray}
We will now show that the mesoscopic limit of the terms in Eqs.~\eref{pr4_clE} - \eref{pr4_last} vanishes. We have to deal with terms proportional to
$$ i) \quad \omega_{\vec{r}_1\vec{r}_2} \left(  e^{t \linn} [\left(m_{\alpha}^{N}-m_{\alpha}(t)\right) X W^{N}_s(\vec{\rho})]\right) \,  ,$$
with $ X=\mathbb{I}$, $X=m_{\beta}^N-m_{\beta}(t)$, $ X = \Gamma^2_{\ell}(\Delta_N^\ell) - \Gamma^2_\ell(\Delta_\ell(t))$, $ X= (\Gamma'_{\ell}(\Delta_N^\ell) - \Gamma'_\ell(\Delta_\ell(t)))( \Gamma_{\ell}(\Delta_N^\ell) - \Gamma_\ell(\Delta_\ell(t)))$, or their products, as in \eref{pr4_clE}, \eref{pr4_tst} - \eref{pr4_last};
And terms of the type 
$$ii)   \quad \omega_{\vec{r_1}\vec{r_2}} \left(  e^{t \linn} [  Y F_s^N(v_{\beta}) W^{N}_s(\vec{\rho})]\right)\, ,$$ 
with
$Y= m_{\alpha}^{N}-m_{\alpha}(t) $, and $Y= \Gamma^2_{\ell}(\Delta_N^\ell) - \Gamma^2_\ell(\Delta_\ell(t)) $, as in \eref{pr4_F}.
Let us first consider terms of the type $i)$, that we can write as
$$
\eqalign{
& | \omega_{\vec{r}_1\vec{r}_2} \left(  e^{t \linn} [\left( m_{\alpha}^N - m_{\alpha}(t) \right)
  X W^{N}_s(\vec{\rho})]\right)| \leq  \| X \| \sqrt{\omega\left(  W^N(\vec{r}_1) e^{t\linn} \left[ \left( m_{\alpha}^N - m_{\alpha}(t) \right)^2 \right]W^{\dagger \, N}(\vec{r}_1) \right)}
  } \, 
$$ 
with $\| X\|$ being finite and independent of $N$. By employing Theorem \ref{theorem} and Lemma \ref{Lemma-dilation}, as well as that $\lim_{N \rightarrow +\infty }\|[ m_{\alpha}^N,W^N(\vec{r})] \| = 0 $, it is 
\begin{equation}\label{pr4_limit_average}
\eqalign{
\lim_{N \rightarrow \infty} & \omega\left(  W^N(\vec{r}_1) e^{t\linn} \left[ \left( m_{\alpha}^N - m_{\alpha}(t) \right)^2 \right]W^{\dagger \, N}(\vec{r}_1) \right) \\
& = \lim_{N \rightarrow \infty} \omega \left( e^{t \linn}[ \left( m_{\alpha}^N - m_{\alpha}(t) \right)^2] \right) = 0 \, .
}
\end{equation} 
We can thus conclude that Eq.~\eref{pr4_clE} vanishes. As for the mesoscopic limits of Eqs.~\eref{pr4_tst}-\eref{pr4_last}, one can see that part of their contributions vanish because of antisymmetry. For instance, terms such as $\vec{\rho} \cdot (\sigma(\vec{\omega}_s))S^{\mathrm{asy}}(\sigma(\vec{\omega}_s)) \vec{\rho}$, being $S^{\mathrm{asy}}$ the antisymmetric component of the matrix $S$ introduced by Proposition 1.
Concerning terms of the type $ii)$, we exploit that both average operators and operator valued-function $\Gamma_{\ell}(\Delta_N^\ell)$ commute, in the thermodynamic limit, with fluctuations and local exponential (see, e.g., Eq. \eref{comm_localexponential} in \ref{app:B}). Thus it is
$$\lim_{N \rightarrow +\infty } \omega_{\vec{r_1}\vec{r_2}} \left(  e^{t \linn} [  Y F_{\beta}^N(s) W^{N}_s(\vec{\rho})]\right) = \lim_{N \rightarrow +\infty } \omega_{\vec{r_1}\vec{r_2}} \left(  e^{t \linn} [  F_{\beta}^N(s) W^{N}_s(\vec{\rho})Y]\right) \, .$$
Furthermore, we can employ 
$$
\eqalign{
& |\omega_{\vec{r}_1\vec{r}_2} \left(  e^{t \linn} [F_{\beta}^N(s)
W^{N}_s(\vec{\rho})Y]\right)| \\
& \leq \sqrt{\omega\left(  W^N(\vec{r}_1) e^{t\linn} \left[ \left( F_{\beta}^N(s) \right)^2 \right]W^{\dagger \, N}(\vec{r}_1) \right)}\sqrt{\omega\left(  W^N(\vec{r}_1) e^{t\linn} \left[ \left(Y \right)^2 \right]W^{\dagger \, N}(\vec{r}_1) \right)} \, ,
  }
$$ 
where the first square root converges to finite values, as given by Lemma \ref{lemma_squarefluctu_convergence}. The second term under square root vanishes in the thermodynamic limit. Indeed, for 
$Y= m_{\alpha}^{N}-m_{\alpha}(t) $ the limit in Eq.~\eref{pr4_limit_average} directly applies. In the case $Y= \Gamma^2_{\ell}(\Delta_N^\ell) - \Gamma^2_\ell(\Delta_\ell(t)) $, one needs to employ Lemma \ref{proof_Lemma-aux}, (see, e.g., as it is formulated by Eq.~\eref{lemma-aux-adt}) and then the limit in Eq.~\eref{pr4_limit_average} holds.
\end{proof}

We stress here that, even though we can structure the above proof similarly to Ref.~\cite{BenattiEtAl18}, we are dealing with a different type of GKS-Lindblad generator. In the cited reference, the Hamiltonian term is indeed perturbed by means of a dissipation that scale as $1/N$. The latter is compatible with generators obtained from microscopic models describing, e.g., light-matter interaction via Dicke models. In our case, while displaying similar Hamiltonian part, we deal with a distinct dissipator, as thoroughly discussed in Section \ref{sec_generators}. Note that, due to the analogy in the Hamiltonian part, the proofs of Lemmas \ref{lemma7}, \ref{lemma8st} in \ref{app:B} retrace their equivalent counterparts shown in Ref.~\cite{BenattiEtAl18}. The reason why we reported them here entirely is twofold: on the one hand, this permits us to introduce useful tools to subsequently prove the other Lemmas; on the other hand, some of the steps are in fact dependent on Lemma \ref{lemma_squarefluctu_convergence}, which involves the whole dynamical generator. Regarding the dissipative part of the latter, its main contribution to the map $\Phi_t^{\vec{\omega}}$ appear in the exponent $Y_t(\vec{\omega})$, via two type of terms, one dependent on $G^{\mathrm{sym}}$ and the other on $S^{\mathrm{sym}}$. The former arises as if the collective-operator valued function $\Gamma_\ell(\Delta^\ell_N)$ behaves as a mean-field rate with respect to a local dissipator [defined in Eq.~\eref{D_loc}], and it is also positive definite, as can be seen from the structure of the operator it originates from, $ [\vec{r}\cdot \vec{F}_t^N,  j_\ell^{\dagger\, (k)}][j_\ell^{(k)}, \vec{r}\cdot \vec{F}_t^N]$ (see \ref{lemma10}). The second one keeps instead track of the contribution of the quantity $\Gamma_\ell(\Delta^\ell_N)$ in the jump operators \eref{jumps} at the operatorial level, and its positive definiteness follows the one of the operator $j_\ell^{\dagger (k)}j_\ell^{(k)}$.

Furthermore, due to the similarity of the form of the map $\Phi_t^{\vec{\omega}}$ with the results in Ref.~\cite{BenattiEtAl18}, some additional considerations can be extended to our case. In particular, it is worth noticing that, because of the dependency of the map $\Phi_t^{\vec{\omega}}$ on the time-dependent variables $\vec{\omega}_t$, the map itself does not evolve Weyl operators onto themselves. Indeed, being the symplectic matrix time dependent via $\vec{\omega}_t$, one obtains time-evolving canonical commutation relations. To deal with quantum fluctuations generating an algebra dependent on time-dependent average operators, the quantum fluctuation algebra must be extended to a more general one. The resulting generator of the dynamical map is rather complicated, and we refer to Ref. \cite{BenattiEtAl18} for a detailed discussion and an example of such a dynamical generator. As a result, one is led to introduce a
quantum-classical hybrid dynamical system, made of both quantum fluctuations and classical degrees of freedom. 

In the following, in the same spirit of Ref. \cite{BonebergLC22, CarolloL_PRA_22, MattesLC_arxiv_23}, we consider a simpler case. First of all, we focus $(i)$ on states $\omega(\cdot)$ that are invariant for average operators. That is, the dynamics of quantum fluctuations we considers occurs with respect to states that are stationary with respect to the average operators. Under such a restriction, the map $\Phi_t^{\vec{\omega}}$ transforms Gaussian state into Gaussian state. Indeed, the evolved state $\Omega \circ \Phi^{\vec{\omega}}_t$ of a Gaussian state $\Omega$, with covariance matrix $\Sigma$, on the Weyl algebra, remains Gaussian, with the time evolved covariance matrix reads $\Sigma^{\omega}(t) = X_t(\vec{\omega}) \Sigma X_t^{T}(\vec{\omega}) + Y_t(\vec{\omega})$, and $Y_t(\vec{\omega}) \geq 0$. In the following we will mainly use the equation for the time derivative of the covariance matrix, that reads
\begin{equation}\label{eq_derivative_covaraince}
    \dot{\Sigma}^{\omega}_t = Q(\vec{\omega}) \Sigma^{\omega}_t + \Sigma^{\omega}_t Q^T(\vec{\omega}) + \sigma(\vec{\omega}) S^{\mathrm{sym}} \sigma^T(\vec{\omega}) + \Gamma^2(\vec{\omega}) G^{\mathrm{sym}} (\vec{\omega}) \, .
\end{equation}
As a second simplification, $(ii)$ we further focus  on a relevant set of quantum fluctuations (as exemplified in the next Section). This allows us to identify, within the quantum fluctuation operators, a set of commuting operators and a set of bosonic modes, enabling us to directly quantify classical and quantum correlations from the covariance matrix.

\section{Application to quantum Hopfield-like neural networks}
\label{sec5}

\subsection{Quantum Hopfield-like neural networks}
In this section, we will exploit the previous results to analyze some features of open quantum generalizations of Hopfield-like models. With Hopfield-like models we refer to certain type of  classical neural networks (NNs), behaving as associative memories \cite{Amit_book}. That is, systems that are capable of retrieving complete information from corrupted input data, following a given ``learning rule". More specifically, Hopfield neural networks (HNN) \cite{Hopfield:1982} are systems made of $N$ classical spins which feature all-to-all interactions \cite{Amit_book,AmitGS:1985a}, and they are described by the energy function $E=-\frac{1}{2}\sum_{i\neq j=1}^{N} w_{i j} \, z^{(i)} z^{(j)}$, where $z^{(i)}$ are the classical Ising spins. The connections, or interaction couplings, $w_{ij}$ embody the learning rule. Indeed, they are chosen in such a way that a set of $p$ spin configurations, $\lbrace \xi_{i}^{\mu} \rbrace_{i = 1,...,N}$ for $\mu=1,2,\dots p$, can be stored and retrieved by the network, in term of system configurations, via a gradient descent dynamics with respect to the energy function $E$. Among the different learning rules, widely known is the \emph{Hebb's prescription}, that sets $w_{i j} = \frac{1}{N}\sum_{\mu=1}^{p}\xi_{i}^{\mu} \xi_{j}^{\mu}$, and makes patterns minima of the energy function. In practice, patterns are chosen as independent identically distributed (i.i.d.) random variables that can assume the values $\xi_i^{\mu} = \pm 1$. In the regime $p/N \ll 1$, the spin configurations which have minimal energy are those in which all spins are aligned with the patterns. In presence of noise, which is introduced in the form of an inverse of a temperature parameter, $\beta$, the retrieval mechanism emerges when endowing the HNN with a Glauber thermal single spin-flip dynamics \cite{glauber1963}.

Quantum generalizations of HNNs have been introduced in Refs.~\cite{Rotondo:JPA:2018,FiorelliLM22}, that embed the classical systems into the more general framework of open quantum Markovian evolution defined by Eq.~\eref{Lindblad}-\eref{jumps}. In this way one could investigate the impact on quantum effects on the retrieval dynamics. In analogy with the mentioned works, we consider here $N$ spin-$1/2$ particles, undergoing a Markovian evolution with jump operators
\begin{equation}\label{hopfield rates}
J_{\pm}^{(k)} = \hat{\sigma}_{\pm}^{(k)}\Gamma_{\pm}^{\mathrm{HN}}(\Delta E), \qquad \Gamma_{\pm}^{\mathrm{HN}  } (\Delta E)= \frac{ e^{\pm \frac{\beta}{2}\Delta E }}{ \sqrt{2}} \, ,
\end{equation}
where $$ \Delta E= \frac{1}{N}\sum_{\mu=1}^{p}\xi^{\mu}_{k}\sum_{j} \xi_{j}^{\mu} \hat{\sigma}^{(j)}_{z} \, .$$ 
Here, the operators $\hat{\sigma}_{\alpha}^{(i)}$ identify Pauli operators, and we introduced the notation $\hat{\cdot}$ to avoid potential confusion with the symplectic matrix previously defined. It is worth noticing that $\Delta E$, which represents the energy difference associated with the configurations before and after the transition, quantifies the energy change associated with a spin-flip at site $j$, as well as the self-energy. Moreover, it is not a simple multiple of the identity but a many-body operator. 
The Hamiltonian term is selected to be the simplest operator that can compete with the dissipation, $$H=\Omega \sum_{i=1}^{N}\hat{\sigma}_{x}^{(i)} \, ,$$ i.e. it is a homogeneous transverse field.

There are two differences between the collective operator-valued functions defined above, and the ones employed in Refs.~\cite{Rotondo:JPA:2018,FiorelliLM22}. First of all, the latter do not account for the self-energy contribution in $\Delta E$, using instead
$\Delta E_k= \frac{1}{N}\sum_{\mu=1}^{p}\xi^{\mu}_{k}\sum_{j \neq k} \xi_{j}^{\mu} \hat{\sigma}^{(j)}_{z}$. Even though Ref.~\cite{FiorelliEtALC_NJP_23} shows the equivalence of the model (with $\Delta E_k$ and $\Delta E$), in the thermodynamic limit, the representation of the system dynamics by means of $\Delta E$ is more related to the results presented in the previous sections. As a second difference, Refs. ~\cite{Rotondo:JPA:2018,FiorelliLM22} take into account collective operator-valued rates such that $\Gamma^{2}_{+}(\Delta E_k)+ \Gamma^{2}_{-}(\Delta E_k) = 1$, considering $\Gamma_{\pm}^{\mathrm{HN}  } (\Delta E)/\sqrt{2\cosh{(\beta \Delta E)}}$. Instead, we adopt rates of the form \eref{hopfield rates}, which in turns represents operator-valued functions that are entire, consistently with our Assumption~\eref{Gamma}. This difference is taken into account in the following. For the sake of a lighter notation, in the rest of the Section we will use $\Gamma_{\pm}(\Delta E) \equiv \Gamma_{\pm}^{\mathrm{HN}}(\Delta E)$, omitting the label $(\cdot)^{\mathrm{HN}}$ of the operator-valued functions.

\subsection{Exactness of the mean-field approach and phase diagram}
\label{S:large-spin_Mapping}
To smoothly apply our results, we first explicitly express the operator-valued functions defined in Eq.~\eref{hopfield rates} as a linear combination of convenient average operators. To this end, we perform a mapping  \cite{KochP_JSP_89, GrensingK86, Gayrard92, CarolloL:PRL:21}, on the following quantity 
\begin{equation}
E=-\frac{1}{2}\sum_{i,j} w_{ij} \hat{\sigma}_z^{(i)} \hat{\sigma}_z^{(j)} = -\frac{1}{2 N}\sum_{\mu=1}^{p}  \left( \sum_{i=1}^{N}\xi_i^{\mu} \hat{\sigma}_z^{(i)} \right)^2,
\end{equation}
where in the last step we explicitly write $w_{ij}$ in terms of the patterns $\xi_i^{\mu}$. We anticipate that, through the mapping, the above energy can be expressed as a quadratic form of magnetization of certain block-spins. As the latter as are of macroscopic size for large $N$, we refer to them as large-spins. The mapping can be understood as a reordering of each $\mu$-th row given by the patter $(\xi_1^{\mu}, ...,\xi_{N}^{\mu}) $, which in turn corresponds to a $ \lbrace \sigma_z^{(i)} \rbrace_{i=1,..., N}$ spin configuration. As we consider $\{ \xi^{\mu}_{i} \}$ i.i.d. random variables, the first pattern, $\xi_i^{1}$ takes the values $\pm 1$ at random positions. We proceed with identifying a partition in  the spins systems, that is, there will be a set of spins with $\xi_{j}^1 = +1$, and a set of spins with $\xi_{j'}^1 = -1$. The first [second] ones are taken to the left [to the right]. Thus, we have identified a label $\tilde{h}$ such that $\xi_h^1=1$ for $h\leq \tilde{h}$, and $\xi_h^1=-1$ otherwise. We repeat the above operation in each subset corresponding to $\xi_h^1=\pm 1$, performing the partition  with respect to the second pattern, $\xi_{i}^{2}$. Therefore, in the subset corresponding to $\xi_h^1=1$ we relabel the spins such that $\xi_i^{2}=1$ [$\xi_i^{2}=-1$] are moved to the left [right], and the same is done for the subset with $\xi_i^{2}=-1$. This procedure is then repeated up to the last pattern. As a result, for large $N$, the mapping yields $2^p$ subset of spins, say $\Lambda_{k}$ for $k=1,...,2^p$. Being $\xi_i^{\mu}$ i.i.d.~random variables, and so long as $N \gg 1 $, each pattern $(\xi_1^{\mu},...,\xi_{N}^{\mu})$  contains, at leading order, an equal number of $+1$ and $-1$. Thus each one of the $2^p$ subsets has at leading order the same number of spins, $N_{\mathrm{s}} \sim N / 2^p$ (assuming that $N/2^p$ is an integer number). Remarkably, under this mapping, the energy function reads, at the leading order
\begin{equation}
E \sim -\frac{1}{2}\sum_{h,k=1}^{2^p N_{\mathrm{s}}}\tilde{w}_{hk} S^{(h)}_{z}S^{(k)}_z \sim -\frac{1}{ 2^{p+1}N_{\mathrm{s}}}\sum_{\mu=1}^{p}\left( \sum_{h=1}^{2^{p}}e_{h}^{\mu} S^{(h)}_z \right)^2 \, ,
\end{equation}
It describes the interaction between large-spin operators $S^{(h)}_z$, where $S^{(h)}_{\alpha}= \sum_{i \in \Lambda_{h}} \hat{\sigma}^{(i)}_{\alpha}$ is given by the sum of spin-$1/2$ operators belonging to the $h$-th subset $\Lambda_h$. Most notably, the large-spin interaction couplings are defined as $\tilde{w}_{hk}=\frac{1}{ 2^{p}} \sum_{\mu=1}^{p} e_{h}^{\mu}e_{k}^{\mu}$, 
where  the coefficients $e_{h}^{\mu}$, which can assume the values $\pm 1$, represent the pattern values for spins in the subset $\Lambda_h$.  Furthermore, for each collection of spin $1/2$ belonging to the set $\Lambda_k$, the operator $\Delta E$  reads
\begin{equation}
\Delta E_{ \Lambda_k}^{N_{s}} = \frac{1}{ N_{\mathrm{s}}} \sum_{h=1}^{2^p}\tilde{w}_{hk} S^{(h)}_{z} =  \sum_{h=1}^{2^p} \tilde{w}_{hk} \, m_{z,h}^{N_{\mathrm{s}}}
\end{equation}
where%
\begin{equation}\label{macro_spin_average}
m_{\alpha,k}^{N_{\mathrm{s}}} \equiv \frac{S^{(k)}_{\alpha}}{N_{\mathrm{s}}} \, ,
\end{equation}
for $\alpha=x,y,z,$ and $k=1,...,2^p$,  are the average magnetization operators. Through this mapping, it becomes clear how to derive the average operator description that we have employed in the previous sections. Moreover, it is now evident that the collective operator-valued rates depend on linear combinations of the average magnetization operators, and satisfy Assumption~\eref{Gamma}. As such,  Theorems \ref{theorem} can be applied to the quantum generalization of the HNN dynamics, which guarantees the validity if the mean-field approximation when choosing initial states with short correlations.

By exploiting the latter results, in the thermodynamic limit and for $p/N \ll 1$, the quantum generalized HNN can be analyzed via the dynamical evolution of mean-field variables associated to the average magnetization operators.
Given the slightly modified choice of the operator valued-functions with respect to Ref. \cite{Rotondo:JPA:2018}, for the sake of completeness, we performed the analysis of the phase diagram, leaving the details of the derivation in Appendix. The retrieval properties of quantum HNNs in the parameter regimes $(\Omega, \beta)$, are summarized in the following. In fact, to highlight regimes where any of the patterns can be stored at stationarity,  the proper mean-field variables to take into account are referred to as \textit{overlap} variables, reading
$$o_{\alpha}^{\mu}(t) = \lim_{N\rightarrow + \infty} \frac{1}{N}  \sum_{i=1}^{N} 
     \xi_{i}^{\mu} \omega_t(\hat{\sigma}_{\alpha}^{(i)})  = \lim_{N\rightarrow + \infty}  \frac{1}{2^p}\sum_{k=1}^{2^p} e_k^{\mu}\omega_t(m_{\alpha,k}^{N_s}) \, .$$ 
The overlaps can be understood as a total magnetization where each large-spin (related to the set $\Lambda_k$) is weighted with respect to the the pattern components ${e_{k}^{\mu}}$. Indeed, the perfect retrieval of the $\mu$-th pattern [anti-pattern] in the $\alpha$ direction is achieved when $o_{\alpha}^{\mu} = 1 [-1]$, whereas $o_{\alpha}^{\mu} = 0$ signals a random configuration with respect to the $\mu$-th pattern. 

The phase diagram of the model is shown in Fig.~\ref{fig:phaseD}, which displays different parametric regimes. These correspond to different stable stationary solutions of the mean-field equations for the overlap variables. The details of the derivation are left in \ref{C:HNN_mean Field}. For large temperatures there is a unique stable stationary solution, featuring $o_{\alpha}^{\mu}=0$, $\forall \alpha, \mu$. The quantum HNN is said to display here a paramagnetic or disordered phase, for which pattern retrieval is not possible. As the temperature is decreased, and for sufficiently small values of the transverse field strength $\Omega$, the system shows instead a ferromagnetic, or retrieval, phase characterized by $o_{\alpha}^{\bar{\mu}} \neq 0$, for the $\bar{\mu}$-th pattern. The system can therefore operate here as an associative memory. For large values of $\Omega$ and a given range of temperature, a region characterized by a limit-cycle phase can be observed. In particular, such a region takes place for $\Omega > 1/4$, $\beta>3/2$, $\beta>(1+8\tilde{\Omega}^2)$, with $\tilde{\Omega}$ defined in Eq.~\eref{c:statRet}. Here, the state of the system features a nonzero overlap with one of the pattern, and has been dubbed in Ref.\cite{Rotondo:JPA:2018} as a type of retrieval phase. 

\begin{figure}[h]
    \centering
    \includegraphics[width=0.6\textwidth]{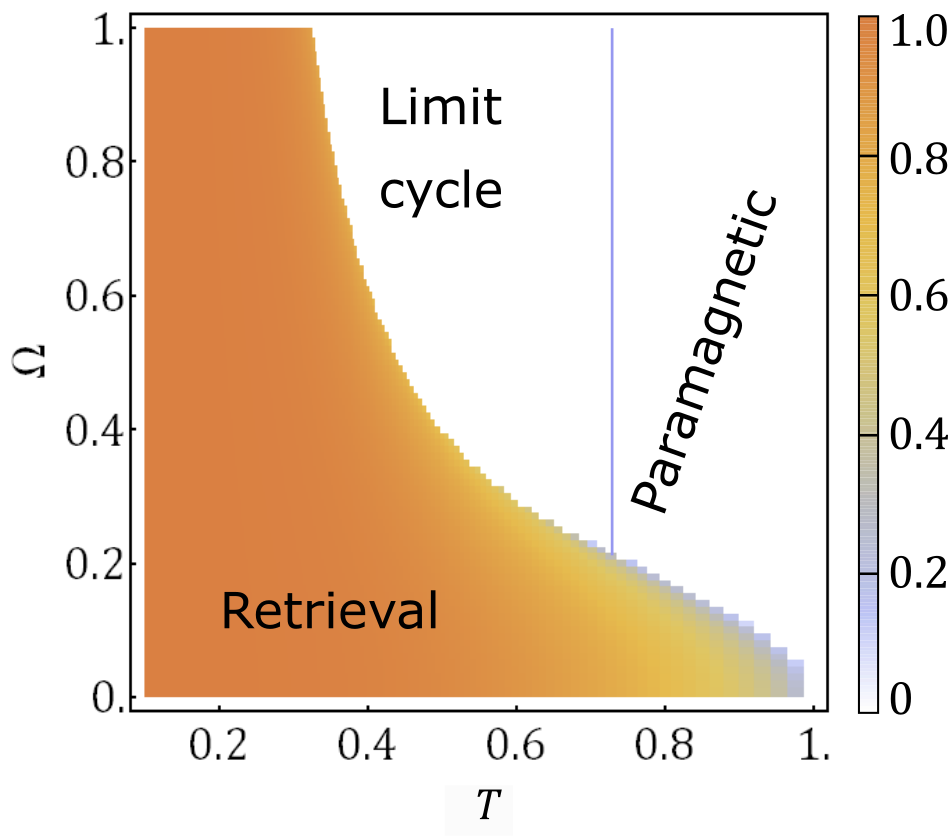}
    \caption{Phase diagram of the quantum Hopfield-like model.  At high temperature the mean-field equations for the overlaps admits a unique, vanishing solution, this corresponding to the paramagnetic phase. At low temperature the stationary solutions are characterized by finite overlap values (colorbar). This phase is referred to as retrieval phase. The (white) region between the paramagnetic and the retrieval one is referred to as a limit-cycle phase, and it displays self-sustained oscillations at stationarity.}
    \label{fig:phaseD}
\end{figure}

\subsection{Quantum fluctuations and covariance matrix}

The goal of this section is understanding and quantifying the extent of quantum correlations that can emerge at the mesoscopic level in the quantum generalization of the HNN introduced in the previous section. We will do so by means of the quantum fluctuations operators
\begin{eqnarray}
   F_{\alpha, k}^{N_s} = \sqrt{N_s}( m_{\alpha, k}^{N_s}  - m_{\alpha, k}(t)) \, .
\end{eqnarray}
and of the time-dependent covariance matrix
\begin{equation}
   \Sigma_{(\mu, h), (\nu,k)}^{\omega}(t) = \lim_{N\rightarrow + \infty} \omega_t \left( \left\lbrace F_{\mu,h}^{N_s}(t), F_{\nu,k}^{N_s}(t) \right\rbrace \right) \, .
\end{equation}
focusing on the behaviour of the correlations between large-spins of different sets $\Lambda_k$, each one with average magnetization operator $m_{\alpha,k}^{N_s}$, and mean-field variables $m_{\alpha,k}(t)$.

As already anticipated at the end of Section \ref{sec4}, we consider the dynamics of quantum fluctuations with respect to an initial state that is stationary with respect to the mean-field variables $m_{\alpha,k}(t)$. In this way we can analyze the faith of quantum correlations in the parametric regime where retrieval and the paramagnetic phases take place.

By exploiting Theorem \eref{theorem-mesoscopicdynamics}, we can thus write the time evolution of the covariance matrix, which is given as
\begin{equation}
\dot{\Sigma}^{\omega} = Q\Sigma^{\omega}+\Sigma^{\omega}Q^T+\sigma S^{\mathrm{sym}}\sigma^T + \Gamma^{2}G^{\mathrm{sym}} \, .
\end{equation}
The step by step derivation, whose details are left in \ref{A:covarianceM_HNN}, allows us to obtain the explicit form of the matrices $Q, S^{\mathrm{sym}}, \Gamma^2 G^{\mathrm{sym}}$. First of all, by means of the equation of motions [see Eq.~\eref{eq:mf-largespinHNN}] for the mean-field variables $m_{\alpha,k}(t)$, the matrix $Q$ reads
\begin{eqnarray}
    [Q]_{(\alpha \, h), (\beta \, k)} & = -2\Omega \mathcal{E}_{x\alpha \beta} \delta_{hk}  \\ \nonumber 
    & + \delta_{hk} \cosh(\beta \Delta E_{\Lambda_h})  \delta_{\alpha \beta} \left(-\frac{1}{2}\delta_{x\alpha}-\frac{1}{2}\delta_{y\alpha} -\delta_{\alpha z}\right) \, , \\ \nonumber
    [Q]^{T}_{(\alpha \, h), (\beta \, k)} & = +2\Omega \mathcal{E}_{x\alpha \beta} \delta_{h k}  \\
    & + \delta_{hk} \cosh(\beta \Delta E_{\Lambda_h}) \delta_{\alpha \beta} \left(-\frac{1}{2}\delta_{x\alpha}-\frac{1}{2}\delta_{y\alpha} -\delta_{\alpha z}\right) \, .\\ \nonumber
\end{eqnarray}
By exploiting the expression of the constant of the algebra $a_{\alpha \beta}^{\gamma} = 2i\mathcal{E}_{\alpha \beta \gamma} $, and having defined the local jump operators as $j_{\pm}^{(k)} = \hat{\sigma}_{\pm}^{(k)}$, one obtains
\begin{equation}
\eqalign{
[\Gamma^2 G^{\mathrm{sym}} ]_{(\alpha h)(\beta k)}& = \delta_{hk}\cosh(\beta \Delta E_{\Lambda_h}) \left\lbrace \delta_{\alpha x}\delta_{x\beta} +\delta_{\alpha y}\delta_{y\beta} \right. \\
& \left. + 2\delta_{\alpha z}\delta_{z\beta}\right[ 1-\tanh(\beta \Delta E_{\Lambda_h}) m_{z,h}\left]  \right. \\
& - \delta_{\alpha x}\delta_{z \beta}\tanh(\beta \Delta E_{\Lambda_h}) m_{x,h} - \delta_{\alpha y}\delta_{z \beta} \tanh(\beta \Delta E_{\Lambda_h}) m_{y,h}\\
& \left. - \delta_{\alpha z}\delta_{x \beta} \tanh(\beta \Delta E_{\Lambda_h})  m_{x,h} - \delta_{\alpha z}\delta_{y \beta} \tanh(\beta \Delta E_{\Lambda_h})  m_{y,h} \right\rbrace \, ,
}
\end{equation}
We write the term $\sigma S^{\mathrm{sym}} \sigma^{T}= \sigma  (  S^{\mathrm{sym}}_1 +  S^{\mathrm{sym}}_2) \sigma^T $, with
\begin{equation}
\eqalign{
 [\sigma S^{\mathrm{sym}}_1 \sigma^T]_{(\alpha h), (\beta k)} & =  \frac{\beta^2}{8}\sigma_{\alpha z (h)} \\ & \times \sum_{k'} [ \cosh{(\beta \Delta E_{\Lambda_{k'}})} - m_{z k'}(t)  \sinh{(\beta \Delta E_{\Lambda_{k'}})}]\tilde{w}^{hk'}\tilde{w}^{k' k}\sigma_{\beta z (k)} \, .
}
\end{equation}
\begin{equation}
    \eqalign{
    [\sigma S^{\mathrm{sym}}_2 \sigma^T]_{(\alpha h), (\beta k)} 
    & = -\frac{\beta}{8} \sigma_{\alpha z (h)} \tilde{w}_{hk} [\cosh(\beta \Delta E_{\Lambda_h })+\cosh(\beta \Delta E_{\Lambda_k }) ] \sigma_{\beta z (k)} 
    }
\end{equation}
Collecting all the above terms together, the EoM for the covariance matrix reads
\begin{equation}\label{EoM_covariancematrix}
    \eqalign{
    \dot{\Sigma}_{(\alpha, h), (\beta,k)}^{\omega} & = \sum_{\mu, k'} \left\lbrace -2\Omega \mathcal{E}_{x\alpha \mu} \delta_{hk'} \right.  \\  
    & \left. + \delta_{hk'} \cosh(\beta \Delta E_{\Lambda_h}) \delta_{\alpha \mu} \left(-\frac{1}{2}\delta_{x\mu}-\frac{1}{2}\delta_{y\mu} -\delta_{\mu z}\right) \right\rbrace\Sigma^{\omega}_{(\mu k') (\beta k)} \, + \\    
    &+\sum_{\mu k'} \Sigma^{\omega}_{(\alpha h) (\mu k')} \left\lbrace2\Omega \mathcal{E}_{x\mu \beta} \delta_{k' k}  \right. \\
    & + \left. \delta_{k'k} \cosh(\beta \Delta E_{\Lambda_k}) \delta_{\mu \beta} \left(-\frac{1}{2}\delta_{x\mu}-\frac{1}{2}\delta_{y\mu} -\delta_{\mu}\right) \right\rbrace+ \\
    & + \delta_{hk} \cosh(\beta \Delta E_{\Lambda_h}) \left\lbrace\delta_{\alpha x}\delta_{x\beta} +\delta_{\alpha y}\delta_{y\beta} \right. \\
    & \left. + 2\delta_{\alpha z}\delta_{z\beta}\right[ 1-\tanh(\beta \Delta E_{\Lambda_h}) m_{z,h}\left]  \right\rbrace \\
    & - \delta_{\alpha x}\delta_{z \beta} \sinh(\beta \Delta E_{\Lambda_h}) m_{x,h} - \delta_{\alpha y} \delta_{z \beta} \sinh(\beta \Delta E_{\Lambda_h}) m_{y,h}\\
    &- \delta_{\alpha z}\delta_{x \beta} \sinh(\beta \Delta E_{\Lambda_h}) m_{x,h} - \delta_{\alpha z}\delta_{y \beta} \sinh(\beta \Delta E_{\Lambda_h}) m_{y,h}\\
    & + \frac{\beta^2}{8} \sigma_{\alpha z (h)}\sum_{k'} [ \cosh(\beta \Delta E_{\Lambda_{k'}}) - m_{z, k'}  \sinh(\beta \Delta E_{\Lambda_{k'}})]\tilde{w}^{hk'}\tilde{w}^{k' k}\sigma_{\beta z (k)} \\ 
    &-\frac{\beta}{8} \sigma_{\alpha z (h)} \tilde{w}_{hk} [\cosh(\beta \Delta E_{\Lambda_h})+\cosh(\beta \Delta E_{\Lambda_k}) ] \sigma_{\beta z (k)} \, .
    }
\end{equation}
We can now comment upon the structure of the above equation, and eventually obtain the general form of the asymptotic covariance matrix. First of all, we can see that the matrix $Q$ cannot couple different blocks $\Sigma_{(\alpha,\bar{h})(\beta,\bar{k})}^{\omega}$ for fixed $\bar{h} \neq \bar{k}$. Thus, the competition between the Hamiltonian -- consisting of a transverse field with coherent strength $\Omega$ -- and the dissipation, -- with collective operator-valued rates, depending only on $m_{z, k}^{N_{s}}$--, cannot contribute to coupling quantum correlations among different large-spins. Similarly, the constant term, proportional to $\Gamma^2 G^{\mathrm{sym} }$ remains diagonal with respect to large-spin components. Finally, the remaining terms originating from $\sigma S^{\mathrm{sym} } \sigma$ do generate off-diagonal contributions, i.e. matrices $\Sigma^{\omega}_{(\alpha,\bar{h})(\beta,\bar{k})}$, for fixed $\bar{k} \neq \bar{h}$, each featuring only one non-zero element, corresponding to $\sigma_{\alpha z (h)}\sigma_{\beta z(k)} \approx m_{y,k}m_{y,h} $ with $\alpha=\beta=x$. Thus, they are characterized by $\det [\Sigma_{(\alpha,\bar{h})(\beta,\bar{k})}] = 0$, which we will use in the next section. By collecting the above observations, the asymptotic covariance matrix has the following structure 
\begin{equation}
    \left( 
    \matrix{
    \Sigma_{11}^{\mathrm{As}} && \matrix{ C_{12} && 0 && 0  \cr
    0 && 0 && 0  \cr
     0 && 0 && 0 } && \matrix{\cdots} && \matrix{C_{ 1 n } && 0 && 0  \cr
    0 && 0 && 0  \cr
     0 && 0 && 0 }\cr
    \matrix{C_{12} && 0 && 0  \cr
    0 &&  0 && 0 \cr
    0 && 0 && 0}  && \Sigma_{22}^{\mathrm{As}}  && \matrix{\cdots} && \matrix{\vdots} \cr
     \matrix{\vdots} && \matrix{\vdots} && \matrix{\ddots} \cr
     \matrix{C_{ 1 n } && 0 && 0  \cr
    0 && 0 && 0  \cr
     0 && 0 && 0 } && \matrix{\cdots} && \matrix{\cdots} &&  \Sigma_{nn}^{\mathrm{As}}
    }
    \right) \,
\end{equation}
where $n=2^p$.
%
%
%
One can compare the above structure with the classical case, obtained by setting $\Omega = 0$, and featuring $m_{y,k} = m_{x,k} = 0$ $\forall k$. In this case, the covariance matrix is diagonal, this implying vanishing two-point correlations among the large-spins structures $\Lambda_k$. 

In the next section, we focus on the case $p=1$ in order to quantify, with an example, the above considerations. The one-memory case correspond to partitioning the system of $N$ spin-$1/2$ particles into two large-spins of size $N/2$. Before going ahead, we point out that the results that we discuss in the next section can be derived also for the two-memory case, $p=2$, upon changing the large-spin mapping described in Section \ref{S:large-spin_Mapping}. Indeed, in the first step of the large-spin mapping, one can perform a gauge transformation $\sigma_z^{(i)} \rightarrow \xi_{i}^{1}\sigma_z^{(i)}$, $\xi_{i}^{\mu} \rightarrow \xi_{i}^{1} \xi_{i}^{\mu}$, which has the effect to align all the first pattern components, such that $e_{k}^{1} = 1$ $\forall k$. Then, the mapping proceeds as described in the Section \ref{S:large-spin_Mapping}, and is equivalent to reordering the spin $1/2$-particles. As a result, there will be $2^{p-1}$ large-spins subsets $\Lambda_k$ for $p$ patterns. That is, the $p=2$ case coincides with two large-spins of size $N_s=N/2$.

\subsubsection{Classical and quantum correlations for one pattern case.}
In this subsection, we now explore the special case of two large-spins, $h=1,2$, in order to get quantitative results upon the behavior of the classical and quantum correlations. In doing so, we remind that, $i)$ we focus on the dynamics of quantum fluctuations arising, in the thermodynamic limit, from an initial state that is \textit{stationary} with respect to the mean-field operators. As such, $ii) $ we notice that the mean-field equations for the model, (featuring $m_{x,k} = 0$, as the the equation of motion of $m_{x,k}$ decouples from the rest), satisfy the symmetry $m_{\alpha, 2} = -m_{\alpha, 1}$ for $\alpha = y,z$. As a result, at stationarity, the mean-field magnetizations of the two large-spin are anti-aligned. 

The equation of motion for ${\Sigma}^{\omega}$, upon the conditions $i)$, $ii)$, $\dot{{\Sigma}}^{\omega} = Q {\Sigma}^{\omega} + {\Sigma}^{\omega}Q^{T} + \Gamma^{2}G^{\mathrm{sym}} + \sigma S^{\mathrm{sym}} \sigma^T $, features
\begin{equation}
\small{
Q =
\left(  
\begin{array}{r@{\quad}cr} 
1 & 0 & \cr
0 & 1 &  \cr
\end{array} 
\right)
\otimes 
\left(  
\begin{array}{r@{\quad}cr} 
-\frac{1}{2}\cosh(\Delta) & 0 & 0  \cr
0 & -\frac{1}{2}\cosh(\Delta) & -2\Omega   \cr
0 & 2\Omega & -\cosh(\Delta) \cr 
\end{array} 
\right)
}
\end{equation}

\begin{equation}
\small{
\Gamma^2G^{\mathrm{sym}} = 
\left(  
\begin{array}{r@{\quad}cr} 
1 & 0 & \cr
0 & 1 &  \cr
\end{array} 
\right)
\otimes
\left(  
\begin{array}{r@{\quad}cr} 
1 & 0 & 0  \cr
0 & 1 & -\tanh(\Delta)m_{y,1}(t) \cr
 0 & -\tanh(\Delta)m_{y,1} & 2 \left[ 1-\tanh(\Delta) m_{z,1}\right]  \cr \cr
\end{array} 
\right)
}
\end{equation}
\begin{equation}
\small{
\sigma S^{\mathrm{sym}} \sigma^{T} = \beta \left[\frac{\beta}{2}m_{z,1}\sinh(\Delta) + (\frac{\beta}{2}-1)(\cosh(\Delta)) \right]
\left(  
\begin{array}{r@{\quad}cr} 
1 & 1 & \cr
1 & 1 &  \cr
\end{array} 
\right)
\otimes
\left(  
\begin{array}{r@{\quad}cr r@{\quad}cr} 
m_{y 1}^2 & 0 & 0  \cr
0 & 0 & 0 &  \cr
0 & 0 & 0 \cr
\end{array} 
\right)
}
\end{equation}
where we employed $\tilde{w}_{hh} =\frac{1}{2} $,  $\tilde{w}_{h \neq k} = - \frac{1}{2} $, $\tilde{w}^{hk'}\tilde{w}^{k'k} = \frac{1}{2}\tilde{w}^{hk}$, and $\Delta E_1 =- \Delta E_2 \equiv \Delta =\frac{ m_{z1} - m_{z2}}{2} = m_{z1}$. We have furthermore used $\sigma_{\alpha \beta (h)}^{\omega} = 2 \mathcal{E}_{\alpha \beta \gamma} m_{\gamma, h}$. As already point out in the more general case, when setting $\Omega = 0$, and noticing that $m_{x,k} = m_{y, k} = 0$, the evolution of the covariance matrix is ruled by diagonal matrices.

We now show that for each of the two large-spins $\Lambda_k$, it is possible to reduce the set of quantum fluctuations operators $F_{\alpha, k}^{N_s}$ for $\alpha = x,y,z$ to a pair of emergent bosonic modes. This follows from the fact that, at stationarity, the mean-field magnetization of the two large-spins are anti-aligned: we can rotate the reference frame by aligning the $z$ direction with respect to the mean-field variables direction. As a result, in the rotating frame $\tilde{m}_{z, k} \neq 0$, and $\tilde{m}_{x/y, k} = 0$. The transformation is thus performed by rotating the reference frame by an angle $\theta$ with respect to the $x$ axis, with $\theta = \mathrm{arcos}(\frac{m_{z,1}}{|\vec{m}|})$ and $|\vec{m}| = \sqrt{m_{z,1}^2 + m_{y,1}^2}= \sqrt{m_{z,2}^2 + m_{y,2}^2}$. The transformation reads
\begin{equation}
U=\left( \matrix{
1 & 0 & 0  \cr
0 & \frac{m_{z,1}}{|\vec{m}|}  & -\frac{m_{y,1}}{|\vec{m}|} \cr
0 & \frac{m_{y,1}}{|\vec{m}|}  & \frac{m_{z,1}}{|\vec{m}|} \cr 
}\right) \, .
\end{equation}
Applying the latter on the quantum fluctuation operators, we obtain $\tilde{F}^{N_s}_{\alpha, k} \equiv U_{\alpha \beta} F^{N_s}_{\beta, k}$, the only non-zero commutator of the symplectic matrix remaining
$$\tilde{\sigma}_{x y (k)} = 2 \tilde{m}_{z,k} \, .$$
Thus, this matrix reproduces the commutators of the Bose operators $\tilde{F}_{\alpha,k}$ obtained as mesoscopic limit of the fluctuation operator $\tilde{F}_{\alpha,k}^{N_s}$. In order to get canonical commutation relation, $iii)$ we perform a re-scaling via the transformation
\begin{equation}
    R = \left( \matrix{ 
    \frac{1}{\sqrt{2 |\vec{m}|}} & 0 & 0  \cr
    0 & \frac{1}{\sqrt{2 |\vec{m}|}} & 0  \cr
    0 & 0 & 1  \cr
    }\right) \, ,
\end{equation}
such that $R_{\alpha \beta} \tilde{F}_{\beta, k} = r_{\alpha k}$, for $k=1,2$, satisfying $[r_{x,k}, r_{y,k}] = i$. We will thus focus on the behavior of the quantum fluctuations $\mathcal{R} \equiv (r_{x,1}, r_{y,1}, \tilde{F}_{z,1}, r_{x,2}, r_{y,2}, \tilde{F}_{z,2})^{T} $, in particularly referring to the covariance matrix $\tilde{\Sigma}_{(\mu h) (\nu k)}^{\omega} = \lbrace \mathcal{R}_{\mu}, \mathcal{R}_{\nu} \rbrace / 2$. It is worth noticing that, starting from the vector of quantum fluctuation $ \vec{F} = ( F_{x,1}, F_{y,1} , F_{z,1}, F_{x,2}, F_{y,2} , F_{z,2}  )^T$ we obtain $\mathcal{R}$ as
\begin{equation}
    \mathcal{R} = ( \mathbb{I}_2  \otimes R) (\mathbb{I}_2 \otimes U  ) \vec{F} \, .
\end{equation}
Given the condition $i)$, the equation of motion for the covariance matrix $\tilde{\Sigma}^{\omega}$ can be obtained from the one derived in the original frame, $\Sigma^{\omega}$ [and given by Eq. \eref{EoM_covariancematrix}] as
\begin{eqnarray}
    \dot{\tilde{\Sigma}}^{\omega}  & = ( \mathbb{I}_2  \otimes R) (\mathbb{I}_2 \otimes U  )  \dot{{\Sigma}}^{\omega}  (\mathbb{I}_2 \otimes U^T  )( \mathbb{I}_2  \otimes R)   \, \\ \nonumber 
    & = \tilde{Q} \tilde{\Sigma}^{\omega} + {\tilde{\Sigma}}^{\omega}\tilde{Q}^{T} + \tilde{\Gamma}^{2} \tilde{G}^{\mathrm{sym}} + \tilde{\sigma} \tilde{S}^{\mathrm{sym}} \tilde{\sigma}^T \, .
\end{eqnarray}
By solving the equation of motion at stationarity we obtain the asymptotic covariance matrix $\tilde{\Sigma}_{\infty}^{\omega}$, from which we can discard the information on the irrelevant quantum fluctuations $\tilde{F}_{z,k}$, $k=1,2$.  Neglecting the third and sixth rows, as well as the third and sixth columns, we get the corresponding two-mode covariance matrix $\Sigma^{\omega}_2$, having the form
\begin{equation}
    2\Sigma^{\omega}_2 =  \left(\matrix{
    a_1 & c_{12} \cr
    c_{12}^T & a_2
    }
    \right) \, .
\end{equation}
We are now in the position of getting information upon the Gaussian quantum discord and a related measure of classical correlations. To do so, we follow Ref.\cite{AdessoDPRL10}. By defining 
\begin{eqnarray}
    A = \mathrm{det}(a_1) \, ,\quad B= \mathrm{det}(a_2) \, , \quad C = \mathrm{det}(c_{12}) \, , \quad D= \mathrm{det}(2\Sigma^{\omega}_2) \, , 
\end{eqnarray}
for a two mode covariance matrix the one-way classical correlations and one-way quantum discord \footnote{The $()^{{1 \leftarrow 2}}$ one-way correlations is obtained performing a measurement on the system labeled by $2$. } are defined, respectively, as
\begin{eqnarray}
    \mathcal{J}^{1 \leftarrow 2} = g(\sqrt{A}) - g(\sqrt{E_{\min}}) \, ,\\
    \mathcal{D}^{1 \leftarrow 2} = 
    g(\sqrt{B}) - g(\nu_{-}) - g(\nu_{+}) + g(\sqrt{E_{\min}}) \, ,
\end{eqnarray}
where 

\scalebox{0.8}
{$
    {E_{\min} = \left\lbrace 
    \eqalign{ & \frac{2C^2 + (B-1)(D-A) + 2|C| \sqrt{C^{2} + (B-1)(D-A)}}{(B-1)^2} \quad \mathrm{if} \, (D-AB)^2 \leq (1+B)C^2(A+D) \\
    & \frac{AB-C^2+D-\sqrt{C^4+(D-AB)^2-2C^2(AB+D)}}{2B} \quad \mathrm{otherwise}
    }
    \right.} \, , \\
$
} 
and the function $g(\cdot)$ is defined as
$$g(x) = \left(\frac{x+1}{2} \right)\ln{\left(\frac{x+1}{2} \right)} - \left(\frac{x-1}{2} \right) \ln{\left(\frac{x-1}{2} \right)} \, .$$
Finally, $\nu_{\pm}$ are the symplectic eigenvalues
of the matrix $2\Sigma^{\omega}_2$, found as
the positive eigenvalues of the matrix $2i\sigma \Sigma^{\omega}_2$, 
$$\nu_{\pm}^2 = \frac{1}{2} (A+B+2C) \pm  \frac{1}{2}\sqrt{(A+B+2C)^2 - 4D} \, ,$$
with $\nu_{\pm} \geq 1$. Collective entanglement between the two large-spins can be classified via logarithmic negativity, $ \mathcal{N} = \max{(0, -\ln(\tilde{\nu}_{-}))} $, where $\tilde{\nu}_{-}$ is the smallest symplectic eigenvalue of the partially transposed covariance matrix $ 2\Sigma^{\omega}_2 $. Equivalently, it can be obtained from $\tilde{\nu}_{-}$ by replacing $C$ with $-C$,  which corresponds to perform a time-reversal operation. In general, a state with covariance matrix $2\Sigma^{\omega}_2 $ will be entangled iff $\tilde{\nu}_{-}<1$. However, in this case it is $C=0$, as we could expect from the consideration performed on the asymptotic form of the covariance matrix. Thus $\tilde{\nu}_{-}\geq 1$ in all the parameter regimes $\Omega, \beta$, yielding a zero-valued logarithmic negativity $\mathcal{N}=0$.

Nonetheless, we can quantify correlations among the two large-spins, as given in terms of quantum discord and classical correlations, which are displayed in Fig.~\ref{fig:QCcor}, (a) and (b), respectively. The parameter regime $(T,\Omega)$ considered, corresponds, with respect to the mean-field magnetization, to the regimes where paramagnetic phase (high-temperature) and the retrieval phase (low-temperature) take place. We can appreciate a weak presence of both type of correlations, both increasing as the transition between paramagnetic phase, retrieval phase, and limit-cycle one is approached. Notice however that the presence of quantum discord is significantly smaller than classical correlations. We can conclude that, in the bulk of retrieval phase, as well as in the paramagnetic one, the system shows very weak correlations, and thus almost no footprint of quantumness is left, asymptotically and at the mesoscopic scale.

The situation could change, as outlined in \ref{A:covarianceM_HNN}, by adding a direct,  all to all, interaction to the model. In this case, some symmetries at the mean-field equation level are lost, this preventing us, presently, to gain quantitative results using the techniques adopted.

\begin{figure}
    \centering    \includegraphics[width=0.9\linewidth]{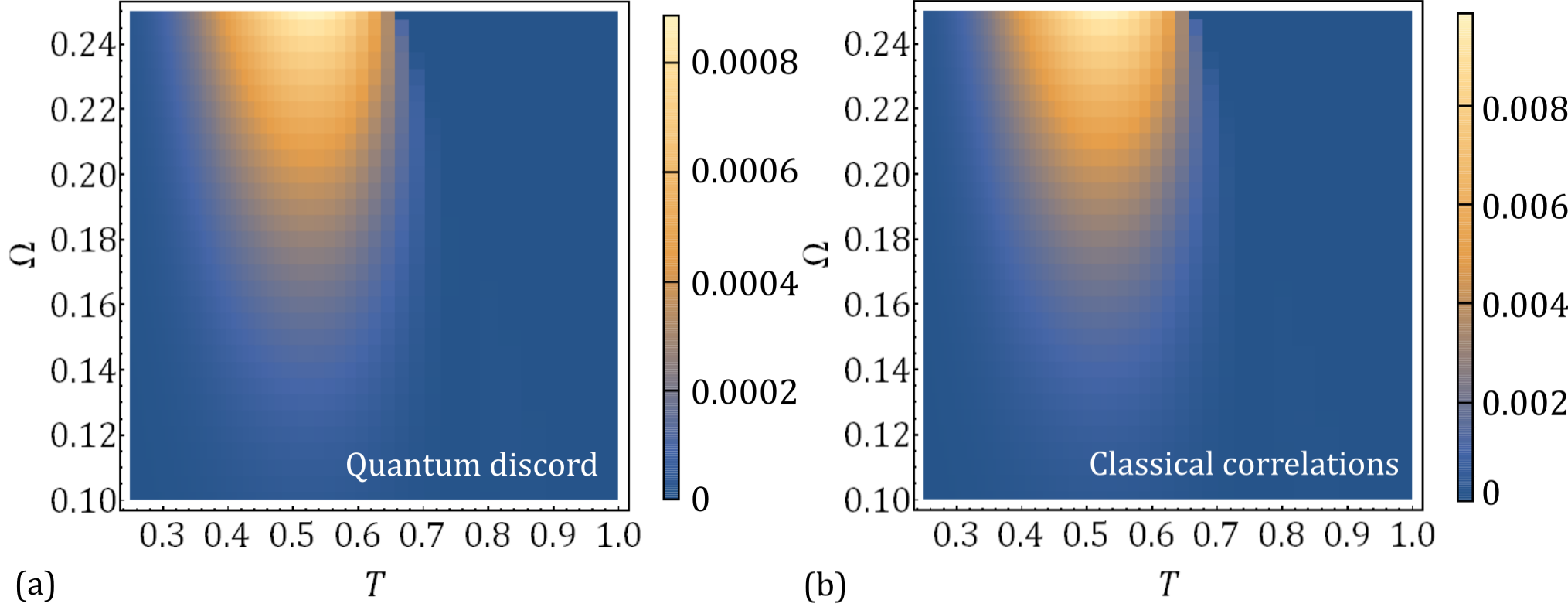}
    \caption{\textbf{Quantum and classical correlations.} (a) One way quantum discord, and (b) one-way classical correlations between the two large-spin operators of the one-memory case, $p=1$. In the displayed parameter regime, the system-average operators identify a retrieval region (low temperature) and a paramagnetic one (high temperature).}
    \label{fig:QCcor}
\end{figure}

\section{Conclusions}
In this work we have considered Markovian open quantum systems, describing their dynamics by means of a GKS-Lindblad master equation. As a most important feature, our model is characterized by a dissipation that involves collective operator-valued rates. To investigate dynamical and stationary properties of these models, we employed tools from algebras of operators, that can be found in the first part of the manuscript. Regarding the model, while focusing on average operators, we first reviewed the results on the validity of the mean-field approximation, in the thermodynamic limit. Building on these outcomes, we focused on the mesoscopic scale, deriving the dynamical maps evolving quantum fluctuation operators. The latter has been found to be a dissipative map that in general would mix, in a complicated manner, quantum and classical degrees of freedom. To simplify the analysis, we considered states that are invariant with respect to average operators, and subsequently obtained the equation ruling the evolution of the covariance matrix. We then applied these general results to the concrete example of an open quantum generalization of HNNs, employing a model similar to Ref.\cite{Rotondo:JPA:2018, FiorelliEtAl20}. Here, considering an infinite system size, and in the limit of vanishing storage capacity, we analyzed the asymptotic behavior of the average operators as described by the mean-field equations. In mapping the out-of-equilibrium phase diagram, one can appreciate phases that are the result of the competition between the coherent term and the dissipative one of the GKS-Lindblad generator. Furthermore, we analyzed the behavior of the asymptotic covariance matrix. This allows us to tackle the yet unexplored question of characterizing the presence of (quantum) correlations in this type of quantum generalized HNNs. In  fact, by means of the asymptotic covariance matrix,  we show that, beyond a small amount of classical correlations, only an even weaker form of quantum discord is present. We could thus conclude that, for this type of quantum generalization of HNNs, the main emerging quantum effects occur at the level of average operators only. 

As a possible outlook, related to the quantum generalized HNN, it would be interesting to understand how a pairwise interaction Hamiltonian \cite{MarshEtAl:PhysRevX:21} could change the amount of asymptotic quantum correlations. To this end, a step forward in this direction could be achieved by relating the asymptotic state with spin coherent states \cite{ZhangFG90, MichoelN04central}. 
Finally, beyond this example, our result can be of interest for analyzing models that feature collective operator-valued rates, which can be derived, e.g., from (mean-field) kinetically-constrained models, or applied, more in general, to those cases where the quantum master equation is compatible with a  collective description.

\section{Acknowledgments}
The author acknowledges L. M. Vasiloiu and S. Marcantoni for valuable discussions and suggestions. E. F. acknowledges funding by the European Union's Horizon Europe programme through Grant No. 101105267.

\appendix
\setcounter{lemma}{6}
\setcounter{theorem}{1}

\section{Lemmata mesoscopic dynamics}
\label{app:B}

\subsection{Lemma 7}
\begin{lemma}\label{lemma_squarefluctu_convergence}
    Given the generator $\linn$ specified by Eqs.~\eref{Lindblad}-\eref{delta} , with operator-valued functions $\Gamma_\ell (\Delta^\ell_N)$ satisfying Assumption~\eref{Gamma}, 
    $\forall \mu = 1,...,d^2$ \, and $\forall t < \infty$,
    $$ \lim_{N \rightarrow +\infty} \omega \left(  W^N(\vec{r}) e^{t \linn} \left[ \left( F_{\mu}^{N}(t)\right)^2 \right]  W^{N \dagger}(\vec{r})  \right) < +\infty \, ,$$ 
    if it is satisfied at $t=0$.
\end{lemma}
\begin{proof}
    We start by defining the quantity 
    \begin{equation}
        \mathcal{E}^W_N(t) \equiv \sum_{\mu=1}^{d^2} \omega \left(  W^N(\vec{r}) e^{t \linn} \left[ \left( F_{\mu}^{N}(t)\right)^2 \right]  W^{N \dagger}(\vec{r})  \right) ,
     \end{equation}
such that, as  $\left( F_{\mu}^{N}(t) \right)^2$ is a positive quantity, $\omega \left(  W^N(\vec{r}) e^{t \linn} \left[ \left( F_{\mu}^{N}(t)\right)^2 \right]  W^{N \dagger}(\vec{r})  \right) \leq \mathcal{E}^W_N(t)$. The strategy of the proof is based on the demonstration of Theorem 1 in Ref.\cite{FiorelliEtALC_NJP_23} (see also \cite{CarolloL:PRL:21, Pickl11}. We will show that there exist $N$-independent constant $C_1, C_2$ such that $\dot{\mathcal{E}}^W_N(t) \leq C_1 \mathcal{E}^W_N+ C_2 $, and thus \cite{Gronwall19}
$$ \mathcal{E}^W_N(t) \leq e^{C_1 t} \mathcal{E}^W_N(0) + C_2(e^{C_1 t}-1)/C_1\, ,$$ 
which concludes the proof, provided that the same condition holds for the initial state, i.e. $\lim_{N \rightarrow + \infty} \mathcal{E}^W_N(0) < +\infty$. By appropriately employing the Cauchy-Schwartz inequality, it can be shown that this is satisfied if 
\begin{equation}
   \lim_{N \rightarrow +\infty }\omega([F_{\mu}^N]^2) = \lim_{N \rightarrow +\infty } \omega([m_{\mu}^N-m_{\mu}(0)]^2)N \leq +\infty \, .
\end{equation}
As commented in the main text, the latter identifies an initial, clustering state. 

We focus on the time derivative $\frac{d}{dt}\mathcal{E}^W_N(t)$, which can be written as
\begin{equation}
\eqalign{
    \frac{d}{dt}\mathcal{E}^W_N(t) & = \sum_{\mu} \omega\left(W^N(\vec{r}) e^{t\linn}\left[ \linn \left( \left( m_{\mu}^{N}-m_{\mu}(t) \right)^2 \right) \right]W^{N \dagger}(\vec{r})\right) N\\
    & -2 \sum_{\mu} \dot{m_{\mu}}(t) \omega \left( W^N(\vec{r}) e^{t\linn} \right[  m_{\mu}^{N}-m_{\mu}(t)  \left]W^{N \dagger}(\vec{r}) \right) N \, ,
    }
\end{equation}
where we used $\left( F_{\mu}^{N}(t) \right)^2 = N \left( m_{\mu}^{N}-m_{\mu}(t) \right)^2 $. Regarding the first term on the right-hand side, we exploit that $\linn[AB] = \linn[A]B+A\linn[B] + \sum_{\ell k}[J_\ell^{k \, \dagger},A][B, J_\ell^{k}]$, where $J_\ell^{k} $ are the jump operators defined by Eq.~\eref{jumps}. Here, taking $A=B=m_{\mu}^N - m_{\mu}(t)$, the operator $P = \sum_{\ell k}[J_\ell^{k \, \dagger},A][B, J_\ell^{k}]$ is norm-bounded as
\begin{equation}
    \| P\| \leq C_P/N \, , \qquad C_P= \sum_{\ell=1}^q \| j_\ell \|^2 [2 \delta \gamma'(\delta_\ell)+d^2a_{\max}\gamma(\delta_\ell)]^2 \, ,
\end{equation}
as shown in Ref.~\cite{FiorelliEtALC_NJP_23} as a consequence of Lemma \eref{lemma_commutators}. Therefore,
\begin{equation}\label{pw}
 P^W_t \equiv \sum_{\mu} \omega\left(W^N(\vec{r}) e^{t\linn}\left[  P  \right]W^{N \dagger}(\vec{r})\right) N\leq d^2 C_P \, ,    
\end{equation}
that is going to contribute to the constant $C_2$. For the remaining terms, we exploit that $i) \, {m}_{\mu}(t)$ is a scalar function, and $ii)$ as per Lemma \eref{lemma_gen_action}, the action of the generator on average operators is bounded as $ \| \mathcal{L}_N[m^{N}_{\alpha}] -  f_\alpha(\vec{m}^N)\| \leq \frac{C_{L}}{N} $. By introducing the quantity
\begin{equation}
\eqalign{
    D^{\mu, I}_{W}(t) & \equiv \omega\left(W^N(\vec{r}) e^{t\linn}\left[  \left(f_{\mu}(\vec{m}^N)- f_{\mu}(\vec{m}(t)) \right) \right.\right. \\
    & \left. \left. \times \left( m_{\mu}^{N}-m_{\mu}(t) \right) \right]W^{N \dagger}(\vec{r})\right) N \, ,
    }
\end{equation}
we can thus write
\begin{eqnarray}
    \eqalign{
    \frac{d}{dt}\mathcal{E}^W_N(t)  & = \sum_{\mu} \left(  D^{\mu, I}_{W}(t)  + D^{\mu, I \, \dagger}_{W}(t) \right)  + P_{t}^W + \\
    & + \sum_{\mu} \omega \left( W^N(\vec{r}) e^{t\linn} \right[L \left( m_{\mu}^{N}-m_{\mu}(t) \right) \left]W^{N \dagger}(\vec{r}) \right)N  + \\
    & + \sum_{\mu} \omega \left( W^N(\vec{r}) e^{t\linn} \right[ \left( m_{\mu}^{N}-m_{\mu}(t) \right) L \left]W^{N \dagger}(\vec{r}) \right)N \, ,
    }
\end{eqnarray}
where for any finite $N$,  $L \equiv \mathcal{L}_N[m^{N}_{\alpha}] -  f_\alpha(\vec{m}^N) $ .
Regarding the last two terms on the right-hand side. To this end, notice that by employing Lemma \eref{Lemma-dilation}, it is 
\begin{equation}
\eqalign{
    & | \omega \left( W^N(\vec{r}) e^{t\linn} \right[L \left( m_{\mu}^{N}-m_{\mu}(t) \right) \left]W^{N \dagger}(\vec{r}) \right) |N  \leq 2C_L\, ,
    }
\end{equation}
having further exploited $\| m_{\mu}^N \| \leq 1$, and $|m_{\mu}(t)| \leq 1$ by Lemma \eref{lemma_bound_mf}. As a consequence, we have the following bound
\begin{equation}\label{contrL}
    | \sum_{\mu} \omega \left( W^N(\vec{r}) e^{t\linn} \right[L \left( m_{\mu}^{N}-m_{\mu}(t) \right) \left]W^{N \dagger}(\vec{r}) \right)N | \leq 2 d^2 C_L \, , 
\end{equation}
and the same for the complex conjugate.
As a next step, let us focus on the term $ D^{\mu, I}_{W}(t)$, which, taking the expression of the generator in Lemma \eref{lemma_gen_action}, and after some algebraic manipulations, reads
\begin{equation}
\eqalign{
D_{W}^{\mu,I}(t) &=i\sum_{\beta=1}^{d^2} A_{\mu \beta} \, \omega \left( W^N e^{t \linn} \left([m_{\beta}^{N} -m_{\beta}(t)][m_{\mu}^{N} -m_{\mu}(t)] \right) W^{N \, \dagger}\right) N\\
& + i\sum_{\gamma, \beta =1}^{d^2} B_{\mu \beta \gamma}  \left\lbrace \omega \left( W^N  e^{t \linn} \left( [m_{\beta}^{N} -m_{\beta}(t)]m_{\gamma}^{N}[m_{\mu}^{N} -m_{\mu}(t)] \right) W^{N \, \dagger} \right) \right. \\
&  \left. + \omega \left( W^N  e^{t \linn} \left( m_{\beta}(t)[m_{\gamma}^{N} -m_{\gamma}(t)][m_{\mu}^{N} -m_{\mu}(t)] \right) W^{N \, \dagger}  \right) \right\rbrace N \\
& + \sum_{\ell=1}^{q}\sum_{ \beta=1}^{d^2} M_{\ell \mu}^{\beta} \left\lbrace \omega \left( W^N  e^{t \linn} \left( \Gamma_{\ell}^{2}(\Delta_{\ell}(t))[m_{\beta}^{N} -m_{\beta}(t)][m_{\mu}^{N} -m_{\mu}(t)]  \right) W^{N \, \dagger} \right)  \right. \\ 
& + \left. \omega \left( W^N  e^{t \linn} \left( \left[ \Gamma^{2}_{\ell}(\Delta^{\ell}_{N})-\Gamma^{2}_{\ell}(\Delta_{\ell}(t)) \right] m_{\beta}^{N} [m_{\mu}^{N} -m_{\mu}(t)]\right) W^{N \, \dagger} \right) \right\rbrace N.   
}
\end{equation}
From the above expression, it follows that we have to deal with terms of the type
\begin{itemize}
    \item[(I)] $\omega \left( W^N e^{t \linn} \left([m_{\beta}^{N} -m_{\beta}(t)]X[m_{\mu}^{N} -m_{\mu}(t)] \right) W^{N \, \dagger}\right) $,
    \item[(II)] $ \omega \left( W^N e^{t \linn} \left( \left[ \Gamma^2_\ell(\Delta^{\ell}_N)-\Gamma^2_\ell(\Delta_{\ell}(t))\right] X[m_{\alpha}^{N} - m_{\alpha}(t)] \right) W^{N \, \dagger} \right) $,
\end{itemize} with $X$ some operators. By exploiting Lemma \eref{Lemma-dilation} (see also Ref.~\cite{CarolloL:PRL:21}) terms such as (I) can be bounded as
\begin{equation}
\eqalign{
& |\omega \left( W^N e^{t \linn} \left([m_{\beta}^{N} -m_{\beta}(t)]X[m_{\mu}^{N} -m_{\mu}(t)] \right) W_N^{\dagger}\right)| N \\
& \le \left\| X \right\| \sqrt{\omega \left( W^N e^{t \linn} \left([m_{\beta}^{N} -m_{\beta}(t)]^2 \right) W_N^{\dagger}\right)} \\
& \times \sqrt{\omega \left( W^N e^{t \linn} \left([m_{\mu}^{N} -m_{\mu}(t)]^2 \right) W_N^{\dagger}\right)}N \\
& \leq  \left\| X \right\| \E^{W}(t) \, ,
}
\end{equation}
and thus we have
\begin{eqnarray}
& | \omega \left( W^N e^{t \linn} \left([m_{\beta}^{N} -m_{\beta}(t)][m_{\mu}^{N} -m_{\mu}(t)] \right) W_N^{\dagger}\right) | N \le  \E^W(t), \\
& |\omega \left( W^N e^{t \linn} \left([m_{\beta}^{N} -m_{\beta}(t)]m_{\gamma}^N[m_{\mu}^{N} -m_{\mu}(t)] \right) W_N^{\dagger}\right)| N \le \E^W(t), \\
& |\omega \left( W^N e^{t \linn} \left([m_{\beta}^{N} -m_{\beta}(t)]m_{\gamma}(t)[m_{\mu}^{N} -m_{\mu}(t)] \right) W_N^{\dagger}\right)| N \le \E^W(t),\\
& | \omega \left( W^N e^{t \linn} \left(\Gamma_{\ell}^{2}(\Delta_\ell(t))[m_{\beta}^{N} -m_{\beta}(t)][m_{\mu}^{N} -m_{\mu}(t)] \right) W_N^{\dagger}\right) | N \\ \nonumber
& \le \gamma^2(\delta_\ell)\E^W(t) \, ,
\end{eqnarray}
with $X=\mathbb{I}$, $X=m_{\gamma}^{N}$, $X=m_{\gamma}(t)$, and $X = \Gamma_{\ell}^{2}(\Delta_{\ell}(t)) $, respectively. Note that we have further exploited that $\|m_{\mu}^{N}\| \leq 1$, and, by Lemma \eref{lemma_bound_mf}, that $|m_{\mu}(t)| \leq 1 $ and $|\Delta_\ell (t)| \leq \delta_{\ell}$.
For terms such as (II), we exploit that, by Lemma \eref{proof_Lemma-aux}
\begin{equation}\label{lemma-aux-adt}
\eqalign{
& |\omega\left(W^N e^{t\mathcal{L}_N}\left[(\Gamma_\ell^2(\Delta_N^\ell)-\Gamma_\ell^2(\Delta_\ell(t)))Y \right] W^{N \, \dagger}\right)| \\
& \le  C \| Y \| \sum_{\alpha=1}^{d^2} |r_{\ell \alpha}| \sqrt{\omega(W^N e^{t \linn}[(m_{\alpha}^N-m_{\alpha}(t))^2]W^{N \, \dagger})} \, ,
}
\end{equation}
where $C=2\gamma(\delta_{\ell})\gamma'(\delta_\ell)$, hence
\begin{equation}
\eqalign{
& |\omega\left(W^N e^{t\mathcal{L}_N}\left[ \Gamma^2_\ell(\Delta^{\ell}_N)-\Gamma^2_\ell(\Delta_{\ell}(t))\right]X[m_{\mu}^{N} - m_{\mu}(t)])W^{N \, \dagger}\right)|N \le \\
&  2 \gamma(\delta_\ell)\gamma'(\delta_{\ell})\| X \| 2 \sum_{\beta} |r_{\ell \beta}| \sqrt{\omega_t([m_{\beta}^N-m_{\beta}(t)]^2)} N \le \\
& 4 \delta_\ell \gamma(\delta_\ell)\gamma'(\delta_{\ell}) \| X \| \E^W(t) \, ,
}
\end{equation}
with $\| X \|=\| m_{\gamma}^{N}\| \le 1$. Collecting the above results, the following bound holds
\begin{equation}
    |D_{W}^{\mu,I}(t)| \leq \frac{C_0}{2} \E^W(t) \, ,
\end{equation}
where 
$$  
C_0=2\left(d^{2} A + d^42 B+ qd^2 M [\gamma^2(\delta)+4\delta \gamma(\delta)\gamma'(\delta)]\right) \, .
$$
Retracing the same steps for the complex conjugate $ [D_{W}^{\mu,I}(t)]^\dagger$, and including the contributions from the bounds \eref{pw}, \eref{contrL}, we find
\begin{equation}
    \frac{d}{dt}\mathcal{E}^W_N(t) \leq |\frac{d}{dt}\mathcal{E}^W_N(t)| \leq (d^2 C_0) \E^W(t) + d^2(4C_L+ C_P) \, .
\end{equation}
We can then identify $C_1 \equiv d^2 C_0 $, and $C_2 \equiv d^2 (4C_L+ C_P)$. 

\end{proof}

\subsection{Lemma 8}
\begin{lemma}\label{lemma7}
    Given the local exponential $W_t^N(\vec{r}) = e^{i \vec{r}\cdot \vec{F}^N_t}$, the action of the single body Hamiltonian term contained in Eq.~\eref{e0_totalHamiltonian} is such that
    \begin{equation*}
    \eqalign{
        \lim_{N \rightarrow \infty} \omega_{\vec{r}_1 \vec{r}_2 } & \left( e^{t\linn} \left( \mathcal{H}_{\mathrm{sp}}^N \left[ W_t^{N}(\vec{r}) \right] \right) \right) \\
        & =        \lim_{N \rightarrow \infty} \omega_{\vec{r}_1 \vec{r}_2 } \left( e^{t\linn} \left[ \left( -\vec{r} \cdot ( A\vec{F}^N_t + \sqrt{N}A \vec{\omega}_t^{N} )  \right. \right. \right. \\
        & \left. \left.\left. + \frac{i}{2}\vec{r} \cdot (A T^N \vec{r})\right)W_t^{N}(\vec{r})  \right] \right) \, ,}
    \end{equation*}
    where $\mathcal{H}_{\mathrm{sp}}[\cdot] = i \sum_{\mu=1}^{d^2}\sum_{k=1}^N \epsilon_{\mu} [v_{\mu}^{(k)}, \cdot]$ is the action of the single-particle Hamiltonian, and $A$ is the matrix with elements $A_{\alpha \beta} = \sum_{\gamma}\epsilon_{\gamma}a^{\gamma}_{\alpha \beta} = \sum_{\gamma} \epsilon_{\gamma} a_{\gamma \alpha}^{\beta}$.
\end{lemma}
\begin{proof}
In the following, we will make use of the following relation,
\begin{equation}
\eqalign{
    e^xye^x  
    & = y + [x,y] + \frac{1}{2!}[x, [x,y ]] + \frac{1}{3!} [x,[x,[x,y]]] + \sum_{n\geq 4 } \frac{1}{n!} \mathbb{K}_{x}^{n}[y] \\
    & = \sum_{n\geq 0 } \frac{1}{n!} \mathbb{K}_{x}^{n}[y] \, ,
    }
\end{equation}
where $\mathbb{K}^n_x[y] = [x,\mathbb{K}_x^{n-1}[y]]$, and $\mathbb{K}_x^0[y]=y$. The action of the single particle Hamiltonian on the local exponential operators reads
    \begin{equation}
    \eqalign{
        & \mathcal{H}_{\mathrm{sp}}^N \left[ W_t^{N}(\vec{r}) \right] =  i \sum_{\mu=1}^{d^2}\sum_{k=1}^N \epsilon_{\mu} [v_{\mu}^{(k)},  W_t^{N}(\vec{r}) ] \\
        & =    i \sum_{\mu=1}^{d^2}\sum_{k=1}^N \epsilon_{\mu} \left( v_{\mu}^{(k)} -W_t^{N}(\vec{r})v_{\mu}^{(k)} W_t^{N \dagger}(\vec{r})  \right)W_t^{N}(\vec{r})  \\
        & = i \sum_{\mu=1}^{d^2}\sum_{k=1}^N \epsilon_{\mu} \left( v_{\mu}^{(k)} - \sum_{n \geq 0} \frac{(i)^n}{n!}\mathbb{K}^n_{\vec{r} \cdot \vec{F}^{N}_t}[v_{\mu}^{(k)}]  \right)W_t^{N}(\vec{r}) \\
        & = -i \sum_{\mu=1}^{d^2}\sum_{k=1}^N \epsilon_{\mu} \left( \sum_{n\geq 1} \frac{(i)^n}{n!}\mathbb{K}^n_{\vec{r} \cdot \vec{F}^{N}_t}[v_{\mu}^{(k)}]  \right)W_t^{N}(\vec{r}) \\
        & = -i \sum_{\mu=1}^{d^2}\sum_{k=1}^N \epsilon_{\mu} \left( i \left[ \vec{r} \cdot \vec{F}^{N}_t, v_{\mu}^{(k)} \right] - \frac{1}{2}\left[ \vec{r} \cdot \vec{F}^{N}_t,\left[ \vec{r} \cdot \vec{F}^{N}_t, v_{\mu}^{(k)} \right] \right] \right. \\
        & \left. +  \sum_{n \geq 3} \frac{i^n}{n!}\mathbb{K}^n_{\vec{r} \cdot \vec{F}^{N}_t}[v_{\mu}^{(k)}]  \right)W_t^{N}(\vec{r}) \\
        & =  \sum_{\mu=1}^{d^2}\sum_{k=1}^N \epsilon_{\mu} \left( \left[ \vec{r} \cdot \vec{F}^{N}_t, v_{\mu}^{(k)} \right] + \frac{i}{2}\left[ \vec{r} \cdot \vec{F}^{N}_t,\left[ \vec{r} \cdot \vec{F}^{N}_t, v_{\mu}^{(k)} \right] \right] \right. \\
        & \left. -i  \sum_{n \geq 3} \frac{(i)^n}{n!}\mathbb{K}^n_{\vec{r} \cdot \vec{F}^{N}_t}[v_{\mu}^{(k)}]  \right)W_t^{N}(\vec{r}) \, .
        }
    \end{equation}
Let us consider the contributions on the last two lines, resulting from  $\mathbb{K}^n_{\vec{r}\cdot \vec{F}^N_t}[v_{\mu}^{(k)}]$ with $n=1$, $n=2$, $n \geq 3$,  respectively. The first term reads
\begin{equation}
\eqalign{
    & \sum_{\mu=1}^{d^2}\sum_{k=1}^N \epsilon_{\mu} \left[ \vec{r} \cdot \vec{F}^{N}_t, v_{\mu}^{(k)} \right] = \sum_{\mu, \nu}\sum_{k} \epsilon_{\mu} r_{\nu} [F_{\nu}^{N}, v_{\mu}^{(k)}] \\ 
    &= \sum_{\mu, \nu}\sum_{k} \epsilon_{\mu} r_{\nu} \frac{1}{\sqrt{N}} \sum_{\gamma}a_{\nu \mu}^{\gamma} v_{\gamma}^{(k)} \\
    & = \sum_{\mu, \nu}\epsilon_{\mu} r_{\nu}  \sum_{\gamma}a_{\nu \mu}^{\gamma}\frac{1}{\sqrt{N}}\sum_{k} \left( v_{\gamma}^{(k)} - \omega_t(v_{\gamma}^{(k)}) + \omega_t(v_{\gamma}^{(k)})\right) =  \\
    & - \sum_{\nu} r_{\nu} \sum_{\mu} \sum_{\gamma}\epsilon_{\mu}a_{\mu \nu}^{\gamma}\left( F_{\gamma}^{N}(t) + \sqrt{N}\omega_t(m_{\gamma}^{N})\right) = - \vec{r}\cdot \left( A \vec{F}_t^N +\sqrt{N} A \vec{\omega}_t^{N} \right)
    }
\end{equation}
 The second contribution reads
 \begin{equation}
     \eqalign{ &  \frac{i}{2}\left[ \vec{r} \cdot \vec{F}^{N}_t,\left[ \vec{r} \cdot \vec{F}^{N}_t, v_{\mu}^{(k)} \right] \right] = \frac{i}{2N}\sum_{\mu, \nu, \gamma} \sum_{k} \epsilon_{\mu} r_{\nu}r_{\gamma} [v_{\nu}^{(k)},[v_{\gamma}^{(k)},v_{\mu}^{(k)}]] \\
     & =  \frac{i}{2N}\sum_{\mu, \nu, \gamma} \sum_{k} \epsilon_{\mu} r_{\nu}r_{\gamma} \sum_{\eta} a_{\gamma \mu}^{\eta} [v_{\nu}^{(k)},v_{\eta}^{(k)}] =-\frac{i}{2}\sum_{\nu, \gamma} r_{\nu}r_{\gamma} \sum_{\mu }\epsilon_{\mu} a_{\gamma \eta}^{\mu} T_{\nu \eta}^{N} \\ & = \frac{i}{2}\sum_{\nu, \gamma} r_{\gamma} A_{\gamma \eta} T_{ \eta \nu }^{N} r_{\nu} = \frac{i}{2} \vec{r} \cdot \left( A T^{N} \vec{r} \right) \, . }
 \end{equation}
Finally, we can show that the third contribution is bounded and negligible in the large $N$ limit. Indeed, it is
\begin{equation}
\eqalign{
   &  \left\| \left( -i \sum_{\mu=1}^{d^2}\sum_{k=1}^N \epsilon_{\mu} \sum_{n\geq 3} \frac{(i)^n}{n!}\mathbb{K}^n_{\vec{r} \cdot \vec{F}^{N}_t}[v_{\mu}^{(k)}]  \right)W_t^{N}(\vec{r}) \right\| \\
   & \leq  \left\|  \sum_{\mu=1}^{d^2}\sum_{k=1}^N \epsilon_{\mu}  \sum_{n \geq 3} \frac{(i)^n}{n!}\mathbb{K}^n_{\vec{r} \cdot \vec{F}^{N}_t}[v_{\mu}^{(k)}]   \right\| \\
   & \leq  \sum_{\mu=1}^{d^2}\sum_{k=1}^N |\epsilon_{\mu}| \sum_{n \geq 3}^{n} \frac{1}{n!}  \left\|  
 \mathbb{K}^n_{\vec{r} \cdot \vec{F}^{N}_t}[v_{\mu}^{(k)}]   \right\| \leq \frac{d^2 v |\epsilon| }{\sqrt{N}} e^{2vd^2r_{\max}},
    }
\end{equation}
where, $r_{\max}= \max_{\nu_i} r_{\nu_i}$, $\| v_{\mu}\| \equiv v$, $|\epsilon|=\max_{\mu}| \epsilon_{\mu}|$, and we used $ \mathbb{K}^n_{\vec{r} \cdot \vec{F}^{N}_t}[v_{\mu}^{(k)}] = \sum_{\nu_1,..., \nu_n} r_{\nu_{1}}\cdots r_{\nu_{n}}[F_{\nu_{1}}^{N},[...,[F_{\nu_{n}}^{N},v_{\mu}^{(k)}]]...]$. The latter expression, by exploiting $\| [X,Y] \| \leq 2\|X\| \, \| Y\|$, can be indeed bounded as
\begin{equation}
\eqalign{
    & \left\|  
 \mathbb{K}^n_{\vec{r} \cdot \vec{F}^{N}_t}[v_{\mu}^{(k)}]   \right\| \leq \sum_{\nu_1,..., \nu_n}\frac{  r_{\nu_{1}}\cdots r_{\nu_{n}}}{(N)^{n/2}}\| [v_{\nu_{1}}^{(k)},[...,[v_{\nu_{n}}^{(k)},v_{\mu}^{(k)}]]...] \| \\
 & \leq (d^2 r_{\max} 2 v)^n \frac{v}{N^{n/2}} \, .
 }
\end{equation}
\end{proof}
\subsection{Lemma 9}
\begin{lemma}\label{lemma8st}
    Given the local exponential $W_t^N(\vec{r}) = e^{i \vec{r}\cdot \vec{F}^N_t}$, the action of the two-body Hamiltonian term in Eq.~\eref{e0_totalHamiltonian} is such that
    \begin{equation*}
    \eqalign{
        \lim_{N \rightarrow \infty} \omega_{\vec{r}_1 \vec{r}_2 } & \left( e^{t\linn} \left( \mathcal{H}_{\mathrm{tp}}^N \left[ W_t^{N}(\vec{r}) \right] \right) \right) \\
        & =   \lim_{N \rightarrow \infty} \omega_{\vec{r}_1 \vec{r}_2 } \left( e^{t\linn} \left[ \left( -\frac{1}{2}\vec{r} \cdot ( T^N(- 2h^{\mathrm{I}} + 2i h^{\mathrm{R}})T^N \, \vec{r})\right. \right. \right. \\
        & - i \vec{r}\cdot (T^N (2ih^{\mathrm{R}})\vec{F}^N_t)  \\
        & \left. \left. \left. + \frac{1}{2}\vec{r} \cdot (\tilde{B}_t T^N \vec{r}) + i\vec{r}\cdot (\tilde{B}_t\vec{F}^N_t)+ \sqrt{N} i \vec{r}\cdot \tilde{B}_t^N \vec{\omega}_t^{N} \right)W_t^{N}(\vec{r})  \right] \right) \, ,}
    \end{equation*} 
    where $\mathcal{H}_{\mathrm{tp}}[\cdot] = i\frac{1}{N}\sum_{j,k}\sum_{\alpha,\beta}h_{\alpha \beta} [v_{\alpha}^{(k)} v_{\beta}^{(j)}, \cdot]$, and $\tilde{B}_t$ is a matrix with elements $\tilde{B}_{\beta \gamma}(t) = \sum_{\mu\nu}(2ih_{\mu\nu}^{\mathrm{R}})a_{\beta\gamma}^{\mu} \omega^{N}_{\nu}(t)$, with $h^{\mathrm{R}}_{\alpha \beta} = \mathrm{Re}( h_{\alpha \beta})$, and $h^{\mathrm{I}}_{\alpha \beta} = \mathrm{Im}( h_{\alpha \beta})$.
\end{lemma}
\begin{proof}
    Let us write the action of the two-body Hamiltonian term in a more convenient form. By introducing the operators $V_{\alpha} ^N= \frac{1}{\sqrt{N}}\sum_{k}v_{\alpha}^{(k)}$, and writing $h_{\alpha \beta} = h_{\alpha \beta}^{\mathrm{R}} + i \ h_{\alpha \beta}^{\mathrm{I}}$, it is 
\begin{equation}\label{action_intarHamiltonian}
\eqalign{
& \mathcal{H}_{\mathrm{tp}}[\cdot]  = i\sum_{\alpha,\beta}h_{\alpha \beta} [V_{\alpha}^{N} V_{\beta}^{N}, \cdot]  
\\
& = i\sum_{\alpha \beta } h_{\alpha \beta}^{\mathrm{R}} \left( V_{\alpha}^{N}[V_{\beta}^{N}, \cdot] + [V_{\alpha}^{N}, \cdot]V_{\beta}^{N} \right) - \sum_{\alpha \beta} h_{\alpha \beta}^{\mathrm{I}} \left( V_{\alpha}^{N}[V_{\beta}^{N}, \cdot] - [V_{\alpha}^{N}, \cdot]V_{\beta}^{N} \right) \\
& =  i\sum_{\alpha \beta } h_{\alpha \beta}^{\mathrm{R}} \left\lbrace  [V_{\alpha}^{N}, \cdot]V_{\beta}^{N} \right\rbrace - \sum_{\alpha \beta} h_{\alpha \beta}^{\mathrm{I}} \left[ \left[V_{\alpha}^{N}, \cdot\right],V_{\beta}^{N} \right] \, ,
}
\end{equation}
where we used that $ h_{\alpha \beta} = -h_{\beta \alpha}^{*} $, and thus $h_{\alpha \beta}^{\mathrm{R}} = h_{ \beta \alpha}^{\mathrm{R}} $, and $h_{\alpha \beta}^{\mathrm{I}} = - h_{ \beta \alpha}^{\mathrm{I}} $.

We then consider the action of the second term on exponential operators, $ H_I \equiv - \sum_{\alpha \beta} h_{\alpha \beta}^{\mathrm{I}} \left[ \left[V_{\alpha}^{N}, W_t^{N}(\vec{r}) \right],V_{\beta}^{N} \right] $. We first note that
\begin{equation}
    \| \left[V_{\mu}^{N}, W_t^{N}(\vec{r}) \right]+ i [\vec{r} \cdot \vec{F}^N_t, V_{\mu}^{N}]\| \leq \frac{v}{\sqrt{N}}e^{2d^2 r_{\max}} \, ,
\end{equation}
which can be shown by considering the commutator
\begin{equation}
\eqalign{
    \left[V_{\mu}^{N}, W_t^{N}(\vec{r}) \right] & = \left( V_{\mu}^N- W_t^N(\vec{r}) V_{\mu}^N W_t^{N, \dagger}(\vec{r})   \right) W_t^N(\vec{r}) \\
    & = \left( V_{\mu}^N - \sum_{n \geq 0 } \mathbb{K}^{n}_{i\vec{r} \cdot \vec{F}^N_t} [V_{\mu}^N]\right)W_t^N(\vec{r})\\
    & = -\left( i [\vec{r} \cdot \vec{F}^N_t, V_{\mu}^{N}] + \sum_{n \geq 2} \frac{i^n}{n!}  \mathbb{K}^{n}_{\vec{r} \cdot \vec{F}^N_t} [V_{\mu}^N] \right) \, ,
    }
\end{equation}
and by exploiting $\| \mathbb{K}^{n}_{\vec{r} \cdot \vec{F}^N_t} [V_{\mu}^N]\| \leq (d^2 r_{\max} 2 v)^n v/N^{n+1}$. Analogously, it can be proven that 
\begin{equation}
\left\| \left[ \sum_{n \geq 2} \frac{i^n}{n!}  \mathbb{K}^{n}_{\vec{r} \cdot \vec{F}^N_t} [V_{\mu}^N] , V_{\nu}^N  \right] \right\|  \leq \frac{v^2}{N}e^{2d^2r_{\max}} \, .
\end{equation}
For the sake of a lighter notation we define $S_{\mu}^N \equiv \sum_{n \geq 2} \frac{i^n}{n!}  \mathbb{K}^{n}_{\vec{r} \cdot \vec{F}^N_t} [V_{\mu}^N] $, and consider the following commutator at leading order
\begin{equation}
\eqalign{
    &\left[ \left[V_{\mu}^{N}, W_t^{N}(\vec{r}) \right],V_{\nu}^{N} \right]  = \left( V_{\mu}^{N} W_t^{N}(\vec{r})V_{\nu}^{N} - W_t^{N}(\vec{r})V_{\mu}^{N} V_{\nu}^{N} \right. \\
    & \left. - V_{\nu}^{N}  V_{\mu}^{N} W_t^{N}(\vec{r}) + V_{\nu}^{N} W_t^{N}(\vec{r})V_{\mu}^{N}\right) \\
    & = \left[ \left( V_{\mu}^{N} - W_t^{N}(\vec{r})V_{\mu}^{N} W_t^{N \dagger}(\vec{r})\right) W_t^{N}(\vec{r})V_{\nu}^{N} W_t^{N \dagger}(\vec{r}) \right. \\ 
    &-  \left. V_{\nu}^{N} \left( V_{\mu}^{N} - W_t^{N}(\vec{r})V_{\mu}^{N} W_t^{N \dagger}(\vec{r}) \right) \right]W_t^N(\vec{r}) \\
    & = \left[ \left( -i [\vec{r} \cdot \vec{F}^N_t, V_{\mu}^{N}] -S_{\mu}^N \right) \left(  V_{\nu }^N + i [\vec{r} \cdot \vec{F}_t^N , V_{\nu}^N] +S_{\nu}^N \right) \right. \\
    & \left. - V_{\nu}^N\left( -i [\vec{r} \cdot \vec{F}^N_t, V_{\mu}^{N}] -S_{\mu}^N \right) \right] W_t^N(\vec{r}) \\
    & = \left( -i\left[[\vec{r} \cdot \vec{F}_{t}^{N}, V_{\mu}^N],V_{\nu}^N \right] - [S_{\mu}^N, V_{\nu}^N] - i S_{\mu}^N [\vec{r} \cdot \vec{F}_{t}^{N}, V_{\nu}^N] \right. \\
    & \left. - i[\vec{r} \cdot \vec{F}_{t}^{N}, V_{\mu}^N] S_{\nu}^N -S_{\mu}^N S_{\nu}^N + \left[ \vec{r} \cdot \vec{F}_{t}^{N}, V_{\mu}^N \right] \left[ \vec{r} \cdot \vec{F}_{t}^{N}, V_{\nu}^N \right] \right) W_t^N(\vec{r}) \\
    & \sim \left( \left[ \vec{r} \cdot \vec{F}_{t}^{N}, V_{\mu}^N \right] \left[ \vec{r} \cdot \vec{F}_{t}^{N}, V_{\nu}^N \right] \right) W_t^N(\vec{r}) \,.
    }
\end{equation}
Thus, 
\begin{equation}\label{l8_Himaginary}
\eqalign{
   & \lim_{N \rightarrow \infty} \omega_{\vec{r}_1 \vec{r}_2} \left( e^{t \linn} [H_I]   \right) = \lim_{N \rightarrow \infty} \omega_{\vec{r}_1 \vec{r}_2} \left( e^{t \linn} \left[- \sum_{\alpha \beta} h_{\alpha \beta}^{\mathrm{I}} \left[ \left[V_{\alpha}^{N}, W_t^{N}(\vec{r}) \right],V_{\beta}^{N} \right] \right]  \right) \\
    & = \lim_{N \rightarrow \infty} \omega_{\vec{r}_1 \vec{r}_2}  
    \left( e^{t \linn} \left[- \sum_{\alpha \beta} h_{\alpha \beta}^{\mathrm{I}}    
      \left( \left[ \vec{r} \cdot \vec{F}_{t}^{N}, V_{\alpha}^N \right] \left[ \vec{r} \cdot \vec{F}_{t}^{N}, V_{\beta}^N \right] \right) W_t^N(\vec{r})\right] \right) \\
      & =\lim_{N \rightarrow \infty} \omega_{\vec{r}_1 \vec{r}_2}  
    \left( e^{t \linn} \left[ \left(
      \sum_{\mu \nu}  \sum_{\alpha \beta} r_{\mu} T_{\mu \alpha}^N h_{\alpha \beta}^{\mathrm{I}} T_{\beta \nu}^N r_{\nu}  \right) W_t^N(\vec{r}) \right] \right)  \\
      & =\lim_{N \rightarrow \infty} \omega_{\vec{r}_1 \vec{r}_2}  
    \left( e^{t \linn} \left[ \left(
      \vec{r} \cdot \left(T^N h^{\mathrm{I}} T^N \vec{r}\right)  \right) W_t^N(\vec{r})
      \right]\right) \,. 
    }
\end{equation}

Let us now consider the first action of the first term in Eq.~\eref{action_intarHamiltonian} on local exponential, i.e. $H_{R} \equiv i\sum_{\alpha \beta } h_{\alpha \beta}^{\mathrm{R}} \left\lbrace  [V_{\alpha}^{N}, W_t^N(\vec{r})],V_{\beta}^{N} \right\rbrace $, that can be written as 
\begin{equation}
    \eqalign{
     H_{R} & = i \sum_{\mu \nu } h_{\mu \nu}^{\mathrm{R}} \left\lbrace  \left[F_{\mu}^N(t) + \sqrt{N} \omega_t(m_{\mu}^N ) , W_t^N{(\vec{r})} \right] F_{\nu}^N(t) + \sqrt{N} \omega_t(m_{\nu}^N)\right\rbrace \\
    &  = i \sum_{\mu \nu } h_{\mu \nu}^{\mathrm{R}} \left\lbrace  \left[F_{\mu}^N(t) , W_t^N{(\vec{r})} \right], F_{\nu}^N(t) \right\rbrace + 2  i \sum_{\mu \nu } h_{\mu \nu}^{\mathrm{R}}  \omega_t(m_{\nu}^N) \sum_{k} \left[v_{\mu}^{(k)}  , W_t^N{(\vec{r})} \right] \\
    & \equiv H_R^{I}+H_R^{II} \, .
    }
\end{equation}
We will first focus on $ H_R^{I} $, starting from the anticommutator 
\begin{equation}
\eqalign{
   & \left\lbrace  \left[F_{\mu}^N(t) , W_t^N{(\vec{r})} \right] ,F_{\nu}^N(t) \right\rbrace = - \left\lbrace \left(i\left[  \vec{r} \cdot \vec{F}_t^N , V_{\mu}^N\right] + S_{\mu}^N\right) W_t^N(\vec{r})  , F_{\nu}^N(t)\right\rbrace  \\
   & = -\left[ \left( i\left[  \vec{r} \cdot \vec{F}_t^N , V_{\mu}^N\right] + S_{\mu}^N \right) \left( 2F_{\nu}^N(t) + i[\vec{r}\cdot \vec{F}_t^N, V_{\nu}^N] + S_{\nu}^N \right) \right.\\
   & \left. + \left[F_{\nu}^N(t), i\left[  \vec{r} \cdot \vec{F}_t^N , V_{\mu}^N\right] + S_{\mu}^N \right] \right] W_t^N{(\vec{r})} \\
   & \sim - \left\lbrace \left( 2 i \left[ \vec{r} \cdot \vec{F}_t^N , V_{\mu}^N  \right] + 2S_{\mu}^N \right) F_{\nu}^N(t)- \left[ \vec{r} \cdot \vec{F}_t^N , V_{\mu}^N  \right] \left[ \vec{r} \cdot \vec{F}_t^N , V_{\nu}^N  \right]  \right\rbrace W_t^N{(\vec{r})} \,
   }
\end{equation}
where the last step we are taking in consideration the leading order only. It can be further expressed as
\begin{equation}
\eqalign{
    & \left\lbrace  \left[F_{\mu}^N(t) , W_t^N{(\vec{r})} \right] ,F_{\nu}^N(t) \right\rbrace =   \\
    & \sim \left( -2i\sum_{\alpha} r_{\alpha } T_{\alpha \mu}^N F_{\nu}^N(t)  - \sum_{\alpha \beta} r_{\alpha} T_{\alpha \mu}^N T_{\nu \beta}^N r_{\nu}  -2S_{\mu}^N F_{\nu}^N(t) \right)W_{t}^N(\vec{r}) \, .
    }
\end{equation}
With regard to the last expression, we need to demonstrate that the last term is norm-bounded and negligible in the large $N$ limit. By using the result of Lemma \eref{lemma_squarefluctu_convergence},
\begin{equation}\label{statement_lemmB2}
\eqalign{
    & \lim_{N \rightarrow +\infty}   \omega \left(  W^N(\vec{r}) e^{t \linn} \left[ \left( F_{\mu}^{N}\right)^2 \right]  W^{N \dagger}(\vec{r})  \right) < +\infty
    }   
\end{equation}
and exploiting that $\lim_{N \rightarrow \infty} \| [ S_{\mu}^N, F_{\nu}^N(t)] \| = 0$, it is
\begin{equation}
\eqalign{
    & \lim_{N \rightarrow \infty} \omega_{\vec{r}_{1} \vec{r}_{2}} \left( e^{t \linn } \left[  S_{\mu}^N F_{\nu}^N(t) W_t^{N}(\vec{r}) \right] \right) = \lim_{N \rightarrow \infty} \omega_{\vec{r}_{1} \vec{r}_{2}} \left( e^{t \linn } \left[  F_{\nu}^N(t) S_{\mu}^N  W_t^{N}(\vec{r}) \right] \right) \\
    & \leq \lim_{N \rightarrow +\infty} \| S_{\mu}^N \| \sqrt{ \omega \left(  W^N(\vec{r}) e^{t \linn} \left[ \left( F_{\mu}^{N}\right)^2 \right]  W^{N \dagger}(\vec{r})  \right)} = 0
    }   
\end{equation}
From the above result it follows that 
\begin{equation}\label{l8_HRealI}
   \eqalign{
    & \lim_{N \rightarrow \infty} \omega_{\vec{r}_{1} \vec{r}_{2}} \left( e^{t \linn } \left[ H_{R}^I \right] \right) = \\
    & \lim_{N \rightarrow \infty} \omega_{\vec{r}_{1} \vec{r}_{2}} \left( e^{t \linn } \left[ \left( -2i \vec{r} \cdot \left( T^N (i h^{\mathrm{R}})\vec{F}_t^N \right) - \vec{r} \cdot \left( T^N (ih^{\mathrm{R}}) T^N \right) \right) W_t^N{\vec{r} }\right] \right) \, .
    }
\end{equation}
We can  now proceed with the term $H_R^{II} = 2  i \sum_{\mu \nu } h_{\mu \nu}^{\mathrm{R}}  \omega_t(m_{\nu}^N) \sum_{k} \left[v_{\mu}^{(k)}  , W_t^N{(\vec{r})} \right]$. To this end, we consider the following commutator at the leading order,
\begin{equation}
\eqalign{
  & \sum_{k} \left[v_{\mu}^{(k)}  , W_t^N{(\vec{r})} \right] =  \\
  & \left( -i\left[  \vec{r} \cdot \vec{F}_t^N , \sum_k v_{\mu}^{(k)} \right] + \frac{1}{2} \left[  \vec{r} \cdot \vec{F}_t^N , \left[  \vec{r} \cdot \vec{F}_t^N , \sum_k v_{\mu}^{(k)} \right]  \right]-\tilde{S}_{\mu}^{N} \right) W_t^N(\vec{r})\\
  & \sim \left( i\left[  \sum_k v_{\mu}^{(k)} , \vec{r} \cdot \vec{F}_t^N  \right] + \frac{1}{2} \left[  \vec{r} \cdot \vec{F}_t^N , \left[  \vec{r} \cdot \vec{F}_t^N , \sum_k v_{\mu}^{(k)} \right]  \right] \right) W_t^N(\vec{r}) \, ,
  }
\end{equation}
where we defined $\tilde{S}_{\mu}^{N} \equiv \sum_{n \geq 3} \frac{i^n}{n!}  \mathbb{K}^{n}_{\vec{r} \cdot \vec{F}^N_t} [\sum_{k} v_{\mu}^{(k)}] $, and in the last step we kept only the leading order contributions, as it can be proven that $\| \tilde{S}_{\mu}^{N} \|  \leq \frac{v}{\sqrt{N}} e^{2d^{2}r_{\max} v} $. We can thus express the two remaining contributions to $H_{R}^{II}$ exploiting 
 \begin{equation}
\eqalign{
    & 2  i \sum_{\alpha} \sum_{\mu \nu } h_{\mu \nu}^{\mathrm{R}}  \omega_t(m_{\nu}^N) \frac{1}{\sqrt{N}} i\sum_k [v_{\mu}^{(k)}, v_{\alpha}^{(k)}]r_{\alpha} \\
    & = - i \sum_{\alpha \beta} \sum_{\mu \nu } 2  ih_{\mu \nu}^{\mathrm{R}}  \omega_t(m_{\nu}^N) \frac{1}{\sqrt{N}} a_{\alpha \mu}^{\beta} \sum_k v_{\beta}^{(k)} r_{\alpha}  \\
    &  = i \sum_{\alpha \beta} r_{\alpha} \sum_{\mu \nu } 2 i h_{\mu \nu}^{\mathrm{R}}  \omega_t(m_{\nu}^N) a_{\alpha \beta}^{\mu} \left( F_{\beta}^{N}(t)+ \sqrt{N}\omega_{\beta}^{N}(t)  \right)\\
        & = i \sum_{\alpha \beta} r_{\alpha } \tilde{B}_{\alpha \beta}(t) \left( F_{\beta}^{N}(t)+ \sqrt{N}\omega_{\beta}^{N}(t)   \right) = i \vec{r} \cdot ( \tilde{B}_t \vec{F}_{t}^{N} ) + \sqrt{N} \vec{r} \cdot (i \tilde{B}_t \vec{\omega}_t^N ) \, ,
    }
\end{equation}
and
\begin{equation}
    \eqalign{
    & 2  i \sum_{\mu \nu } h_{\mu \nu}^{\mathrm{R}}  \omega_t(m_{\nu}^N)
    \frac{1}{2} \sum_k  \left[  \vec{r} \cdot \vec{F}_t^N , \left[  \vec{r} \cdot \vec{F}_t^N , v_{\mu}^{(k)} \right]  \right]   \\
    &= \frac{1}{2} \sum_{\alpha \beta} r_{\alpha} r_{\beta} \sum_{\mu \nu} (2ih_{\mu \nu}^{\mathrm{R}} )\omega_t(m_{\nu}^N) \frac{1}{N}\sum_{k}[v_{\alpha}^{(k)}, [v_{\beta}^{(k)}, v_{\mu}^{(k)}]] \\
    & = \frac{1}{2} \sum_{\alpha \beta \gamma} r_{\alpha} r_{\beta}  \sum_{\mu \nu} (2ih_{\mu \nu}^{\mathrm{R}} )\omega_t(m_{\nu}^N) a^{\mu}_{\beta \gamma}T_{\gamma \alpha}^N = \frac{1}{2} \sum_{\alpha \beta \gamma} r_{\alpha} r_{\beta}  \tilde{B}_{\beta \gamma}(t) T_{\gamma \alpha}^N = \\
    & = \frac{1}{2} \vec{r} \cdot \left( \tilde{B}_t T^{N} \vec{r} \right) \, .
    }
\end{equation}
As a result, the last two expressions allow us to conclude that
\begin{equation}\label{l8_HRealII}
   \eqalign{
    & \lim_{N \rightarrow \infty} \omega_{\vec{r}_{1} \vec{r}_{2}} \left( e^{t \linn } \left[ H_{R}^ {II} \right] \right) = \\
    & \lim_{N \rightarrow \infty} \omega_{\vec{r}_{1} \vec{r}_{2}} \left( e^{t \linn } \left[ \left( i \vec{r} \cdot ( \tilde{B}_t \vec{F}_{t}^{N} ) + \sqrt{N} \vec{r} \cdot (i \tilde{B}_t \vec{\omega}_t^N ) \right.\right.\right. \\
    & \left.\left.\left.+ \frac{1}{2} \vec{r} \cdot \left( \tilde{B}_t T^{N} \vec{r} \right)   \right) W_t^N{\vec{r} }\right] \right) \, .
    }
\end{equation}
The proof of the lemma is obtained by considering together Eqs.~\eref{l8_Himaginary}, \eref{l8_HRealI}, \eref{l8_HRealII}.
\end{proof}

\subsection{Lemma 10}
Before proceeding further, we use that the action of $\mathcal{D}_\ell $ on a local exponential can be split in two terms. We have
$$
\mathcal{D}_\ell[W_t^N(\vec{r})]=\frac{1}{2}\sum_{k=1}^N\left(\left[\Gamma_\ell(\Delta_N^\ell) j_\ell^{\dagger\, (k)}, W_t^N(\vec{r})\right]j_{\ell}^{(k)}\Gamma_\ell(\Delta_N^\ell)+\Gamma_\ell(\Delta_N^\ell)j_\ell^{\dagger \, (k)}\left[W_t^N(\vec{r}),j_\ell^{(k)}\Gamma_\ell(\Delta_N^\ell)\right]\right)\, .
$$
Using that $[AB,C]=A[B,C]+[A,C]B$, we rewrite this as
\begin{equation}
\eqalign{
\mathcal{D}_\ell[W_t^N(\vec{r})] =& \frac{1}{2}\sum_{k=1}^N\left(\Gamma_\ell(\Delta_N^\ell)\left[ j_\ell^{\dagger\, (k)}, W_t^N(\vec{r}) \right]j_{\ell}^{(k)}\Gamma_\ell(\Delta_N^\ell) \right. \\
& \left. +\Gamma_\ell(\Delta_N^\ell) j_\ell^{\dagger\, (k)}\left[W_t^N(\vec{r}),j_\ell^{(k)}\right]\Gamma_\ell(\Delta_N^{\ell})\right)\\
\qquad \qquad \quad +&\frac{1}{2} \sum_{k=1}^N\left(\left[\Gamma_\ell(\Delta_N^\ell), W_t^N(\vec{r})\right]j_\ell^{\dagger\, (k)}j_{\ell}^{(k)}\Gamma_\ell(\Delta_N^\ell) \right. \\
& \left. + \Gamma_\ell(\Delta_N^\ell)j_\ell^{\dagger\, (k)}j_\ell^{(k)}\left[W_t^N(\vec{r}),\Gamma_\ell(\Delta_N^\ell)\right]\right)\, , 
\label{WeylLike-diss-two-terms}
}
\end{equation}
and we separate the action of the dissipation in the following way. The action of the last two terms on the right-hand side of Eq.~\eref{WeylLike-diss-two-terms}, denoted as $D_{I}$, and the first two terms, denoted as $D_{II}$, namely
\begin{equation}\label{diss-first-term}
   \eqalign{
D_I \equiv &\frac{1}{2} \sum_{k=1}^N\left(\left[\Gamma_\ell(\Delta_N^\ell), W_t^N(\vec{r})\right]j_\ell^{\dagger\, (k)}j_{\ell}^{(k)}\Gamma_\ell(\Delta_N^\ell) \right. \\
& \left. + \Gamma_\ell(\Delta_N^\ell)j_\ell^{\dagger\, (k)}j_\ell^{(k)}\left[W_t^N(\vec{r}),\Gamma_\ell(\Delta_N^\ell)\right]\right)\, ,
} 
\end{equation}
and
\begin{equation}\label{diss-second-term}
\eqalign{
D_{II} \equiv & \frac{1}{2}\sum_{k=1}^N\left(\Gamma_\ell(\Delta_N^\ell)\left[ j_\ell^{\dagger\, (k)}, W_t^N(\vec{r}) \right]j_{\ell}^{(k)}\Gamma_\ell(\Delta_N^\ell) \right. \\
& \left. +\Gamma_\ell(\Delta_N^\ell) j_\ell^{\dagger\, (k)}\left[W_t^N(\vec{r}),j_\ell^{(k)}\right]\Gamma_\ell(\Delta_N^{\ell})\right)\, , 
}
\end{equation}
\begin{lemma}\label{lemma9}
    Given a dynamical generator as defined by Eqs.~\eref{e0_totalHamiltonian}-\eref{jumps}, satisfying Assumption~\eref{Gamma}, the dissipation term $D_I$ defined by Eq.~\eref{diss-first-term} acts on local exponential $W_t^N(\vec{r}) = e^{i \vec{r}\cdot \vec{F}^N_t}$ as
        \begin{equation*}
    \eqalign{
        & \lim_{N \rightarrow \infty} \omega_{\vec{r}_1 \vec{r}_2 }  \left( e^{t\linn} \left[ D_I \right] \right) \\
        & =        \lim_{N \rightarrow \infty} \omega_{\vec{r}_1 \vec{r}_2 } \left( e^{t\linn} \left[ \left( - \frac{1}{2} (\vec{r} \cdot T^{N}\vec{r}_\ell) \left(\Gamma_\ell'(\Delta_N^{\ell}) \right)^2  \frac{1}{N} \sum_{k}  j_\ell^{\dagger (k)} j_\ell^{ (k)}  (\vec{r}_\ell \cdot T^{N}\vec{r})  \right. \right.\right. \\
        & \left. \left. \left. -  \frac{1}{2} \Gamma_\ell(\Delta_N^{\ell}) \Gamma_\ell'(\Delta_N^{\ell}) \vec{r}\cdot (T^N \vec{r}_\ell) (\vec{c}_{\ell} \cdot (T^{N} \vec{r} )) \right)W_t^{N}(\vec{r})  \right] \right) \, ,}
    \end{equation*}
    where $\vec{r}_\ell$ is the vector with components 
$r_{\ell \alpha}$, introduced by Eq.~\eref{eq:Delta_lc_average} as coefficients defining $\Delta_N^{\ell}$ in terms of linear combinations of average operators; $\vec{c}_\ell$ is the vector with components ${c}_{\ell \eta} = \sum_{\gamma, \gamma' } c_{\ell \gamma}^{j *} c_{\ell \gamma'}^{j } \frac{1}{2}(a_{\gamma \gamma'}^\eta + b_{\gamma \gamma'}^{\eta}) $, having introduced $j_{\ell}^{(k)} = \sum_{\gamma} c_{\ell \gamma}^{j} v_{\gamma}^{(k)}$ and $b_{\gamma \gamma'}^{\eta}$ such that $  \lbrace v_{\gamma}^{(k)}, v_{\gamma'}^{(k)}\rbrace = \sum_{\eta} b_{\gamma \gamma'}^{\eta}v_{\eta}^{(k)}$. 
\end{lemma}
\begin{proof}
Looking at Eq.~\eref{diss-first-term}, we split $D_I$ into two parts, $D_I= D_{I}^1+ D_{I}^2$, where we have
\begin{equation}
\eqalign{
    D_{I}^1&=\frac{1}{2}\sum_{k=1}^{N}\left[\Gamma_\ell(\Delta_N^\ell),W_t^N(\vec{r}) \right]\left[j_\ell^{\dagger \, (k)} j_\ell^{(k)}, \Gamma_\ell(\Delta_N^\ell)\right] \, ,\\
    D_{I}^2&=\frac{1}{2}\sum_{k=1}^N
\left[\left[\Gamma_\ell(\Delta_N^\ell),W_t^N(\vec{r}) \right],\Gamma_\ell(\Delta_N^\ell)j_\ell^{\dagger\, (k)}j_\ell^{(k)}\right]\, .
}
\end{equation}
We first focus on the terms $D_I^1$. The first result that we proof is the following bound,
\begin{equation}\label{comm_localexponential}
    \| \left[\Gamma_\ell(\Delta_N^\ell), W_t^N(\vec{r})\right]  \| \leq \frac{\gamma(z)}{\sqrt{N}} \, ,
\end{equation}
with $\gamma(\cdot)$ defined in Assumption~\eref{Gamma}, and $z=\delta_\ell (1+e^{d^4 r_{\max} a_{\max}})$. Indeed, let us consider the commutator $ \left[\Gamma_\ell(\Delta_N^\ell), W_t^N(\vec{r})\right] = \left( \Gamma_\ell(\Delta_N^\ell)-W_t^{N}(\vec{r}) \Gamma_\ell(\Delta_N^\ell) W_t^{N \dagger}(\vec{r}) \right)W_t^{N}(\vec{r}) $. The second term in the right-hand side can be written as 
\begin{equation}\label{WgammaW}
\eqalign{
    W_t^{N}(\vec{r}) \Gamma_\ell(\Delta_N^\ell) W_t^{N \dagger}(\vec{r}) & = \sum_{b \geq 0 } c_\ell^b \left( W_t^N(\vec{r}) \Delta_N^\ell W_t^{N \dagger}(\vec{r}) \right)^b \\
    & = \sum_{b \geq 0 } c_\ell^b \left( \Delta_N^\ell +D_{\Delta} \right)^b  \\
    & =\Gamma_\ell(\Delta^{N}_\ell) +   \sum_{b \geq 0} c_{\ell}^b \sum_{s=1}^b  \left( \matrix{b \cr s} \right)(\Delta_N^\ell)^{b-s} D_{\Delta}^s\, , 
    }
\end{equation}
where in the first step we employed Assumption~\eref{Gamma}, and in the second step we introduced
\begin{equation}\label{S_definition}
\eqalign{
    & D_{\Delta} \equiv W_t^N(\vec{r})  \Delta_N^\ell W_t^{N \dagger}(\vec{r}) - \Delta_N^\ell = \sum_{n \geq 1} \frac{i^n}{n!}\mathbb{K}^n_{\vec{r}\cdot \vec{F}_t^N}[\Delta_N^\ell] \\
    & = \sum_{n \geq 1} \frac{i^n}{n!} \sum_{\nu_1 ,...,\nu_n}\sum_{\alpha} r_{\ell \alpha} r_{\nu_1}\cdots r_{\nu_n} \sqrt{N}^n [m^N_{\nu_1},\cdots,[m^N_{\nu_n},m_{\alpha}^N]\cdots] \\
    & = \sum_{n \geq 1} \frac{i^n}{n!} \sum_{\nu_1 ,...,\nu_n}\sum_{\alpha} r_{\ell \alpha} r_{\nu_1}\cdots r_{\nu_n} \frac{1}{N^{n/2}}\sum_{\gamma_1, \cdots ,\gamma_n}a_{\nu_n \alpha}^{\gamma_n} a_{\nu_{n-1} \gamma_n}^{\gamma^{n-1} }\cdots a_{\nu_1 \gamma_{2}}^{\gamma_1} m_{\gamma_1}^N \, .
    }
\end{equation}
The latter operator is normed-bounded as
\begin{equation}\label{norm_bound_S}
    \| D_{\Delta} \| \leq \frac{1}{\sqrt{N}}\delta_\ell \sum_{n \geq 1} \frac{1}{n!} (r_{\max} a_{\max} d^4)^n \leq \frac{1}{\sqrt{N}} \delta_\ell \, e^{d^4 r_{\max} a_{\max}} \, ,
\end{equation} 
so that for the commutator $\left[\Gamma_\ell(\Delta_N^\ell), W_t^N(\vec{r})\right]$ it is
\begin{equation}
    \eqalign{
    & \| \left[\Gamma_\ell(\Delta_N^\ell), W_t^N(\vec{r})\right] \| \leq \| \left( \Gamma_\ell(\Delta_N^\ell)-W_t^{N}(\vec{r}) \Gamma_\ell(\Delta_N^\ell)  W_t^{N \dagger}(\vec{r}) \right)\| \| W_t^{N}(\vec{r}) \| \\
    & \leq  \left\| \sum_{b \geq 0} c_{\ell}^b \sum_{s=1}^b \left( \matrix{b \cr s} \right) (\Delta_N^\ell)^{b-s} (D_{\Delta})^s   \right\| \\
    & \leq \frac{1}{\sqrt{N}}\sum_{b \geq 0} |c_\ell^b| \sum_{s=1}^b \left( \matrix{b \cr s} \right) \delta_\ell^{b-s}(\delta_\ell e^{d^4 r_{\max} a_{\max}}) ^s \\
    & \leq \frac{1}{\sqrt{N}} \sum_{b \geq 0} |c_\ell^b|\delta_\ell^b(1+e^{d^4 r_{\max} a_{\max}})^b \equiv \frac{1}{\sqrt{N}}\gamma( \delta_\ell (1+e^{d^4 r_{\max} a_{\max}})) \, ,
    }
\end{equation}
It is straightforward now to show that $D_I^1$ is a norm-bounded operator. Indeed, by exploiting Lemma \eref{lemma_commutators} and the bound \eref{comm_localexponential}, it is 
\begin{equation}
    \| D_I^1 \| \leq \frac{1}{\sqrt{N}} \|j_\ell \|^2  \delta_\ell \gamma'(\delta_\ell) \gamma( \delta_\ell (1+e^{d^4 r_{\max} a_{\max}})) \, . 
\end{equation}
Let us now consider the term $D_{I}^2$, which will be in turn split in the following two contributions
\begin{equation}\label{dIsplit}
    \eqalign{D_{I}^2&=\frac{1}{2}\sum_{k=1}^N
\left[\left[\Gamma_\ell(\Delta_N^\ell),W_t^N(\vec{r}) \right],\Gamma_\ell(\Delta_N^\ell)\right] j_\ell^{\dagger\, (k)}j_\ell^{(k)} \\
& + \frac{1}{2}\sum_{k=1}^N \Gamma_\ell(\Delta_N^\ell)
\left[\left[\Gamma_\ell(\Delta_N^\ell),W_t^N(\vec{r}) \right],j_\ell^{\dagger\, (k)}j_\ell^{(k)}\right] \\
& = D_{I}^{21} + D_{I}^{22} \, . 
}
\end{equation}
In order to tackle the first term of the last equation, $D_I^{21}$, let us consider the following commutator at the leading order
\begin{equation}\label{GammaWdouble}
    \eqalign{
   & \left[\left[\Gamma_\ell(\Delta_N^\ell),W_t^N(\vec{r}) \right],\Gamma_\ell(\Delta_N^\ell)\right] 
    \\ &= \left\lbrace \Gamma_\ell(\Delta_N^\ell)\left( \Gamma_\ell(\Delta_N^\ell)-W_t^{N}(\vec{r}) \Gamma_\ell(\Delta_N^\ell)  W_t^{N \dagger}(\vec{r}) \right) \right. \\
    & - \left. \left( \Gamma_\ell(\Delta_N^\ell)-W_t^{N}(\vec{r}) \Gamma_\ell(\Delta_N^\ell)  W_t^{N \dagger}(\vec{r}) \right)W_t^{N}(\vec{r}) \Gamma_\ell(\Delta_N^\ell)  W_t^{N \dagger}(\vec{r}) \right\rbrace W_t^{N }(\vec{r}) \\
    & = \left\lbrace \Gamma_\ell(\Delta_N^\ell) \sum_{b \geq 0} c_{\ell}^b \sum_{s=1}^b \left(  \matrix{b \cr s} \right)  (\Delta_N^\ell)^{b-s} (D_{\Delta})^s \right. \\
    & - \left. \left( \sum_{b \geq 0} c_{\ell}^b \sum_{s=1}^b \left( \matrix{b \cr s} \right)  (\Delta_N^\ell)^{b-s} (D_{\Delta})^s \right) W_t^{N}(\vec{r}) \Gamma_\ell(\Delta_N^\ell)  W_t^{N \dagger}(\vec{r}) \right\rbrace W_t^{N }(\vec{r}) \\
    & \sim  -\left( \sum_{b \geq 0} c_{\ell}^b \sum_{s=1}^b \left( \matrix{b \cr s} \right) (\Delta_N^\ell)^{b-s} (D_{\Delta})^s \right)^2 W_t^{N }(\vec{r})
    }
\end{equation}
where in the second last line we used the bound \eref{comm_localexponential}, so that
$$ \| [W_t^{N}(\vec{r}) \Gamma_\ell(\Delta_N^\ell)  W_t^{N \dagger}(\vec{r}), \sum_{b \geq 0} c_{\ell}^b \sum_{s=1}^b \left( \matrix{b \cr s} \right) (\Delta_N^\ell)^{b-s} (D_{\Delta})^s ] \| \leq O(N^{3/2}) \, .$$  
The squared term in the last line of Eq.~\eref{GammaWdouble} can be further treated. To do so, we first split it into two terms,
$$\sum_{b \geq 0} c_{\ell}^b \sum_{s=1}^b \left( \matrix{b \cr s} \right) (\Delta_N^\ell)^{b-s} (D_{\Delta})^s = \sum_{b \geq 0} c_{\ell}^b  b (\Delta_N^\ell)^{b-1} D_{\Delta}  + \sum_{b \geq 0} c_{\ell}^b \sum_{s=2}^b \left( \matrix{b \cr s} \right) (\Delta_N^\ell)^{b-s} (D_{\Delta})^s \, , $$
and we notice that, by  using Assumption~\eref{Gamma} and Eq.~\eref{norm_bound_S} one obtains
\begin{equation}\label{bound_DI222}
    \eqalign{
 & \|\sum_{b \geq 0} b c_{\ell}^b   (\Delta_N^\ell)^{b-1} D_{\Delta}  \|  \leq \frac{1}{\sqrt{N}} e^{d^4r_{\mathrm{max}} a_{\mathrm{max}}}\delta_{\ell}\gamma'(\delta_\ell) \,  ,
    } 
\end{equation}
as well as
\begin{equation}\label{bound_DI222_b}
\eqalign{
    & \|  \sum_{b \geq 0} c_{\ell}^b \sum_{ s=2 }^b \left( \matrix{b \cr s} \right) (\Delta_N^\ell)^{(b-s)}   \left(D_{\Delta} \right)^s   \| \\
    & \leq  \sum_{b \geq 0} |c_\ell^b| \sum_{s\geq 2} \left( \matrix{b \cr s}\right) \delta_\ell^{(b-s)} \left(\frac{1}{\sqrt{N}} \delta_\ell \, e^{d^4 r_{\max} a_{\max}} \right)^s  \\
    & \leq \frac{1}{N} \sum_{b \geq 0} |c_\ell^b| \delta_{\ell}^{b}\sum_{s=0}^b \left( \matrix{b \cr s} \right)   e^{s d^4 r_{\max} a_{\max}}  \leq \frac{\gamma(\delta_\ell(1+e^{d^4 r_{\max} a_{\max}}))}{N} \, ,
    }
\end{equation}
and also the following
\begin{equation}\label{bound_DI222_c}
\eqalign{
    & \left\|  \sum_{b \geq 0} b c_{\ell}^b  (\Delta_N^\ell)^{b-1}  \right.  \\
    & \left. \times  \left(\sum_{ n \geq 2
    } \frac{i^n}{n!} \sum_{\nu_1 ,...,\nu_n}\sum_{\alpha} r_{\ell \alpha} r_{\nu_1}\cdots r_{\nu_n} \frac{1}{N^{n/2}}\sum_{\gamma_1, \cdots ,\gamma_n}a_{\nu_n \alpha}^{\gamma_n} a_{\nu_{n-1} \gamma_n}^{\gamma^{n-1} }\cdots a_{\nu_1 \gamma_{2}}^{\gamma_1} m_{\gamma_1}^N \right)   \right\| \\
    & \leq \frac{1}{N}  \delta_\ell e^{d^4 r_{\max} a_{\max}} \sum_b b |c_\ell^b|  \delta{_\ell}^{(b-1)} \leq \frac{ \delta_\ell \gamma^{'}(\delta_\ell) e^{d^4 r_{\max} a_{\max}}}{N} \, .
    }
\end{equation}
By means the norm bounds in Eqs.~\eref{bound_DI222}-\eref{bound_DI222_c} it follows that
\begin{equation}
    \eqalign{
    & \| \left( \sum_{b \geq 0} c_{\ell}^b \sum_{s=1}^b \left( \matrix{b \cr s} \right) (\Delta_N^\ell)^{b-s} (D_{\Delta})^s \right)^2  \\
    & - 
    \left( \sum_{b \geq 0} c_{\ell}^b b (\Delta_N^\ell)^{b-1} i \sum_{\alpha} r_{\ell \alpha} \sum_{\nu} \frac{1}{\sqrt{N}}r_{\nu} \sum_{\gamma}a_{\nu\alpha}^{\gamma} m_{\gamma}^N \right)^2 \| \\
    & \leq  \frac{1}{N^2} \underbrace{ \left[ \gamma^2(\delta_\ell(1+e^{d^4 r_{\max} a_{\max}})) + \delta_\ell^2 \gamma^{' \, 2}(\delta_\ell) e^{2 d^4 r_{\max} a_{\max}} \right]}_{C_{211}} \\ & + \frac{1}{N^{3/2}} \underbrace{2\delta_{\ell}\gamma'(\delta_\ell)e^{d^4r_{\mathrm{max}} a_{\mathrm{max}}} \left[ \gamma(\delta_\ell(1+e^{d^4 r_{\max} a_{\max}}))     + \delta_\ell \gamma'(\delta_\ell) e^{ d^4 r_{\max} a_{\max}}\right]}_{C_{212}}
    \, ,} 
\end{equation}
where we defined $C_{211} \equiv \gamma^2(\delta_\ell(1+e^{d^4 r_{\max} a_{\max}})) + \delta_\ell^2 \gamma^{' \, 2}(\delta_\ell) e^{2 d^4 r_{\max} a_{\max}} $, and $C_{212} = 2\delta_{\ell}\gamma'(\delta_\ell)e^{d^4r_{\mathrm{max}} a_{\mathrm{max}}} \left[ \gamma(\delta_\ell(1+e^{d^4 r_{\max} a_{\max}}))     + \delta_\ell \gamma'(\delta_\ell) e^{ d^4 r_{\max} a_{\max}}\right]$.
As a result it is, at leading order,
\begin{equation}
    \eqalign{   D_{I}^{21} & \sim -\frac{1}{2} \sum_{k} \left( \sum_{b \geq 0} c_{\ell}^b b (\Delta_N^\ell)^{b-1} i \sum_{\alpha} r_{\ell \alpha} \sum_{\nu} \frac{1}{\sqrt{N}}r_{\nu} \sum_{\gamma}a_{\nu\alpha}^{\gamma} m_{\gamma}^N \right)^2 \\
    & \times W_t^{N }(\vec{r}) j_\ell^{\dagger (k)} j_\ell^{ (k)} W_t^{N \dagger }(\vec{r})W_t^{N }(\vec{r})  \\
    & \sim \frac{1}{2} \left(\Gamma_\ell'(\Delta_N^{\ell}) \sum_{\nu} \sum_{\ell \alpha} r_{\nu} a_{\nu \alpha}^{\gamma} r_{\ell \alpha } m_{\gamma}^N \right)^2 \frac{1}{N} \sum_{k}  j_\ell^{\dagger (k)} j_\ell^{ (k)} \\
    & = \frac{1}{2} (\vec{r} \cdot T^{N}\vec{r}_\ell) \left(\Gamma_\ell'(\Delta_N^{\ell}) \right)^2  \frac{1}{N} \sum_{k}  j_\ell^{\dagger (k)} j_\ell^{ (k)}  (\vec{r} \cdot T^{N}\vec{r}_\ell) \\
    & = - \frac{1}{2} (\vec{r} \cdot T^{N}\vec{r}_\ell) \left(\Gamma_\ell'(\Delta_N^{\ell}) \right)^2  \frac{1}{N} \sum_{k}  j_\ell^{\dagger (k)} j_\ell^{ (k)}  (\vec{r}_\ell \cdot T^{N}\vec{r}) \, ,
    }
\end{equation}
where the symbol $\sim$ is identifying here the following bound
\begin{equation}
\eqalign{
    &  \| D_{I}^{21} + \frac{1}{2} (\vec{r} \cdot T^{N}\vec{r}_\ell) \left(\Gamma_\ell'(\Delta_N^{\ell}) \right)^2  \frac{1}{N} \sum_{k}  j_\ell^{\dagger (k)} j_\ell^{ (k)}  (\vec{r}_\ell \cdot T^{N}\vec{r}) \| \\
    &\leq \frac{1}{N^{1/2}} \frac{C_{212}}{2}\| j_\ell \|^2 + \frac{1}{N} \frac{C_{211}}{2} \| j_\ell \|^2 +\frac{1}{N^{1/2}} \| j_\ell\|^{2} e^{d^2 2 v r_{\mathrm{max} }} (\delta_\ell\gamma'(\delta_\ell) e^{d^4 r_{\mathrm{max} a_{\mathrm{max} }}})^2 \, . \\ 
}
\end{equation}

Let us then consider the term $D_I^{22}$ of Eq.~\eref{dIsplit}, which we write here for the sake of clarity as
\begin{equation}
\eqalign{
D_I^{22} & = \frac{1}{2}\sum_{k=1}^N \Gamma_\ell(\Delta_N^\ell)
\left[\left[\Gamma_\ell(\Delta_N^\ell),W_t^N(\vec{r}) \right],j_\ell^{\dagger\, (k)}j_\ell^{(k)}\right] \\
& =  \frac{1}{2}\sum_{k=1}^N \Gamma_\ell(\Delta_N^\ell) \left[ \left( \Gamma_\ell(\Delta_N^\ell)-W_t^{N}(\vec{r}) \Gamma_\ell(\Delta_N^\ell) W_t^{N \dagger}(\vec{r}) \right)W_t^{N}(\vec{r}), j_\ell^{\dagger\, (k)}j_\ell^{(k)} \right] \\
& = \frac{1}{2}\sum_{k=1}^N \Gamma_\ell(\Delta_N^\ell)  \left[ \left( \Gamma_\ell(\Delta_N^\ell)-W_t^{N}(\vec{r}) \Gamma_\ell(\Delta_N^\ell) W_t^{N \dagger}(\vec{r}) \right), j_\ell^{\dagger\, (k)}j_\ell^{(k)} \right]W_t^{N}(\vec{r}) \\
& + \frac{1}{2}\sum_{k=1}^N \Gamma_\ell(\Delta_N^\ell)  \left( \Gamma_\ell(\Delta_N^\ell)-W_t^{N}(\vec{r}) \Gamma_\ell(\Delta_N^\ell) W_t^{N \dagger}(\vec{r}) \right) \left[ W_t^{N}(\vec{r}), j_\ell^{\dagger\, (k)}j_\ell^{(k)} \right] \\
&\equiv D_I^{221} + D_{I}^{222} \, . \\
  }  
\end{equation}
We start to inspect the second-last line term, which we denoted $D_I^{221}$. Its norm-bound, by using Eq.~\eref{WgammaW}, can be written as
\begin{equation}
\eqalign{
  & \| \frac{1}{2}\sum_{k=1}^N \Gamma_\ell(\Delta_N^\ell) \left[ \left( \Gamma_\ell(\Delta_N^\ell)-W_t^{N}(\vec{r}) \Gamma_\ell(\Delta_N^\ell) W_t^{N \dagger}(\vec{r}) \right), j_\ell^{\dagger\, (k)}j_\ell^{(k)} \right] \| \\
  &\leq  \frac{N}{2}\gamma(\delta_\ell) \sum_{b \geq 0} |c_{\ell}^b| \sum_{s=1}^b \left( \matrix{b \cr s} \right)  \left\| \left[ (\Delta_N^\ell)^{b-s} (D_{\Delta})^s ,  j_\ell^{\dagger\, (k)}j_\ell^{(k)} \right] \right\| \\
  & \leq  \frac{N}{2}\gamma(\delta_\ell) \sum_{b \geq 0} |c_{\ell}^b| \sum_{s=1}^b \left( \matrix{b \cr s} \right) \\ 
  & \times \left( \left\| \left[ (\Delta_N^\ell)^{b-s}  ,  j_\ell^{\dagger\, (k)}j_\ell^{(k)} \right]  \right\| \|D_{\Delta} \|^s + \| (\Delta_N^\ell)\|^{b-s} \left\| \left[  (D_{\Delta})^s ,  j_\ell^{\dagger\, (k)}j_\ell^{(k)} \right]  \right\| \right) \\
  & \equiv \|D_{I}^{221'}\| +\| D_{I}^{221''}\| \, .
   }
\end{equation}
Here, employing \eref{norm_bound_S}, the first term is norm-bounded as  
\begin{equation}
   \| D_{I}^{221'}\| \leq  \frac{1}{N^{1/2}}\| j_\ell\|^2 \delta_\ell  \gamma(\delta_\ell)\gamma'(\delta_\ell[1+e^{s d^4 r_{\max} a_{\max}}]) \, .
\end{equation}
With respect to the second term, $\| D_{I}^{221''}\|$, we use that $[(D_{\Delta})^{s},\cdot] = \sum_{j=0}^{s-1} (D_{\Delta})^j[D_{\Delta},\cdot](D_{\Delta})^{s-1-j}$, and that, by employing the definition of $D_{\Delta}$ in Eq.~\eref{S_definition}, it is
\begin{equation}
\eqalign{
  & [D_{\Delta},j^{(k),\dagger}_\ell j^{(k)}_\ell]\\
  &=  \sum_{n \geq 1} \frac{i^n}{n!} \sum_{\nu_1 ,...,\nu_n}\sum_{\alpha} r_{\ell \alpha} r_{\nu_1}\cdots r_{\nu_n} \frac{1}{N^{n/2}}\sum_{\gamma_1, \cdots ,\gamma_n}a_{\nu_n \alpha}^{\gamma_n} a_{\nu_{n-1} \gamma_n}^{\gamma^{n-1} }\cdots a_{\nu_1 \gamma_{2}}^{\gamma_1} [m_{\gamma_1}^N, j^{(k),\dagger}_\ell j^{(k)}_\ell] \\
   &=  \sum_{n \geq 1} \frac{i^n}{n!} \sum_{\nu_1 ,...,\nu_n}\sum_{\alpha} r_{\ell \alpha} r_{\nu_1}\cdots r_{\nu_n} \frac{1}{N^{n/2}}\sum_{\gamma_1, \cdots ,\gamma_n}a_{\nu_n \alpha}^{\gamma_n} a_{\nu_{n-1} \gamma_n}^{\gamma^{n-1} }\cdots a_{\nu_1 \gamma_{2}}^{\gamma_1} \frac{O_{\ell \gamma_1}^{(k)}}{N},
  } 
\end{equation}
with $O_{\ell \gamma_1}^{(k)}$ a local operator with finite support. Therefore,
\begin{equation}
    \eqalign{
    \| D_{I}^{221''} \| &\leq \frac{N\gamma(\delta_\ell)}{2} \sum_{b}|c_\ell^b|\sum_{s=1}^b \left( \matrix{b \cr s} \right) \delta_\ell^{b-s} \sum_j (\frac{1}{\sqrt{N}}\delta_\ell e^{d^4 r_{\max} a_{\max}})^{s}\|O_{\ell \gamma}\|\frac{1}{N} \\
    & \leq \frac{1}{2N^{1/2}}\| O_{\ell \gamma}\|  \delta_\ell e^{d^4r_{\max} a_{\max}} \gamma(\delta_\ell)\gamma'(\delta_\ell(1+e^{d^4 r_{\max} a_{\max}})) \, .
    }
\end{equation}
As a result, we see that the contribution $D_I^{221}$ in norm bounded as
\begin{equation}
    \| D_I^{221} \| \leq \frac{1}{N^{1/2}}  \delta_\ell  \gamma(\delta_\ell)\gamma'(\delta_\ell[1+e^{s d^4 r_{\max} a_{\max}}]) \left(\| j_\ell \|^2 + \frac{1}{2}\|O_{\ell \gamma} \| e^{d^4r_{\max} a_{\max}}\right) \, .
\end{equation}

In order to conclude the proof of the Lemma, we have to consider the remaining term, that is
\begin{equation*}
\eqalign{
    D_I^{222} & = \frac{1}{2}\sum_{k=1}^N \Gamma_\ell(\Delta_N^\ell)  \left( \Gamma_\ell(\Delta_N^\ell)-W_t^{N}(\vec{r}) \Gamma_\ell(\Delta_N^\ell) W_t^{N \dagger}(\vec{r}) \right) \\
    & \times \left(W_t^{N}(\vec{r})j_\ell^{\dagger\, (k)}j_\ell^{(k)}W_t^{N \dagger}(\vec{r})- j_\ell^{\dagger\, (k)}j_\ell^{(k)} \right) W_t^{N}(\vec{r}) \, ,\\
    }
\end{equation*}
Let us first consider 
\begin{equation}
\eqalign{
   &  \left( \Gamma_\ell(\Delta_N^\ell)-W_t^{N}(\vec{r}) \Gamma_\ell(\Delta_N^\ell) W_t^{N \dagger}(\vec{r}) \right)
     \left(W_t^{N}(\vec{r})j_\ell^{\dagger\, (k)}j_\ell^{(k)}W_t^{N \dagger}(\vec{r})- j_\ell^{\dagger\, (k)}j_\ell^{(k)} \right) \\
     & = - \left(\sum_{b \geq 0} c_{\ell}^b \sum_{s=1}^b \left( \matrix{b \cr s} \right) (\Delta_N^\ell)^{b-s} (D_{\Delta})^s\right) \left( \sum_{n\geq 1} \frac{i^n}{n!} \mathbb{K}_{\vec{r} \cdot \vec{F}_t^N }[ j_\ell^{\dagger\, (k)}j_\ell^{(k)}] \right) \\
    & = - \left(\sum_{b \geq 0} c_{\ell}^b b (\Delta_N^\ell)^{b-1} \frac{i}{\sqrt{N}}  \sum_{\alpha} r_{\ell \alpha} \sum_{\nu \gamma} r_{\nu} a_{\nu \alpha}^{\gamma} m_{\gamma}^N \right) \left( i [ \vec{r} \cdot \vec{F}_t^N ,j_\ell^{\dagger\, (k)}j_\ell^{(k)}] \right) + P \\
    & = \left(\Gamma' (\Delta_N^\ell)  \sum_{\alpha} r_{\ell \alpha} \sum_{\nu } r_{\nu} T_{\nu \alpha}^N \right) \left(\frac{1}{\sqrt{N}} [ \vec{r} \cdot \vec{F}_t^N ,j_\ell^{\dagger\, (k)}j_\ell^{(k)}] \right) + P  \, ,\\
    }
\end{equation}
where we have defined 
\begin{equation}
    P = - \left(\sum_{b \geq 0} c_{\ell}^b \sum_{s=1}^b \left( \matrix{b \cr s} \right) (\Delta_N^\ell)^{b-s} (D_{\Delta}|_{n \neq s \neq 1})^s \right) \left( \sum_{n\geq 2} \frac{i^n}{n!} \mathbb{K}_{\vec{r} \cdot \vec{F}_t^N }[ j_\ell^{\dagger\, (k)}j_\ell^{(k)}] \right) \, .
\end{equation}
Here, the notation $D_{\Delta}|_{n \neq s \neq 1}$ signals that when the term $s=1$ is considered, the sum over $n$ in the operator $D_{\Delta}$ [as defined by Eq.~\eref{S_definition}] starts from $n=2$.
By using the norm-bound 
\begin{equation}
\eqalign{
\left\|  
 \mathbb{K}^n_{\vec{r} \cdot \vec{F}^{N}_t}[j_\ell^{\dagger\, (k)}j_\ell^{(k)}]   \right\| 
 \leq (d^2 r_{\max} 2 v)^n \frac{\| j_\ell\|^2}{N^{n/2}} \, ,
 }
\end{equation}
and that, by exploiting Eqs.~\eref{bound_DI222_b} and \eref{bound_DI222_c}, the following bound holds,
\begin{equation}
\eqalign{
    & \left\| \sum_{b \geq 0} c_{\ell}^b \sum_{s=1}^b \left( \matrix{b \cr s} \right) (\Delta_N^\ell)^{b-s} (D_{\Delta}|_{n \neq s \neq 1})^s
    \right\| \\
    & \leq \frac{1}{N}\left[\gamma(\delta_\ell(1+ e^{ d^{4}r_{\max} a_{\max}} ))+ \delta_\ell \gamma'(\delta_\ell) e^{d^4 r_{\max} a_{\max}} \right] \, ,
    }
\end{equation}
then it is 
\begin{equation}
    \| P \| \leq \frac{1}{N^2} \| j_{\ell}\|^2 e^{2 d^{2}r_{\max} v}\left[ \gamma(\delta_\ell(1+ e^{ d^{4}r_{\max} a_{\max}} )) + \delta_\ell \gamma'(\delta_\ell) e^{d^4 r_{\max} a_{\max}}\right] \, .
\end{equation}
Therefore we can conclude the proof, by using $\sum_{k}\frac{1}{\sqrt{N}} [ \vec{r} \cdot \vec{F}_t^N ,j_\ell^{\dagger\, (k)}j_\ell^{(k)}] = \sum_{\mu} \sum_{\beta} r_{\mu} c_{\ell \beta} \frac{1}{N}\sum_k[v_{\mu}^{(k)}, v_{\beta}^{(k)}] = - \sum_{\mu \beta} c_{\ell \beta} T_{\beta \mu}^{N} r_{\mu} = -\vec{c}_{\ell} \cdot (T^{N} \vec{r} )$, 
\begin{equation}
\eqalign{
    \| D_I^{222} - & \frac{1}{2} \Gamma(\Delta_N^{\ell}) \Gamma'(\Delta_N^{\ell}) \vec{r}\cdot (T^N \vec{r}_\ell) (-\vec{c}_{\ell} \cdot (T^{N} \vec{r} ))\|  \\
    & \leq \frac{1}{2N}   \| j_{\ell}\|^2 e^{2 d^{2}r_{\max} v} \gamma(\delta_\ell) \left[ \gamma(\delta_\ell(1+ e^{ d^{4}r_{\max} a_{\max}} )) +\delta_\ell \gamma'(\delta_\ell) e^{d^4 r_{\max} a_{\max}}\right] \, ,
    }
\end{equation}
where we further used $j_{\ell}^{(k)} = \sum_{\gamma} c_{\ell \gamma}^{j} v_{\gamma}^{(k)}$, and introduced ${c}_{\ell \beta} = \sum_{\gamma, \gamma' } c_{\ell \gamma}^{j *} c_{\ell \gamma'}^{j } \frac{1}{2}(a_{\gamma \gamma'}^\beta + b_{\gamma \gamma'}^{\beta}) $, with $\lbrace v_{\gamma}^{(k)} ,v_{\gamma'}^{(k)}\rbrace =  \sum_{\beta} b^{\beta}_{\gamma \gamma'} v_{\beta}^{(k)}$.
\end{proof}
\subsection{Lemma 11}

\begin{lemma}\label{lemma10}
    Given a dynamical generator as defined by Eqs.~\eref{e0_totalHamiltonian}-\eref{jumps}, satisfying Assumption~\eref{Gamma}, the action of the dissipation term $D_{II}$ defined by Eq.~\eref{diss-second-term} on local exponentials  is such that
        \begin{equation*}
    \eqalign{
    & \lim_{N \rightarrow \infty} \omega_{\vec{r}_1 \vec{r}_2 }  \left( e^{t\linn} \left[ D_{II} \right] \right)= \\
    & = \lim_{N \rightarrow \infty} \omega_{\vec{r}_1 \vec{r}_2 }  \left( e^{t\linn}\left[ i \vec{r} \cdot \left( \Gamma_\ell^2(\Delta_\ell^N)M_\ell \vec{F}_t^N \right) + i \sqrt{N} \vec{r} \cdot \left( \Gamma_\ell^2(\Delta_\ell^N) M_\ell \vec{\omega}_t^N \right) \right.\right.\\
    & \left. \left. + \frac{1}{2 } \vec{r} \cdot \left( \Gamma_{\ell}^2(\Delta_\ell^N)M_\ell T^N\vec{r} \right) + \frac{1}{2}\vec{r} \cdot \left( \Gamma_\ell^2(\Delta_\ell^N) G^{N}_\ell \vec{r} \right)\right] W_{t}^N(\vec{r})
    \right) \, . }
    \end{equation*}
\end{lemma}
\begin{proof}
To start the proof, we define the following operators 
$$
\tilde{W}_1=\frac{1}{2}\sum_{k}\left[ j_\ell^{\dagger\, (k)}, W_t^{N} \right]j_{\ell}^{(k)}\, , \qquad  \tilde{W}_2=\frac{1}{2}\sum_{k}j_\ell^{\dagger\, (k)}\left[W_t^N , j_{\ell}^{(k)}\right]\, 
$$
which are such that $\tilde{W}_1+\tilde{W}_2=\mathcal{D}_\ell^{\rm Loc}[W_t^N]$, being the action of the local dissipator $\mathcal{D}_\ell^{\rm Loc}[\cdot]$ defined by Eq.~\eref{D_loc} . Through such operators we write
$$
D_{II}=\Gamma_\ell^2(\Delta_N^\ell)\mathcal{D}_\ell^{\rm Loc}[W_t^N(\vec{r})]+\Gamma_\ell(\Delta_N^\ell)\left(\left[\tilde{W}_1,\Gamma_\ell(\Delta_N^\ell)\right]+\left[\tilde{W}_2,\Gamma_\ell(\Delta_N^\ell)\right]\right)\, ,
$$
that allows us to split $D_{II}$ in the following terms
\begin{eqnarray}
    & D_{II}^1 = \Gamma_\ell^2(\Delta_N^\ell)\mathcal{D}_\ell^{\rm Loc}[W_t^N(\vec{r})] \, ,\\
    & D_{II}^2 = \Gamma_\ell(\Delta_N^\ell)\left(\left[\tilde{W}_1,\Gamma_\ell(\Delta_N^\ell)\right]+\left[\tilde{W}_2,\Gamma_\ell(\Delta_N^\ell)\right]\right) \, .
\end{eqnarray}
Let us start by considering the former term, $D_{II}^1 $, and explicitly derive $ \mathcal{D}_\ell^{\rm Loc}[W_t^N(\vec{r})]$, that is
\begin{equation}\label{D_local_W}
\eqalign{
    \mathcal{D}_\ell^{\rm Loc}[W_t^N(\vec{r})] & = \frac{1}{2}\sum_{k}\left( \left[ j_\ell^{\dagger\, (k)}, W_t^{N} \right]j_{\ell}^{(k)} + j_\ell^{\dagger\, (k)}\left[W_t^N , j_{\ell}^{(k)}\right]\right) \\
    & =  \frac{1}{2}\sum_{k} \left[ \left( j_\ell^{\dagger\, (k)}-  W_t^{N} j_\ell^{\dagger\, (k)}W_t^{N \dagger} \right) W_t^{N} j_\ell^{(k)}W_t^{N \dagger} \right. \\
    & \left. +  j_\ell^{\dagger\, (k)} \left( W_t^{N} j_\ell^{(k)}W_t^{N \dagger} - j_\ell^{ (k)}\right)   \right] W_t^N(\vec{r}) = \\
    &  = \frac{1}{2}\sum_{k} \left\lbrace \left( -\sum_{n \geq 1} \frac{i^n}{n!} \mathbb{K}_{\vec{r}\cdot \vec{F}_t^N} \left[ j_\ell^{\dagger \, (k)} \right] \right) \left( \sum_{n \geq 0} \frac{i^n}{n!} \mathbb{K}_{\vec{r}\cdot \vec{F}_t^N} \left[ j_\ell^{ (k)} \right]  \right) \right.  \\
    & + \left. j_\ell^{\dagger\, (k)} \left( \sum_{n \geq 1} \frac{i^n}{n!} \mathbb{K}_{\vec{r}\cdot \vec{F}_t^N} \left[ j_\ell^{ (k)} \right] \right) \right\rbrace W_t^N(\vec{r}) = \\
    &  = \frac{1}{2}\sum_{k} \left\lbrace \left( - i[\vec{r}\cdot \vec{F}_t^N,  j_\ell^{\dagger\, (k)}]+\frac{1}{2}\left[\vec{r}\cdot \vec{F}_t^N,[\vec{r}\cdot \vec{F}_t^N,   j_\ell^{\dagger\, (k)}]\right] -\tilde{S}_3[j_{\ell}^{\dagger \, (k)}] \right) \right.\\
    & \times \left. \left( j_\ell^{(k)} +  i[\vec{r}\cdot \vec{F}_t^N,  j_\ell^{ (k)}] + \tilde{S}_2[j_{\ell}^{ (k)}]   \right) \right.  \\ 
    & + \left. j_\ell^{\dagger\, (k)} \left( + i[\vec{r}\cdot \vec{F}_t^N,  j_\ell^{(k)}] - \frac{1}{2}\left[\vec{r}\cdot \vec{F}_t^N,[\vec{r}\cdot \vec{F}_t^N,   j_\ell^{ (k)}]\right] -\tilde{S}_3[j_{\ell}^{ (k)}] \right) \right\rbrace W_t^N(\vec{r}),
    }
\end{equation}
where $\tilde{S}_{\bar{n}}[O^{(k)}] \equiv \sum_{n \geq \bar{n}} \frac{i^n}{n!}  \mathbb{K}^{n}_{\vec{r} \cdot \vec{F}^N_t} [O^{(k)}] $. By using the norm-bound $ \| \tilde{S}_{\bar{n}} \|  \leq \|O \| e^{2d^{2}r_{\max} v} / \sqrt{N}^{\bar{n}}$, we thus have 
\begin{equation}\label{Dloc_onW}
    \eqalign{
     \mathcal{D}_\ell^{\rm Loc}[W_t^N(\vec{r})] & = \frac{1}{2}\sum_{k} \left( - i[\vec{r}\cdot \vec{F}_t^N,  j_\ell^{\dagger\, (k)}] j_\ell^{(k)} + j_\ell^{\dagger\, (k)}i[\vec{r}\cdot \vec{F}_t^N,  j_\ell^{(k)}] \right) \\
     &   +   \frac{1}{2}\sum_{k} \left(  \frac{1}{2}\left[\vec{r}\cdot \vec{F}_t^N,[\vec{r}\cdot \vec{F}_t^N,   j_\ell^{\dagger\, (k)}]\right] j_\ell^{(k)} - j_\ell^{\dagger\, (k)} \frac{1}{2}\left[\vec{r}\cdot \vec{F}_t^N,[\vec{r}\cdot \vec{F}_t^N,   j_\ell^{ (k)}]\right] \right.\\
     & \left. + [\vec{r}\cdot \vec{F}_t^N,  j_\ell^{\dagger\, (k)}][\vec{r}\cdot \vec{F}_t^N,  j_\ell^{(k)}] \right) + P\,
     }
\end{equation}
with 
\begin{equation}
\eqalign{
   \| P\| &  \leq  \frac{1}{2} N \left( \left\|\tilde{S}_{1}[j_\ell^{\dagger, (k)}] \right\| \, \left\| \tilde{S}_{2}[j_\ell^{ (k)}]  \right\| +2 \left\| \tilde{S}_{3}[j_\ell^{ (k)}]  \right\| \| j_\ell^{(k)}\| \right. \\
   & \left. + \|\frac{1}{2}\left[\vec{r}\cdot \vec{F}_t^N,[\vec{r}\cdot \vec{F}_t^N,   j_\ell^{\dagger\, (k)}]\right] [\vec{r}\cdot \vec{F}_t^N,  j_\ell^{(k)}] \| \right) \\
   & \leq \frac{1}{2\sqrt{N}} \|j_\ell \|^2 \left[ e^{4d^{2}r_{\max} v} + 2e^{2d^{2}r_{\max} v} + (2 d^2 r_{\max} v)^3 \right]
   }
\end{equation}
In order to express the leading order terms of $ \mathcal{D}_\ell^{\rm Loc}[W_t^N(\vec{r})]$ in a compact form, we rewrite the first line of Eq.~\eref{Dloc_onW} as
\begin{equation}
\eqalign{
    & \frac{1}{2}\sum_{k} \left( - i[\vec{r}\cdot \vec{F}_t^N,  j_\ell^{\dagger\, (k)}] j_\ell^{(k)} + j_\ell^{\dagger\, (k)}i[\vec{r}\cdot \vec{F}_t^N,  j_\ell^{(k)}] \right) \\
    & = i \sum_{\mu} r_{\mu}  \left( \frac{1}{2}\sum_{k}  \left(  \left[   j_\ell^{\dagger\, (k)}, {F}_{\mu}^N(t) \right] j_\ell^{(k)} + j_\ell^{\dagger\, (k)} [{F}_{\mu}^N(t),  j_\ell^{(k)}]  \right) \right)\\
    & =  i \sum_{\mu} r_{\mu}  \mathcal{D}_\ell^{\rm Loc}[F_{\mu}^N(t)] = i \sum_{\mu} r_{\mu}  \sum_{\nu} M_{\ell \mu}^{\nu} (F_{\nu}+ \sqrt{N}\omega_t^{\nu}) \,  .
     }
\end{equation}
Moreover, employing 
\begin{equation}
    \eqalign{
    &    \frac{1}{2}\left[\vec{r}\cdot \vec{F}_t^N,[\vec{r}\cdot \vec{F}_t^N,   j_\ell^{\dagger\, (k)}]\right] j_\ell^{(k)} - j_\ell^{\dagger\, (k)} \frac{1}{2}\left[\vec{r}\cdot \vec{F}_t^N,[\vec{r}\cdot \vec{F}_t^N,   j_\ell^{ (k)}]\right]  \\
    & = -\frac{1}{2}\left[\vec{r}\cdot \vec{F}_t^N,[   j_\ell^{\dagger\, (k)}, \vec{r}\cdot \vec{F}_t^N ]j_\ell^{(k)} \right]  -  \frac{1}{2}\left[\vec{r}\cdot \vec{F}_t^N, j_\ell^{\dagger\, (k)}[\vec{r}\cdot \vec{F}_t^N,   j_\ell^{ (k)}]\right] \, ,
    }
\end{equation}
the second line of Eq.~\eref{Dloc_onW} reads
\begin{equation}
\eqalign{
    &\frac{1}{2} \sum_{\mu} \sum_{\nu} r_{\mu} r_{\nu} \left[ \mathcal{D}_\ell^{\rm Loc}[F_{\mu}^N(t)], F_{\nu}^{N}(t) \right] = \\
    & \frac{1}{2} \sum_{\mu} \sum_{\nu} r_{\mu} r_{\nu} \sum_{\gamma} M_{\ell \mu}^{\gamma}[F_{\gamma}^{N}(t), F_{\nu}^{N}(t)] = \frac{1}{2} \sum_{\mu} \sum_{\nu} r_{\mu}  \sum_{\gamma} M_{\ell \mu}^{\gamma}T_{\gamma \nu}^{N} r_{\nu} \, ,\\
   }
\end{equation}
and, using $j_\ell^{(k)} = \sum_{\alpha}c_{\ell \alpha}^{j}v_{\alpha}^{(k)}$, the last line can be written as 
\begin{equation}
\eqalign{
  & \frac{1}{2} \sum_{k }  [\vec{r}\cdot \vec{F}_t^N,  j_\ell^{\dagger\, (k)}][\vec{r}\cdot \vec{F}_t^N,  j_\ell^{(k)}] \\
  &= -\frac{1}{2} \sum_{\mu \nu}r_{\mu}  r_{\nu}  \frac{1}{N} \sum_k \sum_{\alpha \beta} c^{j,*}_{\ell \alpha} c^{j}_{\ell \beta} [v_{\mu}^{(k)},v_{\alpha}^{(k)}][v_{\beta}^{(k)},v_{\nu}^{(k)}] \\
  &= -\frac{1}{2} \sum_{\mu \nu}r_{\mu}  r_{\nu}  \frac{1}{N} \sum_k \sum_{\alpha \beta} c^{j,*}_{\ell \alpha} c^{j}_{\ell \beta} a_{\mu \alpha}^{\gamma} a_{\beta \nu}^{\gamma'} v_{\gamma}^{(k)} v_{\gamma'}^{(k)} \\
  & = -\frac{1}{2} \sum_{\mu \nu}r_{\mu}  r_{\nu}  \sum_{\alpha \beta} \sum_{\gamma \gamma'} c^{j,*}_{\ell \alpha} c^{j}_{\ell \beta} a_{\mu \alpha}^{\gamma} a_{\beta \nu}^{\gamma'} d_{\gamma \gamma'}^{\eta} m_{\eta}^N\\
  }
\end{equation}
where we have further used $v_{\gamma }v_{\gamma'} =( \lbrace v_{\gamma }, v_{\gamma'} \rbrace + [v_{\gamma },v_{\gamma'}])/2 = \sum_{\eta} d_{\gamma \gamma'}^{\eta}m_{\eta}^N$. Before going further on, notice that, by construction, $ \frac{1}{2} \sum_{k }  [\vec{r}\cdot \vec{F}_t^N,  j_\ell^{\dagger\, (k)}][j_\ell^{(k)}, \vec{r}\cdot \vec{F}_t^N] $ is positive definite. For the sake of a lighter notation, we can define the matrix $G_{\ell, \mu \nu}^N  \equiv   \sum_{\alpha \beta} \sum_{\gamma \gamma'} a_{\mu \alpha}^{\gamma} c^{j,*}_{\ell \alpha}d_{\gamma \gamma'}^{\eta} m_{\eta}^N c^{j}_{\ell \beta}  a_{\beta \nu}^{\gamma'} $
By all the above contributions, we can finally write
\begin{equation}
    \eqalign{
    & \lim_{N \rightarrow \infty} \omega_{\vec{r}_1 \vec{r}_2 }  \left( e^{t\linn} \left[ D_{II}^{1} \right] \right) = \\
    &  \lim_{N \rightarrow \infty} \omega_{\vec{r}_1 \vec{r}_2 }  \left( e^{t\linn}\left[ i \vec{r} \cdot \left( \Gamma_\ell^2(\Delta_\ell^N)M_\ell \vec{F}_t^N \right) + i \sqrt{N} \vec{r} \cdot \left( \Gamma_\ell^2(\Delta_\ell^N)M_\ell \vec{\omega}_t^N \right) \right.\right.\\
    & \left. \left. + \frac{1}{2 } \vec{r} \cdot \left( \Gamma_\ell^2(\Delta_\ell^N)M_\ell T^N\vec{r} \right) 
    -\frac{1}{2}\vec{r} \cdot( \Gamma_\ell^2(\Delta_\ell^N) G_\ell^{N}\vec{r})\right] W_{t}^N(\vec{r})
    \right) \, .
    }
\end{equation}

To proceed further and evaluating the contributions coming from $D_{II}^2$, we notice that the latter can be written as
$   D_{II}^2 = \Gamma_\ell(\Delta_N^\ell) \left[\mathcal{D}_\ell^{\mathrm{Loc}}[W_t^{N}(\vec{r})],\Gamma_\ell(\Delta_N^\ell)\right] $. Thus, by employing the expression $\mathcal{D}_\ell^{\mathrm{Loc}}[W_t^{N}(\vec{r})]$ of Eq.~\eref{D_local_W} it is
\begin{equation}
\eqalign{
    D_{II}^2 & = \Gamma_\ell(\Delta_N^\ell) \\
    & \times \left\lbrace \left[  \left( i \vec{r} \cdot \left( M_\ell \vec{F}_t^N \right) + i \sqrt{N} \vec{r} \cdot \left( M_\ell \vec{\omega}_t^N \right) \right) W_t^N(\vec{r}), \Gamma_\ell(\Delta_N^\ell)  \right] \right. \\
    &  + \left[ \left( \frac{1}{2 } \vec{r} \cdot \left( M_\ell T^N\vec{r} \right) + \frac{1}{2}\vec{r} \cdot( G_\ell^{N}\vec{r}) \right)W_t^N(\vec{r}) , \Gamma_\ell(\Delta_N^\ell) \right] \\
    & \left. + [P W_t^N(\vec{r}), \Gamma_\ell(\Delta_N^\ell)  ]\right\rbrace \, .
    }
\end{equation}
The contribution in the last line can be bounded as 
\begin{equation}
\eqalign{
 & \| \Gamma_\ell(\Delta_N^\ell) \left[P W_t^N(\vec{r}), \Gamma_\ell(\Delta_N^\ell) \right] \|  \leq  2\gamma^2(\delta_\ell)\| P \| \\
 & \leq  \frac{\gamma^2(\delta_\ell)}{\sqrt{N}} \|j_\ell \|^2 \left[ e^{4d^{2}r_{\max} v} + 2e^{2d^{2}r_{\max} v} + (2 d^2 r_{\max} v)^3 \right] \, .
 }
\end{equation}
Similarly, by employing Lemma \eref{lemma_commutators} and the bound \eref{comm_localexponential}, for the second-last line it is 
\begin{equation}
    \eqalign{
   & \left\| \Gamma_\ell(\Delta_N^\ell) \left[ \left( \frac{1}{2 } \vec{r} \cdot \left( M_\ell T^N\vec{r} \right) + \frac{1}{2}\vec{r} \cdot( G_\ell^{N}\vec{r}) \right) W_t^N(\vec{r}), \Gamma_\ell(\Delta_N^\ell) \right] \right\| \\
   & \leq \gamma(\delta_\ell)r_{\max}^2 \left( | M_\ell| a_{\max} +|G_\ell|   \right) \left(\frac{1}{2} \frac{1}{\sqrt{N}}\gamma( \delta_\ell (1+e^{d^4 r_{\max} a_{\max}})) + \frac{1}{N}\delta_\ell \gamma'(\delta_\ell)  \right) \,.
    }
\end{equation}
Finally, we are left with the following term, 
\begin{equation}
    \eqalign{
    & \Gamma_\ell(\Delta_N^\ell)  \left[  \left( i \vec{r} \cdot \left( M_\ell \vec{F}_t^N \right)  \right) W_t^N(\vec{r}), \Gamma_\ell(\Delta_N^\ell)  \right]\\
    &= i \Gamma_\ell(\Delta_N^\ell) \sum_{\mu \nu}r_{\mu}M_{\ell \mu}^{\nu}\left[F_{\nu}^{N}(t),  \Gamma_\ell(\Delta_N^\ell) \right]W_t^{N}(\vec{r}) \\
    & + i \Gamma_\ell(\Delta_N^\ell) \sum_{\mu \nu}r_{\mu}M_{\ell \mu}^{\nu} F_{\nu}^{N}(t) \left[W_t^{N}(\vec{r}), \Gamma_\ell(\Delta_N^\ell) \right] \,.
    }
\end{equation}
The second line of the latter expression can be bounded, using that 
\begin{equation}
    \| \left[F_{\nu}^{N}(t),  \Gamma_\ell(\Delta_N^\ell) \right] \| \leq \frac{2}{\sqrt{N}}\delta_\ell \gamma'(\delta_\ell)\, ,
\end{equation}
as
\begin{equation}
   \eqalign{ 
   & \|\Gamma_\ell(\Delta_N^\ell) \sum_{\mu \nu}r_{\mu}M_{\ell \mu}^{\nu}\left[F_{\nu}^{N}(t),  \Gamma_\ell(\Delta_N^\ell) \right]W_t^{N}(\vec{r}) \| \\
   & \leq \frac{2}{\sqrt{N}} \delta_\ell \gamma'(\delta_\ell) \gamma(\delta_\ell) d^4 r_{\max}|M_\ell| \, .
   }
\end{equation}
Finally, using Eq.~\eref{comm_localexponential}, the result \eref{statement_lemmB2}, and Lemma \eref{lemma_squarefluctu_convergence} it is
\begin{equation}
\eqalign{
  &   \lim_{N \rightarrow} \omega_{\vec{r}_1 \vec{r}_2} \left(e^{t\linn}\left[ \Gamma_\ell(\Delta_N^{\ell})F_{\nu}^N(t) \left[W_t^{N}(\vec{r}), \Gamma_\ell(\Delta_N^\ell) \right] \right] \right) \\
  & =  \lim_{N \rightarrow} \omega_{\vec{r}_1 \vec{r}_2} \left(e^{t\linn}\left[ F_{\nu}^N(t) \Gamma_\ell(\Delta_N^{\ell})\left[W_t^{N}(\vec{r}), \Gamma_\ell(\Delta_N^\ell) \right] \right] \right) = \\
  & \leq \lim_{N\rightarrow +\infty} \gamma^2(\delta_\ell) \| \left[W_t^{N}(\vec{r}), \Gamma_\ell(\Delta_N^\ell) \right] \|^2 \sqrt{\omega\left( W^{N}(\vec{r}_1) \left[ \left(F_{\nu}^N(t)\right)^2 \right]  W^{N \dagger}(\vec{r}_1) \right)} = 0 \, ,
    }
\end{equation}
which concludes the proof of the Lemma.
\end{proof}

\section{Analysis of a quantum Hopfield Neural Network}

\subsection{Mean-field equations }
\label{C:HNN_mean Field}
In this Appendix, we derive the mean-field equations for the model presented in Section \ref{sec5}. We consider a system made of $N$ spin $1/2$-particles,  which evolves according to a Markovian quantum master equation, characterized by jumps operators %
$$
J_{\pm}^{(k)} = \hat{\sigma}_{\pm}^{(k)}\Gamma_{\pm}(\Delta E), \qquad \Gamma_{\pm}(\Delta E)= \frac{ e^{\pm \frac{\beta}{2}\Delta E }}{ \sqrt{2}} \, , \qquad \Delta E= \frac{1}{N}\sum_{\mu=1}^{p}\xi^{\mu}_{k}\sum_{j } \xi_{j}^{\mu} \hat{\sigma}^{(j)}_{z} \, ,$$
and Hamiltonian
$$H=\Omega\sum_{i=1}^{N} \hat{\sigma}^{(i)}_x \, .$$
As already commented in the main text, the variables $\lbrace \xi_{k}^{\mu} \rbrace$ are i.i.d. random variables, with $\xi_{k}^{\mu} = \pm 1$ and probability distribution $\mathbb{P}(\xi_{k}^{\mu} = \pm 1) = \frac{1}{2}$. By applying the large-spin mapping described in Section \ref{S:large-spin_Mapping}, we consider the mean-field equations related to the the average operators $m^{N_s}_{\alpha, k}$, which reads
\begin{equation}\label{eq:mf-largespinHNN}
\eqalign{
    & \dot{m}_{k,z}(t) = 2\Omega m_{k,y}(t)-\cosh{(\beta \Delta E_{\Lambda_k}(t))} m_{k,z}(t)+\sinh{(\beta \Delta E_{\Lambda_k}(t))} \, ,\\
    & \dot{m}_{k,y}(t) = -2\Omega m_{k,z}(t)-\frac{1}{2}\cosh{(\beta \Delta E_{\Lambda_k}(t))}m_{k,y}(t) \, ,\\
    & \dot{m}_{k,x}(t) =-\frac{1}{2}\cosh{(\beta \Delta E_{\Lambda_k}(t)) }m_{k,x}(t) \, .\\
    }
\end{equation}
It is worth noticing that the equation for the variable $\dot{m}_{k,x}(t)$ decouples. In the following, without loss of generality, we will consider an initial condition $m_{k,x}(0)=0$ and focus on the equations of motion along the $y$ and $z$ directions. 

We remind that we aim at characterizing, in the long time limit behaviour, the retrieval property of the system. To this purpose, it is more useful to consider the equations of motion for the \textit{overlap} variables. For the sake of completeness, we introduce the latter starting from the corresponding operator formulation, that reads
$$ o^{\mu,N}_{\alpha} = \frac{1}{N}  \sum_{i=1}^{N} 
     \xi_{i}^{\mu}\sigma_{\alpha}^{(i)} = \frac{1}{2^p} \sum_{k=1}^{2^p}e_{k}^{\mu} \frac{S_{\alpha}^{(k)}}{N_s} \, ,$$
for $\mu=1,...,p$, having employed the large-spin mapping in the last step. It is worth noticing that the expectation value of these operators quantifies the amount of alignment of the spin $1/2$ system with respect to to the (classical) configuration defined by the $\mu$-th pattern, $(\xi_{1}^{\mu},..., \xi_{N}^{\mu})$ [or, equivalently, of the system of $2^p$ large-spins with respect to the $\mu$-th configuration $(e_{1}^{\mu},..., e_{2^p}^{\mu})$]. The equations of motion for the overlap parameters $o^{\mu}_{\alpha} = \lim_{N \rightarrow + \infty} \omega_t(o^{\mu,N}_{\alpha})$ read
\begin{equation}\label{eq_overlap}
\eqalign{
    & \dot{o}_z^{\mu}(t) = 2\Omega o_{y}^{\mu} (t) - \braket{\braket{ \cosh{(\beta \vec{\xi}\cdot \vec{o}_z(t) )} }} {o}_z^{\mu}(t) + \braket{\braket{\xi^{\mu} \sinh{ (\beta \vec{\xi}\cdot \vec{o}_z(t) )}}} \, , \\
    & \dot{o}_y^{\mu}(t) = -2\Omega o_{z}^{\mu}(t) - \frac{1}{2}\braket{\braket{ \cosh{(\beta \vec{\xi}\cdot \vec{o}_z(t) )} }} {o}_y^{\mu}(t)\, .
    } \\
\end{equation}
The above equations have been further simplified under the self-averaging hypothesis \cite{Rotondo:JPA:2018, GrensingK86, KochP_JSP_89}, allowing us to employ the substitution $\frac{1}{N} \sum_{i}f(\xi_i^{\mu}) \rightarrow \frac{1}{2^p} \sum_{\mathbb{\xi}}\mathbb{P}(\xi) f{(\mathbb{\xi})} \equiv \braket{\braket{f{(\mathbb{\xi}})}}$, wit h $f(\cdot)$ a generic function. We also employed the substitution $\omega_t(\sigma_i^{\alpha}) = \xi_i^{\mu} o^{\mu}_\alpha(t)$, assuming a homogeneous distribution of the misalignment between patterns and spins \cite{Fiorelli:PRA:2019, FiorelliLM22, BoedekerFM23}. The equations \eref{eq_overlap} are symmetric with respect to $(i)$ variable permutation $(o_{z}^{\mu},o_{y}^{\mu}) \leftrightarrow (o_{z}^{\nu},o_{y}^{\nu})$, and $(i)$ mode inversion $(o_{z}^{\mu},o_{y}^{\mu}) \leftrightarrow (-o_{z}^{\mu},-o_{y}^{\mu})$, and therefore the dynamics displays the following invariant subspaces: $(o_{z}^{\mu},o_{y}^{\mu}) = (0,0)$,  $(o_{z}^{\mu},o_{y}^{\mu}) = (o_{z}^{\nu},o_{y}^{\nu})$, $(o_{z}^{\mu},o_{y}^{\mu}) = (-o_{z}^{\nu},-o_{y}^{\nu})$. 

First of all we consider the paramagnetic solutions, $(o_{z}^{\mu},o_{y}^{\mu}) = (0,0)$ $\forall \mu$, analyzing their stability. We linearize the EoMs \eref{eq_overlap}, with respect to the solution $o_{\alpha}^{\mu} = 0+\delta o^{\mu}_{\alpha} $, $\alpha=z,y$, obtaining
\begin{equation}
    \left( 
    \matrix{ \delta \dot{o}^{\mu}_z \cr
    \delta \dot{o}^{\mu}_y }
    \right) =
    \left(
    \matrix{
    \beta -1 && 2\Omega \cr
    -2 \Omega && -\frac{1}{2}
    }
    \right)
    \left( 
    \matrix{ \delta o^{\mu}_z \cr
    \delta o^{\mu}_y
    }
    \right) \, .
\end{equation}
The stability matrix coincides with the one analyzed in Ref. \cite{Rotondo:JPA:2018}. We summarize here the results, for the sake of completeness. The eigenvalues reads
$\lambda_{\pm} = \frac{1}{2}[\beta - \frac{3}{2} \pm \sqrt{(\beta -\frac{1}{2})^2 - 16\Omega^2}]$, and give rise to the following different regimes:
\begin{itemize}
    \item[$(V)$] $ \lbrace \beta < 1+8\Omega^2 \rbrace \cap \lbrace \beta < 3/2 \rbrace \cap \lbrace |\beta - 1/2| > 4\Omega \rbrace $,  $\lambda_{\pm} \in \mathbb{R}^-$, the fixed point is stable;  
    \item[$(IV)$]$ \lbrace \beta < 1+8\Omega^2 \rbrace \cap \lbrace \beta < 3/2 \rbrace \cap \lbrace |\beta - 1/2| < 4\Omega \rbrace $, $\lambda_{\pm} \in \mathbb{C}$, with $\mathrm{Re}(\lambda_{\pm}) \in \mathbb{R}^-$, the fixed point is stable and spiralizing. 
    \item[$(III)$] $ \lbrace \beta < 1+8\Omega^2 \rbrace \cap \lbrace \beta > 3/2 \rbrace \cap \lbrace \beta  < 4\Omega +1/2 \rbrace $, $\lambda_{\pm} \in \mathbb{C}$, with $\mathrm{Re}(\lambda_{\pm}) \in \mathbb{R}^+$, the fixed point is unstable and spiralizing. 
    \item[$(II)$] $ \lbrace \beta < 1+8\Omega^2 \rbrace \cap \lbrace \beta > 3/2 \rbrace \cap \lbrace \beta  > 4\Omega +1/2 \rbrace $, $\lambda_{\pm} \in \mathbb{C}$, with $\mathrm{Re}(\lambda_{\pm}) \in \mathbb{R}^+$, the fixed point is unstable.
    \item[$(I)$] $  \beta > 1+8\Omega^2$, $\lambda_{\pm} \in \mathbb{R}$, with different sign, the fixed point is a saddle point.
\end{itemize}

It is worth noticing that the change in the sign of $\mathrm{Re}[\lambda_{\pm}]$ between regions $(III)$ and $(IV)$ identifies a Hopf bifurcation, that signals limit cycle solutions, already found by Refs. \cite{Rotondo:JPA:2018, FiorelliLM22}. 

Let us now turn our attention to the ferromagnetic, or retrieval, solutions, that have the form $o^{\nu}_\alpha = \bar{o}_\alpha \delta_{\nu \mu} $  for a specific $\mu=1,...,p$. By employing Eqs. \eref{eq_overlap}, the stationary equations for these type of solutions read
\begin{equation}\label{c:statRet}
    \left[1+8\frac{\Omega^2}{\cosh^2{(\beta \bar{o}_z)}} \right] \bar{o}_z = \tanh{(\beta \bar{o}_z)} \, ,
\end{equation}
which admit finite values solution for $(1+8 \tilde{\Omega}^2)<\beta$, where $\tilde{\Omega} = \Omega/ \cosh{(\beta \bar{o}_z)}$. The stability analysis can be performed by linearizing Eqs. \eref{eq_overlap} with respect to the solutions $o^{\nu}_{\alpha} = \bar{o}_{\alpha} \delta_{\nu \mu} + \delta o_{\alpha}^{\nu} $, further taking $\mu=1$ without loss of generality. The stability matrix reads in this case
\begin{equation}
    \left(  \matrix{
    \cosh{(\beta \bar{o}_z)}\left[ \beta\left( 1- \bar{o}_z  \tan{(\beta \bar{o}_z)} \right)-1 \right] && 2\Omega \cr
    -2\Omega( 1-\beta {o}_z  \tan{(\beta \bar{o}_z)}) && -\frac{1}{2} \cosh{(\beta \bar{o}_z)} \,
    }
    \right)
\end{equation}
with eigenvalues that can be written as
$$
\eqalign{
\lambda_{ \pm} & = \frac{\cosh{(\beta \bar{o}_z)}}{2} \left\lbrace \beta' -\frac{3}{2} \pm \sqrt{(\beta'-\frac{3}{2})^2 + 2(\beta'-1+8 \tilde{\Omega}^2(\beta-1-\beta')} \right\rbrace  \\
& = \frac{\cosh{(\beta \bar{o}_z)}}{2} \left\lbrace \beta' -\frac{3}{2} \pm \sqrt{(\beta')^2 - (1+16\tilde{\Omega}^2)\beta'+16 \tilde{\Omega}^2(\beta-1)+\frac{1}{4}} \right\rbrace \, .  \\
}
$$
where we introduced $\beta' = \beta(1-\bar{o}_z \tanh{(\beta \bar{o}_z)}).$ For the parameter regime $\beta > 1+8\tilde{\Omega}$ we have numerically evaluated the eigenvalues $\lambda_{\pm}$ with respect to the retrieval solutions, obtaining $\lambda_{\pm} \in \mathbb{R}^{-}$, and thus stable retrieval solutions. 

The resulting phase diagram is illustrated in Fig.~\ref{fig:phaseD}. It displays a high-temperature paramagnetic phase, and a low-temperature retrieval phase, taking place for values of $\beta$ and $\Omega$ such that $\beta>(1+8\tilde{\Omega}^2)$. There exists an intermediate region where stationary configurations feature self-sustained oscillations. Such a parametric region ranges from the parameters at which the Hopf bifurcation takes place, to the parameter values corresponding to stable retrieval solutions, $\beta < 3/2$, $\Omega > 1/4$,  $\beta < 1+8\tilde{\Omega}^2$. We further performed numerical analysis to unveil the characteristic of this phase, which enforce the suggestion that it correspond to a limit-cycle phase. Figure \ref{fig:LC_SD}(a) shows the  standard deviation 
$$o_{\mathrm{sd}} = \sqrt{\frac{1}{N_I}\sum_{i \in I}[o_z(t_i)-\bar{o}_z]^2} \, , $$
of the overlap component $o_z(t)$ with respect to the corresponding long-time average, say $\bar{o}_z$. Here, $\lbrace o_{z}(t_i)\rbrace_{i=1}^{N_I}$ are obtained at times $t \in I$, being $I$ a time-interval chosen at long times. We set $N_I=2000$, $I=[9 \times 10^3, 10^4]$. The upper and lower panel in Fig.~\ref{fig:LC_SD}(b) show the parametric plot and flux diagram of the vector field $( o_z(t), o_{y}(t))$ at long times. In the lower panel, results are shown at $T=0.6$, $\Omega=0.35$, corresponding to the retrieval phase. Here the convergence to two stable solutions is illustrated. The upper panel is obtained for values of $(T,\Omega)$ in the limit-cycle phase, $T=0.6$, $\Omega=0.6$, and it displays a stable orbit. 
\begin{figure}
    \centering
\includegraphics[width=0.8\linewidth]{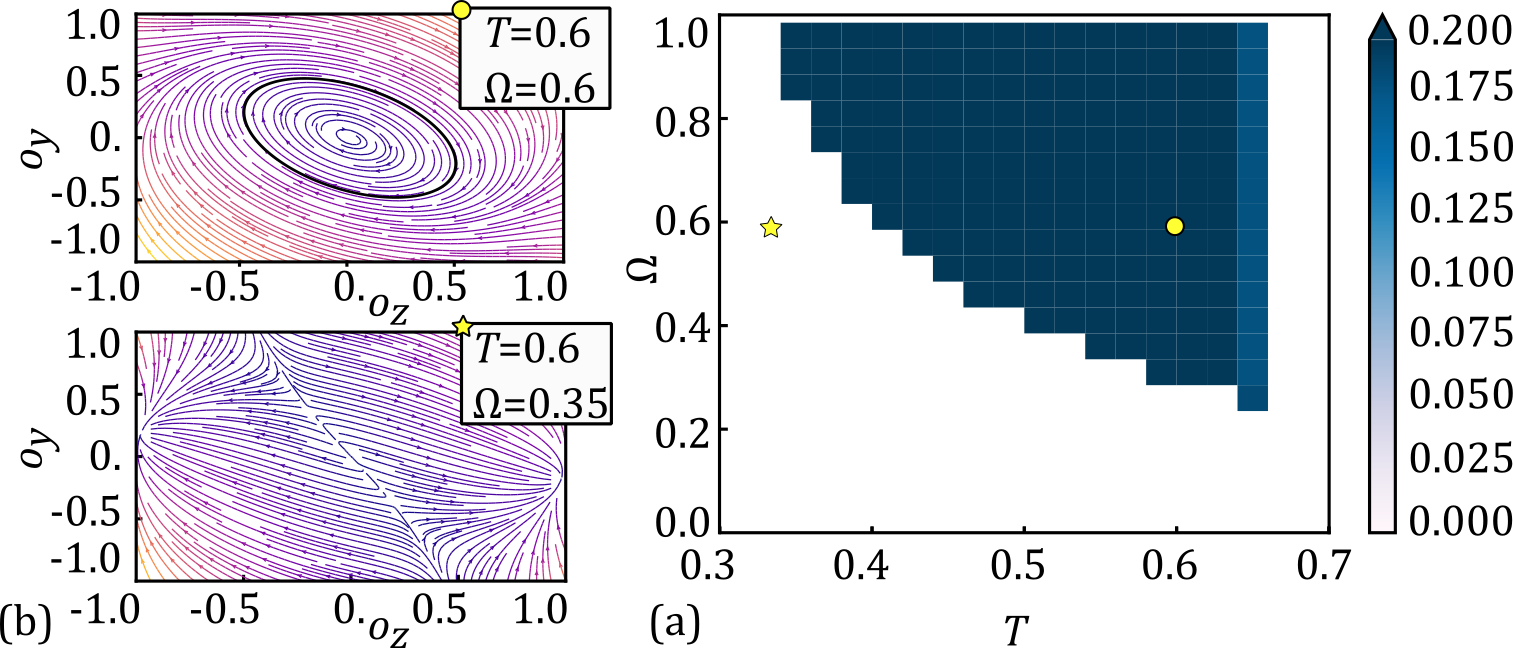}
    \caption{\textbf{Limit-cycle phase}. (a) For the case $p=1$, we display the standard deviation $o_{\mathrm{sd}}$ of the overlap $o_z(t)$ with respect to the stationary solutions, at long times. We set initial conditions such that $o_{\alpha}(0) \approx 0$, for $\alpha = y, z$. In the blue region $o_{\mathrm{sd}}$ is finite, signalling a limit-cycle phase. (b) With $T=0.6$, the two panels display the parametric plot and the flux diagram of the vector field $(o_y(t),o_z(t))$, for $\Omega=0.35$ (bottom panel), and $\Omega=0.6$ (upper panel). The latter shows a closed orbit, that in general characterizes limit-cycle phases. Colors for the flux diagram identify the norm of the corresponding vector field, which increases from purple to yellow. }
    \label{fig:LC_SD}
\end{figure}

\subsection{Derivation of the equation for the covariance matrix}
\label{A:covarianceM_HNN}

By exploiting Theorem \eref{theorem-mesoscopicdynamics} and the general results of Section \ref{sec4}, we derive the time evolution of the covariance matrix related to the system open quantum generalized HNN. For the sake of completeness, and for deriving some considerations (as it will be clearer at the end of the Section), we add here also an additional, all-to-all Hamiltonian of the type 
$H_2= \frac{g}{2}\frac{1}{2^{p}N_{s}}\sum_{h,k} \tilde{w}_{h,k} S^{(h)}_z S^{(k)}_z$.
The time evolution of the covariance matrix reads
\begin{equation}
\dot{\Sigma}^{\omega} = Q\Sigma^{\omega}+\Sigma^{\omega}Q^T+\sigma S^{\mathrm{sym}}\sigma^T + \Gamma^{2}G^{\mathrm{sym}} \, ,
\end{equation}
where we assume to consider an initial state that is stationary with respect to the mean-field variables.
The matrices appearing in the equation of the covariance matrix read 
$${Q} \equiv  \vec{f}(\vec{\omega}) - i\sigma(\vec{\omega}) 2ih^{\mathrm{R}} \, , $$
$$\Gamma^2 G^{\mathrm{sym}} = \sum_\ell\sum_{k'} \Gamma_\ell^2(\Delta_{\ell}^{k'}) G_\ell^{\mathrm{sym} \, k'} \, , $$
$$[G^{\mathrm{sym} \, k'}_\ell]_{(\mu h) (\nu k)} = \delta_{hk} \delta_{hk'} \left[ \sum_{\alpha \beta} c^{j,*}_{\alpha} c^{j}_{\beta} \sum_{\gamma \gamma' } a_{\mu \alpha}^{\gamma}  a_{\nu \beta }^{\gamma'} \frac{1}{2}\sum_{i}\left( \lbrace v_{\gamma}^{(i)}, v_{\gamma'}^{(i)} \rbrace + [v_{\gamma}^{(i)}, v_{\gamma'}^{(i)}] \right)\right]^{\mathrm{sym}}   \,  ,$$
$$\eqalign{
[S^{\mathrm{sym}}]_{(\mu h ),(\nu k)} & =  \frac{1}{2}\sum_\ell \sum_{k'} (r_{\ell \mu}^{h k'} {r}_{\ell \nu}^{k' k} + r_{\ell \nu}^{h k'} {r}_{\ell \mu}^{k' k} ) \left(\Gamma'(\Delta^{k'}_{\ell}) \right)^2   \omega_t\left(\frac{1}{N_s}\sum_{i \in \Lambda_{k'}}  j_\ell^{\dagger (i)} j_\ell^{ (i)}   \right)  \\
& + \frac{1}{2} \sum_\ell \sum_{k'} \Gamma(\Delta_{\ell}^{k'}) \Gamma'(\Delta_{\ell}^{k'})  (r_{\ell \alpha}^{hk'} \delta_{k'k}c_{\ell \beta} + r_{\ell \beta}^{k'k}\delta_{k'h} c_{\ell \alpha}  ) \, .
}$$
We will now specify the terms of the above expression in relation to the model we are considering, starting from the matrix $Q$. We remind that the matrix form $f_{\alpha}$ in the contributions $\vec{f}(\vec{\omega})$ is formally given by Lemma \ref{lemma_gen_action}. We further employ that, for the two-body interaction Hamiltonian $H_2$, it is $h_{(\alpha h) \, (\beta, k) } = \delta_{\alpha \beta} \delta_{\beta z} \frac{g}{2}\tilde{w}_{h,k}$; the on-site jump operators reads $j_{\pm}^{(k)} = \sigma_{\pm}^{(k)}$; and the structure constant of the single-particle algebra read $a_{\alpha \beta}^{\gamma} = 2i\mathcal{E}_{\alpha \beta \gamma}$, with $ \mathcal{E}_{\alpha \beta \gamma }$ the Levi-Civita tensor. Thus, the action of the local dissipator $\mathcal{D}_{\pm}^{\mathrm{Loc}}[v_{\alpha}^{(i)}] = \sum_{\beta=1}^{d^2}M_{l (\alpha h)}^{\beta k} $, on elements of the basis reads
\begin{eqnarray}
& \mathcal{D}_{\pm}^{\mathrm{Loc}}[\sigma_{x}^{(i)}] = -\frac{1}{2} \sigma_{x}^{(i)} \, , \\
& \mathcal{D}_{\pm}^{\mathrm{Loc}}[\sigma_{y}^{(i)}] =  -\frac{1}{2} \sigma_{y}^{(i)} \, , \\
& \mathcal{D}_{\pm}^{\mathrm{Loc}}[\sigma_{z}^{(i)}] = \pm \mathbb{I}-\sigma_z^{(i)} \,. \\
\end{eqnarray}
This defines a real matrix $M$ such that
\begin{equation}
M_{l (\alpha h)}^{(\beta k)} = \delta_{h k}\left[ \delta_{\alpha \beta} \left(-\frac{1}{2}\delta_{x\alpha}-\frac{1}{2}\delta_{y\alpha} -\delta_{\alpha z}\right) + \ell \delta_{\alpha z} \delta_{\beta 4}  \right]\, .
\end{equation}
It is worth noticing that, as far as it concerns the contribution to the quantum fluctuations, we consider the components $\alpha = x,y,z$, the identity playing the role of an irrelevant operator, as it commutes with the other fluctuation operators. As a consequence, the action of the mean-field equations of motion on these components can be written as
\begin{eqnarray*}
f_{(\alpha \, h) \, (\beta \, k)} & =  i  A_{(\alpha h),  (\beta,k)}  + i \sum_{\gamma, h'} B_{(\alpha h), (\gamma h') (\beta k)} m_{\gamma h'} \\
&+ \sum_{\ell}  \Gamma^{2}_{\ell}(\Delta_{\ell}^{h}) M_{\ell (\alpha h) }^{(\beta k)} \, ,  \\
\end{eqnarray*} 
where
\begin{eqnarray*}
 & A_{(\alpha \, h), (\beta \, k)} = \Omega \sum_{\beta'=x,y,z}   2i \delta_{\beta' x}\mathcal{E}_{ \beta' \alpha \beta } \delta_{k,h} \, , \\
 & B_{(\alpha,h) (\gamma,h') (\beta k)} = 2i \delta_{\gamma \gamma'} \delta_{\gamma' z} \mathcal{E}_{\gamma' \alpha \beta} (g\tilde{w}_{h h'}) \delta_{hk}\, , 
 \end{eqnarray*} 
\begin{eqnarray*}
 \sum_{\ell}  \Gamma^{2}_{\ell} (\Delta_{\ell}^{h}(t)) M_{\ell (\alpha h) }^{(\beta k)} & = \sum_{\ell=\pm}  \Gamma^{2}_{\ell} (\Delta E_{\Lambda_h}) M_{\ell (\alpha h) }^{(\beta k)} \\
& = \delta_{hk} \cosh{(\beta \Delta E_{\Lambda_h})}  \delta_{\alpha \beta} \left(-\frac{1}{2}\delta_{x\alpha}-\frac{1}{2}\delta_{y\alpha} -\delta_{\alpha z}\right) \,. \\
\end{eqnarray*} 
Finally, using also
\begin{equation}
\eqalign{
    -i\sum_{\gamma h'  } \sigma_{(\alpha h) \, (\gamma h')} 2ih_{(\gamma h') (\beta k)} &  =  -i\sum_{\gamma  } \sigma_{(\alpha h) \, (\gamma h')} \delta_{h h'} 2i\delta_{\gamma \beta} \delta_{\beta z}\frac{g}{2}\tilde{w}_{h' k}\\
    & = + g \sigma_{\alpha \beta (h)} \delta_{\beta z} \tilde{w}_{h k} \, ,\\
    }
\end{equation}
it is
\begin{eqnarray}
    [Q]_{(\alpha \, h), (\beta \, k)} & = -2\Omega \mathcal{E}_{x\alpha \beta} \delta_{hk} -2g\mathcal{E}_{z\alpha \beta}\sum_{h'} \tilde{w}_{hh'} m_{z h'}\delta_{h k} \\
    &+ g \sigma_{\alpha \beta (h)} \delta_{\beta z} \tilde{w}_{h k} \\ \nonumber 
    & + \delta_{hk} \cosh{(\beta \Delta E_{\Lambda_h})} \delta_{\alpha \beta} \left(-\frac{1}{2}\delta_{x\alpha}-\frac{1}{2}\delta_{y\alpha} -\delta_{\alpha z}\right) \, , \\ \nonumber
    [Q]^{T}_{(\alpha \, h), (\beta \, k)} & = +2\Omega \mathcal{E}_{x\alpha \beta} \delta_{h k} + 2g\mathcal{E}_{z\alpha \beta}\sum_{h'} \tilde{w}_{hh'} m_{z h'}  \delta_{h k} \\
    &- g \sigma_{\alpha \beta (h)} \delta_{\beta z} \tilde{w}_{h k} \\ \nonumber
    & + \delta_{hk} \cosh{(\beta \Delta E_{\Lambda_h})} \delta_{\alpha \beta} \left(-\frac{1}{2}\delta_{x\alpha}-\frac{1}{2}\delta_{y\alpha} -\delta_{\alpha z}\right) \, .\\ \nonumber
\end{eqnarray}
By exploiting the expression of the constant of the algebra $a_{\alpha \beta}^{\gamma} $, and having defined the local jump operators, one obtains
\begin{equation}
\eqalign{
[\Gamma^2 G^{\mathrm{sym}} ]_{(\alpha h)(\beta k)}& = \delta_{hk} \cosh{(\beta \Delta E_{\Lambda_h})} \left\lbrace\delta_{\alpha x}\delta_{x\beta} +\delta_{\alpha y}\delta_{y\beta} \right. \\
& \left. + 2\delta_{\alpha z}\delta_{z\beta}\right[ 1-\tanh{(\beta \Delta E_{\Lambda_h})} m_{z,h}\left]  \right\rbrace \\
& - \delta_{\alpha x}\delta_{z \beta} \sinh{(\beta \Delta E_{\Lambda_h})} m_{x,h} - \delta_{\alpha y}\delta_{z \beta} \sinh{(\beta \Delta E_{\Lambda_h})} m_{y,h}\\
&- \delta_{\alpha z}\delta_{x \beta} \sinh{(\beta \Delta E_{\Lambda_h})} m_{x,h} - \delta_{\alpha z}\delta_{y \beta} \sinh{(\beta \Delta E_{\Lambda_h})} m_{y,h} \, ,
}
\end{equation}
We can thus derive the term $\sigma S^{\mathrm{sym}} \sigma^{T}$. To begin with, we divide it according to the two followng terms,
\begin{equation}
\eqalign{
& [S^{\mathrm{sym}}_1]_{(\mu h)(\nu k)} = \frac{1}{2}\sum_\ell \sum_{k'} (r_{\ell \mu}^{hk'} r_{\ell \nu}^{k' k}+ r_{\ell \nu}^{hk'} r_{\ell \mu}^{k' k})\left(\Gamma'(\Delta^{k'}_{\ell} ) \right)^2   \omega_t\left(\frac{1}{N_s}\sum_{i \in \Lambda_{k'}}  j_\ell^{\dagger (i)} j_\ell^{ (i)}   \right) \\
& [S^{\mathrm{sym}}_2]_{(\mu h)(\nu k)} = \frac{1}{2} \sum_\ell [\Gamma(\Delta_{\ell}^{h}) \Gamma'(\Delta_{\ell}^{h})  r_{\ell \mu}^{hk} c_{\ell \nu} + r_{\ell \nu}^{h k} c_{\ell \mu} \Gamma(\Delta_{\ell}^{k}) \Gamma'(\Delta_{\ell}^{k})  ] 
}
\end{equation}
such that $S^{\mathrm{sym}} =  S^{\mathrm{sym}}_1 +  S^{\mathrm{sym}}_2 $. We start from the terms $S^{\mathrm{sym}}_1$, and we notice that the components of the vector $\vec{r}_\ell$ read here $ r_{\ell \alpha}^{h k} = \delta_{\alpha z} \tilde{w}_{kh}= \delta_{\alpha z} \tilde{w}_{hk}$, so that they are independent of $\ell$ and symmetric in the large-spin index. Moreover, the average operator reads
$$\frac{1}{N_s} \sum_{i \in \Lambda_{k}} j_{\ell}^{(i)\dagger} j_{\ell}^{(i)} = \frac{1- \ell m_{z,k}^{N_s}}{2} \, ,$$
and it has to be evaluated on a clustering state, as we are considering the thermodynamic limit. Furthermore, for the derivative of the rates, $\Gamma_{\ell}(\Delta E_{\Lambda_k})=  e^{\ell \frac{\beta}{2}\Delta E_{k}}/\sqrt{2}$, it is 
\begin{equation}
\left(\Gamma_{\ell}'(\Delta E_{\Lambda_{k}}) \right)^2 = \frac{\beta^2}{4} \Gamma_{\ell}^2( \Delta E_{\Lambda_{k}} ) \, .
\end{equation}
Thus we have
\begin{equation}
\eqalign{
& [\sigma S^{\mathrm{sym}}_1 \sigma^T]_{(\alpha h), (\beta k)}  = \frac{1}{2}  \sigma_{\alpha \mu (h)} \delta_{\mu, z}\sum_{\ell, k'} \left[\Gamma_{\ell}^{'\, 2}(\Delta E_{\Lambda_{k'}})(1- \ell m_{z,k}) \right] \tilde{w}^{hk'}\tilde{w}^{k' k} \delta_{\nu z}\sigma_{ \beta \nu (k)} \\
& =  \frac{\beta^2}{8}\sigma_{\alpha z (h)}\sum_{k'} [ \cosh{(\beta \Delta E_{\Lambda_{k'}})} - m_{z k'}  \sinh{(\beta \Delta E_{\Lambda_{k'}})}]\tilde{w}^{hk'}\tilde{w}^{k' k}\sigma_{\beta z (k)} \, .
}
\end{equation}
Finally, we focus on the term $S^{\mathrm{sym}}_2$. To this end, we notice that the elements $c_{\ell \eta} $ for $\eta =x, y, z$ reads here 
$$\eqalign{
{c}_{\ell , \eta} & = \sum_{\gamma, \gamma' } c_{\ell, \gamma}^{j *} c_{\ell, \gamma'}^{j } \frac{1}{2}(a_{\gamma \gamma'}^\eta + b_{\gamma \gamma'}^{\eta}) \\
& = \sum_{\gamma \gamma'} c_{\gamma \ell}^{j \, *}c_{\gamma' \ell}^{j \, *} [\delta_{\gamma \gamma'} + i\mathcal{E}_{\gamma \gamma' \eta}] = - \frac{1}{2}\ell \mathcal{E}_{xyz}\delta_{\eta z}
 \, , }$$
where we used $j_{\ell}^{(k)} = \sigma_{\pm}^{(k)}$, and thus $c_{\ell, \gamma}^j = \lbrace \frac{1}{2}, \pm \frac{i}{2}  ,0 ,0\rbrace$. Hence, we obtain 
\begin{equation}
    \eqalign{
    [\sigma S^{\mathrm{sym}}_2 \sigma^T]_{(\alpha h), (\beta k)} 
    & = \sigma_{\alpha \mu (h)} \delta_{\mu z} \sum_{\ell} \left\lbrace -\frac{\ell}{4} \left[ \Gamma_{\ell}(\Delta E_{\Lambda_{h}}) \Gamma_{\ell}'(\Delta  E_{\Lambda_{h}})+ \right. \right. \\
    & \left. \left. \Gamma_{\ell}(\Delta  E_{\Lambda_{k}}) \Gamma_{\ell}'(\Delta  E_{\Lambda_{k}} )\right] \right\rbrace \tilde{w}_{hk} \delta_{\nu z}\sigma_{\beta \nu (k)} \\
    & = -\frac{\beta}{8} \sigma_{\alpha z (h)} \tilde{w}_{hk} \sum_{\ell}[\Gamma_{\ell}^2(\Delta  E_{\Lambda_{h}})+\Gamma_{\ell}^2(\Delta E_{\Lambda_{k}} ) ] \sigma_{\beta z (k)} \\
    & = -\frac{\beta}{8} \sigma_{\alpha z (h)} \tilde{w}_{hk} [\cosh(\beta \Delta E_{\Lambda_{h}})+\cosh(\beta  \Delta E_{\Lambda_{k}}) ] \sigma_{\beta z (k)} \, .
    }
\end{equation}
Collecting all the above terms together, the EoMs for the covariance matrix reads
\begin{equation}\label{C:EoM_covariancematrix}
    \eqalign{
    \dot{\Sigma}_{(\alpha, h), (\beta,k)}^{\omega} & = \sum_{\mu, k'} \left\lbrace -2\Omega \mathcal{E}_{x\alpha \mu} \delta_{hk'} -2g\mathcal{E}_{z\alpha \mu}\sum_{h'} \tilde{w}_{hh'} m_{z h'} \delta_{h k'} \right. \\ 
    & + g \sigma_{\alpha \mu (h)} \delta_{\mu z} \tilde{w}_{h k'} \\  
    & \left. + \delta_{hk'} \cosh( \beta \Delta E_{\Lambda_{h}}) \delta_{\alpha \mu} \left(-\frac{1}{2}\delta_{x\mu}-\frac{1}{2}\delta_{y\mu} -\delta_{\mu z}\right) \right\rbrace\Sigma^{\omega}_{(\mu k') (\beta k)} \, + \\    
    &+\sum_{\mu k'} \Sigma^{\omega}_{(\alpha h) (\mu k')} \left\lbrace2\Omega \mathcal{E}_{x\mu \beta} \delta_{k' k} + 2g\mathcal{E}_{z\mu \beta}\sum_{h'} \tilde{w}_{k'h'} m_{z h'}(t)  \delta_{k' k} \right. \\
    &- g \sigma_{\mu \beta (k')} \delta_{\beta z} \tilde{w}_{k' k} \\
    & + \left. \delta_{k'k} \cosh( \beta \Delta E_{\Lambda_{k'}}) \delta_{\mu \beta} \left(-\frac{1}{2}\delta_{x\mu}-\frac{1}{2}\delta_{y\mu} -\delta_{\mu}\right) \right\rbrace+ \\
    & + \delta_{hk} \cosh( \beta \Delta E_{\Lambda_{h}}) \left\lbrace\delta_{\alpha x}\delta_{x\beta} +\delta_{\alpha y}\delta_{y\beta} \right. \\
    & \left. + 2\delta_{\alpha z}\delta_{z\beta}\right[ 1-\tanh( \beta \Delta E_{\Lambda_{h}} ) m_{z,h} \left]  \right\rbrace \\
    & - \delta_{\alpha x}\delta_{z \beta}\sinh\left(  \beta \Delta E_{\Lambda_{h}} \right)m_{x,h} - \delta_{\alpha y}\delta_{z \beta}\sinh\left( \beta \Delta E_{\Lambda_{h}} \right)m_{y,h}\\
    &- \delta_{\alpha z}\delta_{x \beta}\sinh\left(\beta \Delta E_{\Lambda_{h}} \right)m_{x,h} - \delta_{\alpha z}\delta_{y \beta}\sinh\left( \beta \Delta E_{\Lambda_{h}} \right)m_{y,h}\\
    & + \frac{\beta^2}{8}\sigma_{\alpha z (h)}\sum_{k'} [ \cosh{(\beta  \Delta E_{\Lambda_{k'}})} - m_{z k'}  \sinh{(\beta  \Delta E_{\Lambda_{k'}})}]\tilde{w}^{hk'}\tilde{w}^{k' k}\sigma_{\beta z (k)} \\ 
    &-\frac{\beta}{8} \sigma_{\alpha z (h)} \tilde{w}_{hk} [\cosh(\beta  \Delta E_{\Lambda_{h}})+\cosh(\beta  \Delta E_{\Lambda_{k}}) ] \sigma_{\beta z (k)} \, .
    }
\end{equation}
By setting $g=0$, and $m_{x,h}=0$ $\forall h$, from the above expression one can obtain the one in Eq.~\eref{EoM_covariancematrix}. In the following, we will however derive some considerations that go beyond the $g=0$ case. As observed in the main text, without the all-to-all interaction, i.e. at $g=0$, the matrix $Q$ does not couple the correlations amongst large-spins, that are given in terms of the off-diagonal elements $\Sigma^{\omega}_{(\alpha h),(\beta k)}$ for $h\neq k$; the matrix $\Gamma^2 G^{\mathrm{sym}}$ does not depend on $g$ and it is always diagonal with respect to the large-spin components. Finally, the term $\sigma S^{\mathrm{sym} } \sigma$ presents off-diagonal terms: at $g=0$, the mean-field stationary solutions features $m_{x,k} = 0$ $\forall k$, implying that the non-zero element are at position ${(xh),(xk)} $ $\forall k, h$. instead, for $g \neq 0$, there is more than one element for each block $(\alpha,h) (\beta,k)$, thus giving a corresponding finite determinant. In general, one could thus expect stronger quantum correlations in the presence of a direct, all-to-all, interaction. 

 \newpage
\bibliography{DM_bib2}        \bibliographystyle{unsrt}

\begin{thebibliography}{10}

\bibitem{BreuerP:2002}
H.~P. Breuer and F.~Petruccione.
\newblock {\em The theory of open quantum systems}.
\newblock Oxford University Press, Great Clarendon Street, 2002.

\bibitem{Diehl08}
S.~Diehl, A.~Micheli, A.~Kantian, B.~Kraus, H.~P. B{\"u}chler, and P.~Zoller.
\newblock Quantum states and phases in driven open quantum systems with cold
  atoms.
\newblock {\em Nat. Phys.}, 4:878--883, 2008.

\bibitem{Verstraete:NatPhys:2009}
F.~Verstraete, M.~M. Wolf, and J.~I. Cirac.
\newblock Quantum computation and quantum-state engineering driven by
  dissipation.
\newblock {\em Nature Physics}, 5(9):633--636, 2009.

\bibitem{Weimer:Nat:10}
H.~Weimer, M.~Müller, I.~Lesanovsky, P.~Zoller, and H.~P. Büchler.
\newblock A rydberg quantum simulator.
\newblock {\em Nature Physics}, 6(5):382–388, March 2010.

\bibitem{diehl2010}
S.~Diehl, A.~Tomadin, A.~Micheli, R.~Fazio, and P.~Zoller.
\newblock {Dynamical Phase Transitions and Instabilities in Open Atomic
  Many-Body Systems}.
\newblock {\em Phys. Rev. Lett.}, 105:015702, 2010.

\bibitem{dallatorre2010}
E.~G. Dalla~Torre, E.~Demler, T.~Giamarchi, and E.~Altman.
\newblock Quantum critical states and phase transitions in the presence of
  non-equilibrium noise.
\newblock {\em Nat. Phys.}, 6:806--810, 2010.

\bibitem{Schindler2013}
P.~Schindler, M.~M{\"u}ller, D.~Nigg, J.~T. Barreiro, E.~A. Martinez,
  M.~Hennrich, T.~Monz, S.~Diehl, P.~Zoller, and R.~Blatt.
\newblock Quantum simulation of dynamical maps with trapped ions.
\newblock {\em Nat. Phys.}, 9:361--367, 2013.

\bibitem{tauber2014}
U.~C. T\"auber and S.~Diehl.
\newblock {Perturbative Field-Theoretical Renormalization Group Approach to
  Driven-Dissipative Bose-Einstein Criticality}.
\newblock {\em Phys. Rev. X}, 4:021010, 2014.

\bibitem{marcuzzi2016}
M.~Marcuzzi, M.~Buchhold, S.~Diehl, and I.~Lesanovsky.
\newblock {Absorbing State Phase Transition with Competing Quantum and
  Classical Fluctuations}.
\newblock {\em Phys. Rev. Lett.}, 116:245701, 2016.

\bibitem{minganti2018}
F.~Minganti, A.~Biella, N.~Bartolo, and C.~Ciuti.
\newblock Spectral theory of liouvillians for dissipative phase transitions.
\newblock {\em Phys. Rev. A}, 98:042118, 2018.

\bibitem{iemini2018}
F.~Iemini, A.~Russomanno, J.~Keeling, M.~Schir\`o, M.~Dalmonte, and R.~Fazio.
\newblock {Boundary Time Crystals}.
\newblock {\em Phys. Rev. Lett.}, 121:035301, 2018.

\bibitem{carollo2019}
F.~Carollo, E.~Gillman, H.~Weimer, and I.~Lesanovsky.
\newblock {Critical Behavior of the Quantum Contact Process in One Dimension}.
\newblock {\em Phys. Rev. Lett.}, 123:100604, 2019.

\bibitem{chertkov2022}
E.~Chertkov, Z.~Cheng, A.~C. Potter, S.~Gopalakrishnan, T.~M. Gatterman, J.~A.
  Gerber, K.~Gilmore, D.~Gresh, A.~Hall, A.~Hankin, M.~Matheny, T.~Mengle,
  D.~Hayes, B.~Neyenhuis, R.~Stutz, and M.~Foss-Feig.
\newblock Characterizing a non-equilibrium phase transition on a quantum
  computer.
\newblock {\em Nature Physics}, 19(12):1799–1804, September 2023.

\bibitem{sieberer2013}
L.~M. Sieberer, S.~D. Huber, E.~Altman, and S.~Diehl.
\newblock {Dynamical Critical Phenomena in Driven-Dissipative Systems}.
\newblock {\em Phys. Rev. Lett.}, 110:195301, 2013.

\bibitem{helmrich2020}
S.~Helmrich, A.~Arias, G.~Lochead, T.~M. Wintermantel, M.~Buchhold, S.~Diehl,
  and S.~Whitlock.
\newblock Signatures of self-organized criticality in an ultracold atomic gas.
\newblock {\em Nature}, 577:481--486, 2020.

\bibitem{jo2021}
M.~Jo, J.~Lee, K.~Choi, and B.~Kahng.
\newblock Absorbing phase transition with a continuously varying exponent in a
  quantum contact process: A neural network approach.
\newblock {\em Phys. Rev. Research}, 3:013238, 2021.

\bibitem{jo2022}
M.~Jo and M.~Kim.
\newblock Simulating open quantum many-body systems using optimised circuits in
  digital quantum simulation.
\newblock {\em arXiv:2203.14295}, 2022.

\bibitem{KellyRM21}
S.~P. Kelly, A.~M. Rey, and J.~Marino.
\newblock Effect of active photons on dynamical frustration in cavity qed.
\newblock {\em Phys. Rev. Lett.}, 126:133603, Apr 2021.

\bibitem{weimer2021}
H.~Weimer, A.~Kshetrimayum, and R.~Or\'us.
\newblock Simulation methods for open quantum many-body systems.
\newblock {\em Rev. Mod. Phys.}, 93:015008, 2021.

\bibitem{LandfordR69}
O.~E.~Lanford III and D.~Ruelle.
\newblock {Observables at infinity and states with short range correlations in
  statistical mechanics}.
\newblock {\em Commun. Math. Phys.}, 13:194 -- 215, 1969.

\bibitem{BenattiEtAl18}
F.~Benatti, F.~Carollo, R.~Floreanini, and H.~Narnhofer.
\newblock Quantum spin chain dissipative mean-field dynamics.
\newblock {\em J. Phys. A: Math. Theor.}, 51:325001, 2018.

\bibitem{SewellS88}
G.~L. Sewell and Herbert Spohn.
\newblock {Quantum Theory of Collective Phenomena}.
\newblock {\em Physics Today}, 41(1):82--82, 01 1988.

\bibitem{Strocchi85}
F.~Strocchi.
\newblock {\em Elements of Quantum Mechanics of Infinite Systems}.
\newblock WORLD SCIENTIFIC, 1985.

\bibitem{hepp1973}
K.~Hepp and E.~H. Lieb.
\newblock {On the superradiant phase transition for molecules in a quantized
  radiation field: the Dicke maser model}.
\newblock {\em Ann. Phys.}, 76:360--404, 1973.

\bibitem{alicki1983}
R.~Alicki and J.~Messer.
\newblock Nonlinear quantum dynamical semigroups for many-body open systems.
\newblock {\em J. Stat. Phys.}, 32:299--312, 1983.

\bibitem{BenedikterPS15}
N.~Benedikter, M.~Porta, and B.~Schlein.
\newblock {\em {Effective Evolution Equations from Quantum Dynamics}}.
\newblock SpringerBriefs in Mathematical Physics. Springer International
  Publishing, 2015.

\bibitem{MerkliR18}
M.~Merkli and A.~Rafiyi.
\newblock Mean field dynamics of some open quantum systems.
\newblock {\em Proc. R. Soc. A: Math. Phys. Eng. Sci.}, 474:20170856, 2018.

\bibitem{Porta16}
M.~Porta.
\newblock Mean field dynamics of interacting fermionic systems.
\newblock {\em Mathematical Problems in Quantum Physics}, 717:13, 2016.

\bibitem{Pickl11}
P.~Pickl.
\newblock A simple derivation of mean field limits for quantum systems.
\newblock {\em Lett. Math. Phys.}, 97:151--164, 2011.

\bibitem{hioe1973}
F.~T. Hioe.
\newblock {Phase Transitions in Some Generalized Dicke Models of
  Superradiance}.
\newblock {\em Phys. Rev. A}, 8:1440--1445, 1973.

\bibitem{CarolloL:PRL:21}
F.~Carollo and I.~Lesanovsky.
\newblock {Exactness of Mean-Field Equations for Open Dicke Models with an
  Application to Pattern Retrieval Dynamics}.
\newblock {\em Phys. Rev. Lett.}, 126:230601, 2021.

\bibitem{BenattiCF_ADP_15}
F.~Benatti, F.~Carollo, and R.~Floreanini.
\newblock Dissipative dynamics of quantum fluctuations.
\newblock {\em Annalen der Physik}, 527(9-10):639--655, 2015.

\bibitem{Verbeure10}
A.~F. Verbeure.
\newblock {\em Many-Body Boson Systems: Half a Century Later}.
\newblock Springer London, 2011.

\bibitem{Goderis89}
D.~Goderis, A.~F. Verbeure, and P.~Vets.
\newblock Non-commutative central limits.
\newblock {\em Probability Theory and Related Fields}, 82:527--544, 1989.

\bibitem{Goderis90s}
D.~Goderis, A.~Verbeure, and P.~Vets.
\newblock Dynamics of fluctuations for quantum lattice systems.
\newblock {\em Communications in Mathematical Physics}, 128:533--549, 1990.

\bibitem{Benatti2016}
F.~Benatti, F.~Carollo, R.~Floreanini, and H.~Narnhofer.
\newblock Non-markovian mesoscopic dissipative dynamics of open quantum spin
  chains.
\newblock {\em Phys. Lett. A}, 380:381--389, 2016.

\bibitem{BenattiEtAl_JPA_17Qflu}
F.~Benatti, F.~Carollo, R.~Floreanini, and H.~Narnhofer.
\newblock Quantum fluctuations in mesoscopic systems.
\newblock {\em J. Phys. A: Math. Theor.}, 50(42):423001, 2017.

\bibitem{NarnhoferT_PRA_02}
H.~Narnhofer and W.~Thirring.
\newblock Entanglement of mesoscopic systems.
\newblock {\em Phys. Rev. A}, 66:052304, Nov 2002.

\bibitem{Lindblad76}
G.~Lindblad.
\newblock On the generators of quantum dynamical semigroups.
\newblock {\em Commun. Math. Phys.}, 48:119--130, 1976.

\bibitem{GoriniKS76}
V.~Gorini, A.~Kossakowski, and E.~C.~G. Sudarshan.
\newblock Completely positive dynamical semigroups of n‐level systems.
\newblock {\em Journal of Mathematical Physics}, 17(5):821--825, 1976.

\bibitem{Garrahan18}
J.~P. Garrahan.
\newblock Aspects of non-equilibrium in classical and quantum systems: Slow
  relaxation and glasses, dynamical large deviations, quantum non-ergodicity,
  and open quantum dynamics.
\newblock {\em Physica A: Statistical Mechanics and its Applications},
  504:130--154, 2018.
\newblock Lecture Notes of the 14th International Summer School on Fundamental
  Problems in Statistical Physics.

\bibitem{CarolloGK:JSP:21}
F.~Carollo, J.~P. Garrahan, and R.~L. Jack.
\newblock {Large Deviations at Level 2.5 for Markovian Open Quantum Systems:
  Quantum Jumps and Quantum State Diffusion}.
\newblock {\em J. Stat. Phys.}, 184, 2021.

\bibitem{FiorelliEtALC_NJP_23}
E.~Fiorelli, M.~Müller, I.~Lesanovsky, and F.~Carollo.
\newblock Mean-field dynamics of open quantum systems with collective
  operator-valued rates: validity and application.
\newblock {\em New Journal of Physics}, 25(8):083010, aug 2023.

\bibitem{Hopfield:1982}
J.~J. Hopfield.
\newblock Neural networks and physical systems with emergent collective
  computational abilities.
\newblock {\em PNAS}, 79:2554--2558, 1982.

\bibitem{LabayMoraZG23}
A.~Labay-Mora, R.~Zambrini, and G.~L. Giorgi.
\newblock Quantum associative memory with a single driven-dissipative nonlinear
  oscillator.
\newblock {\em Phys. Rev. Lett.}, 130:190602, May 2023.

\bibitem{BoedekerFM23}
L.~B\"odeker, E.~Fiorelli, and M.~M\"uller.
\newblock Optimal storage capacity of quantum hopfield neural networks.
\newblock {\em Phys. Rev. Res.}, 5:023074, May 2023.

\bibitem{Rotondo:JPA:2018}
P.~Rotondo, M.~Marcuzzi, J.~P. Garrahan, I.~Lesanovsky, and M.~M\"uller.
\newblock {Open quantum generalisation of Hopfield neural networks}.
\newblock {\em J. Phys. A: Math. Theor.}, 51:115301, 2018.

\bibitem{MarshEtAl:PhysRevX:21}
B.~P. Marsh, Y.~Guo, R.~M. Kroeze, S.~Gopalakrishnan, S.~Ganguli, J.~Keeling,
  and B.~L. Lev.
\newblock {Enhancing Associative Memory Recall and Storage Capacity Using
  Confocal Cavity QED}.
\newblock {\em Phys. Rev. X}, 11:021048, 2021.

\bibitem{GuoEtAl:PRL:2019}
Y.~Guo, R.~M. Kroeze, V.~D. Vaidya, J.~Keeling, and B.~L. Lev.
\newblock Sign-changing photon-mediated atom interactions in multimode cavity
  quantum electrodynamics.
\newblock {\em Phys. Rev. Lett.}, 122:193601, May 2019.

\bibitem{VaidyaEtAl:PRX:2018}
V.~D. Vaidya, Y.~Guo, R~M. Kroeze, Kyle~E. Ballantine, Alicia~J. Koll\'ar,
  J.~Keeling, and B.~L. Lev.
\newblock Tunable-range, photon-mediated atomic interactions in multimode
  cavity qed.
\newblock {\em Phys. Rev. X}, 8:011002, Jan 2018.

\bibitem{Rotondo:PRL:2015}
P.~Rotondo, M.~Cosentino~Lagomarsino, and G.~Viola.
\newblock Dicke simulators with emergent collective quantum computational
  abilities.
\newblock {\em Phys. Rev. Lett.}, 114:143601, Apr 2015.

\bibitem{LuoEtAl:arxiv:23}
C.~Luo, H.~Zhang, V.~P.~W. Koh, J.~D. Wilson, A.~Chu, M.~J. Holland, A.~M. Rey,
  and J.~K. Thompson.
\newblock Cavity-mediated collective momentum-exchange interactions, 2023.

\bibitem{SeetharamEtAl:PRR:2022}
K.~Seetharam, A.~Lerose, R.~Fazio, and J.~Marino.
\newblock Correlation engineering via nonlocal dissipation.
\newblock {\em Phys. Rev. Res.}, 4:013089, Feb 2022.

\bibitem{NorciaEtAl:Science:2018}
M.~A. Norcia, R.~J. Lewis-Swan, J.~R.~K. Cline, B.~Zhu, A.~M. Rey, and J.~K.
  Thompson.
\newblock Cavity-mediated collective spin-exchange interactions in a strontium
  superradiant laser.
\newblock {\em Science}, 361(6399):259--262, 2018.

\bibitem{TorgglerKR17}
V.~Torggler, S.~Kr\"amer, and H.~Ritsch.
\newblock Quantum annealing with ultracold atoms in a multimode optical
  resonator.
\newblock {\em Phys. Rev. A}, 95:032310, Mar 2017.

\bibitem{BratteliR82}
O.~Bratteli and D.~W. Robinson.
\newblock {\em Operator Algebras and Quantum Statistical Mechanics II.
  Equilibrium States Models in Quantum Statistical Mechanics.}
\newblock Springer Berlin, Heidelberg, 1981.

\bibitem{Strocchi05}
F.~Strocchi.
\newblock {\em Symmetry breaking}.
\newblock Springer Berlin, Heidelberg, 2021.

\bibitem{thirring2013quantum}
W.~Thirring.
\newblock {\em Quantum mathematical physics: atoms, molecules and large
  systems}.
\newblock Springer Science \& Business Media, 2013.

\bibitem{grimmett2020probability}
G.~Grimmett and D.~Stirzaker.
\newblock {\em Probability and random processes}.
\newblock Oxford university press, 2020.

\bibitem{fredrickson1984}
G.~H. Fredrickson and H.~C. Andersen.
\newblock {Kinetic Ising Model of the Glass Transition}.
\newblock {\em Phys. Rev. Lett.}, 53:1244--1247, 1984.

\bibitem{cancrini2008}
N.~Cancrini, F.~Martinelli, C.~Roberto, and C.~Toninelli.
\newblock Kinetically constrained spin models.
\newblock {\em Probab. Theory. Relat. Fields.}, 140:459--504, 2008.

\bibitem{garrahan2011}
J.~P Garrahan, P.~Sollich, and C.~Toninelli.
\newblock Kinetically constrained models.
\newblock {\em in "Dynamical heterogeneities in glasses, colloids, and granular
  media", Eds.: L. Berthier, G. Biroli, J.-P. Bouchaud, L. Cipelletti and W.
  van Saarloos (Oxford University Press, 2011)}, 150:111--137, 2011.

\bibitem{glauber1963}
R.~J. Glauber.
\newblock {Time‐Dependent Statistics of the Ising Model}.
\newblock {\em J. Math. Phys.}, 4:294--307, 1963.

\bibitem{walter2015}
J.-C. Walter and G.T. Barkema.
\newblock {An introduction to Monte Carlo methods}.
\newblock {\em Phys. A: Stat. Mech. Appl.}, 418:78--87, 2015.
\newblock Proceedings of the 13th International Summer School on Fundamental
  Problems in Statistical Physics.

\bibitem{BonebergLC22}
M.~Boneberg, I.~Lesanovsky, and F.~Carollo.
\newblock Quantum fluctuations and correlations in open quantum dicke models.
\newblock {\em Phys. Rev. A}, 106:012212, Jul 2022.

\bibitem{CarolloL_PRA_22}
F.~Carollo and I.~Lesanovsky.
\newblock Exact solution of a boundary time-crystal phase transition:
  Time-translation symmetry breaking and non-markovian dynamics of
  correlations.
\newblock {\em Phys. Rev. A}, 105:L040202, Apr 2022.

\bibitem{MattesLC_arxiv_23}
R.~Mattes, I.~Lesanovsky, and F.~Carollo.
\newblock Entangled time-crystal phase in an open quantum light-matter system,
  2023.

\bibitem{Amit_book}
D.~J. Amit.
\newblock {\em Modelling Brain Function: The World of Attractor Neural
  Networks}.
\newblock Cambridge University Press, USA, 1st edition, 1992.

\bibitem{AmitGS:1985a}
D.~J. Amit, H.~Gutfreund, and H.~Sompolinsky.
\newblock Spin-glass models of neural networks.
\newblock {\em Phys. Rev. A}, 32:1007--1018, 1985.

\bibitem{FiorelliLM22}
E.~Fiorelli, I.~Lesanovsky, and M.~M\"uller.
\newblock {Phase diagram of quantum generalized Potts-Hopfield neural
  networks}.
\newblock {\em New J. Phys.}, 24:033012, 2022.

\bibitem{KochP_JSP_89}
H.~{Koch} and J.~{Piasko}.
\newblock {Some rigorous results on the Hopfield neural network model}.
\newblock {\em Journal of Statistical Physics}, 55(5-6):903--928, June 1989.

\bibitem{GrensingK86}
D.~Grensing and R.~Kuhn.
\newblock Random-site spin-glass models.
\newblock {\em Journal of Physics A: Mathematical and General}, 19(18):L1153,
  dec 1986.

\bibitem{Gayrard92}
V.~{Gayrard}.
\newblock {Thermodynamic limit of the q-state Potts-Hopfield model with
  infinitely many patterns}.
\newblock {\em J. Stat. Phys.}, 68:977--1011, 1992.

\bibitem{AdessoDPRL10}
G.~Adesso and A.~Datta.
\newblock Quantum versus classical correlations in gaussian states.
\newblock {\em Phys. Rev. Lett.}, 105:030501, Jul 2010.

\bibitem{FiorelliEtAl20}
E.~Fiorelli, M.~Marcuzzi, P.~Rotondo, F.~Carollo, and I.~Lesanovsky.
\newblock {Signatures of Associative Memory Behavior in a Multimode Dicke
  Model}.
\newblock {\em Phys. Rev. Lett.}, 125:070604, 2020.

\bibitem{ZhangFG90}
W-M. Zhang, D.~H. Feng, and R.~Gilmore.
\newblock Coherent states: Theory and some applications.
\newblock {\em Rev. Mod. Phys.}, 62:867--927, Oct 1990.

\bibitem{MichoelN04central}
Tom Michoel and Bruno Nachtergaele.
\newblock Central limit theorems for the large-spin asymptotics of quantum
  spins.
\newblock {\em Probability theory and related fields}, 130(4):493--517, 2004.

\bibitem{Gronwall19}
T.~H. Gronwall.
\newblock Note on the derivatives with respect to a parameter of the solutions
  of a system of differential equations.
\newblock {\em Ann. Math.}, 20(4):292--296, 1919.

\bibitem{Fiorelli:PRA:2019}
E.~Fiorelli, P.~Rotondo, M.~Marcuzzi, J.~P. Garrahan, and I.~Lesanovsky.
\newblock {Quantum accelerated approach to the thermal state of classical
  all-to-all connected spin systems with applications to pattern retrieval in
  the Hopfield neural network}.
\newblock {\em Phys. Rev. A}, 99:032126, 2019.

\end{thebibliography}

\end{document}